\documentclass[aos]{imsart}
\RequirePackage{amsthm,amsmath,amsfonts,amssymb}
\RequirePackage[numbers]{natbib}
\usepackage{mathtools, algorithm, subcaption, algpseudocode, mleftright, arydshln, enumitem, mathrsfs, bm, caption}
\RequirePackage[colorlinks,citecolor=blue,urlcolor=blue]{hyperref}
\usepackage{xr}
\usepackage[title]{appendix}
\newcommand{\B}{\mathbf B}
\newcommand{\PP}{\mathbf P}
\newcommand{\E}{\mathrm E}
\newcommand{\X}{\mathbf X}
\newcommand{\x}{\mathbf x}
\newcommand{\Y}{\mathbf Y}
\newcommand{\y}{\mathbf y}
\newcommand{\K}{\mathbf K}

\newcommand{\T}{^\mathsf{T}}
\newcommand{\bmid}{\mathrel{\bigg|}}

\newcommand{\indep}{\;\, \rule[0em]{.03em}{.67em} \hspace{-.25em}
	\rule[0em]{.65em}{.03em} \hspace{-.25em}
	\rule[0em]{.03em}{.67em}\;\,}
	
\newlength\myindent
\startlocaldefs
\newtheorem{condition}{Condition}
\newtheorem{assump}{Assumption}
\newtheorem{assumpB}{Assumption}

\newtheorem{theorem}{Theorem}
\newtheorem{lemma}{Lemma}
\DeclareMathOperator*{\argmin}{arg\,min}
\DeclareMathOperator*{\argmax}{arg\,max}
\newtheorem{manualtheoreminner}{Condition}
\newenvironment{manualtheorem}[1]{%
  \renewcommand\themanualtheoreminner{#1}%
  \manualtheoreminner
}{\endmanualtheoreminner}

\theoremstyle{remark}

\newtheorem{innercustomthm}{Lemma}
\newenvironment{customthm}[1]
  {\renewcommand\theinnercustomthm{#1}\innercustomthm}
  {\endinnercustomthm}
\endlocaldefs
\allowdisplaybreaks

\begin{document}

\begin{frontmatter}
\title{Semiparametric Efficient Dimension Reduction in multivariate regression with an Inner Envelope}

\runtitle{Semiparametric Efficient Inner Envelope}

\begin{aug}
\author[A]{\fnms{Linquan} \snm{Ma}},
\author[A]{\fnms{Hyunseung} \snm{Kang}}
\and
\author[B]{\fnms{Lan} \snm{Liu}}
\address[A]{University of Wisconsin-Madison}
\address[B]{University of Minnesota at Twin Cities}
\end{aug}

\begin{abstract}
  Recently, Su and Cook \cite{su2012inner} proposed a dimension reduction technique called the inner envelope which can be substantially more efficient than the original envelope \cite{cook2010envelope} or existing dimension reduction techniques for multivariate regression. However, their technique relied on a linear model with normally distributed error, which may be violated in practice. In this work, we propose a semiparametric variant of the inner envelope that does not rely on the linear model nor the normality assumption. We show that our proposal leads to globally and locally efficient estimators of the inner envelope spaces. We also present a computationally tractable algorithm to estimate the inner envelope. Our simulations and real data analysis show that our method is both robust and efficient compared to existing dimension reduction methods in a diverse array of settings.
\end{abstract}


\end{frontmatter}

\section{Introduction}
\subsection{Background} \label{sec: background}

In 2010, Cook et al. \cite{cook2010envelope} proposed a new approach to dimension reduction in multivariate linear regression models called the envelope approach, with the  aim of more efficiently estimating the underlying regression coefficients. Briefly, for each individual $i=1,\ldots,n$, consider a multivariate linear regression model where we regress $r$-dimensional, multivariate responses $\Y_i \in \mathbb R^{r\times 1}$ onto $p$ regressors $\X_i \in \mathbb R^{p\times 1}$, i.e. 
\begin{equation}\label{eq: lm}
\Y_i = \bm\beta \X_i + \bm\varepsilon_i, \quad{} \bm \varepsilon_i \sim\mathcal N(\mathbf{0}, \bm\Sigma). 
\end{equation}
Here, the term $\bm \varepsilon_i \in \mathbb R^{r \times 1}$ is a random vector that is independent from $\X_i$ and follows a multivariate normal distribution with mean $\mathbf{0}$ and unknown covariance $\bm\Sigma\in\mathbb R^{r\times r}$. The parameter $\bm\beta\in\mathbb R^{r\times p}$ is the unknown matrix regression coefficients of interest. A natural estimator of $\bm \beta$ is the ordinary least square (OLS) estimator, which takes the form $\hat{\bm\beta}\T_{\rm OLS } =  (\X\T\X)^{-1}\X\T\Y$ where $\X = (\X_1,\ldots,\X_n)\T\in\mathbb R^{n\times p}$ and $\Y = (\Y_{1},\ldots,\Y_n)\T\in\mathbb R^{n\times r}$. However, when the number of responses $r$ is 
large, Cook and co-authors proposed to improve on the OLS estimator by considering an ``envelope'' subspace in the space of the responses $\mathbb R^r$. Specifically, if there is a linear combination of the responses that is an ancillary for $\bm\beta$, \cite{cook2010envelope}  proposed a new estimator of $\bm\beta$, called the envelope estimator, and showed that the new envelope estimator is at least as efficient as the OLS estimator; see Section \ref{sec: env} for details.

Since the introduction of the envelope estimator, there has been an explosion of work on using envelope-based dimension reduction methods for multivariate regression models \cite{cook2010envelope, su2011partial, su2012inner, cook2013scaled, cook2013envelopes, cook2015envelopes, cook2015simultaneous, cook2015foundations, li2017parsimonious,shi2020mixed,cook2021envelopes,rekabdarkolaee2020new,MaLiYa2021}. Our paper focuses on one popular extension and improvement of the original work, the inner envelope estimator of Su and Cook \cite{su2012inner}. Specifically, 
the inner envelope estimator supposes that there is an extra subspace that is dependent on some responses, but becomes independent after conditioning on the regressors; see Section \ref{sec: env} for the exact definitions. Critically, the inner envelope estimator of $\bm \beta$ can achieve efficiency gains even when the envelope estimator offers no gains.

Despite its superior performance, the inner envelope estimator relies on parametric assumptions (e.g. model \eqref{eq: lm}) to guarantee its efficiency gains over the envelope estimator or the OLS estimator. In general, most envelope-based approaches rely on a linear relationship between the responses and the regressors \cite{cook2010envelope,su2012inner,cook2013scaled,li2017parsimonious,shi2020mixed, cook2021envelopes} and/or the joint normality of the error terms \cite{cook2010envelope,cook2013scaled,li2017parsimonious,cook2021envelopes}, even though the envelope-based assumptions are not stated with respect to a particular parametric model; see, for example, our high-level descriptions of the envelope or the inner envelope above or the exact definitions in Section \ref{sec: env}. Often, parametric assumptions are imposed out of theoretical convenience, especially to leverage the useful properties of multivariate normal distributions concerning independence from its covariance matrix or the variational independence between the parameters for the mean and the covariance matrix.

More importantly, in practice, these parametric assumptions are often violated or infeasiable for certain data types. For example, if the responses are binary, the existing inner envelope estimator cannot be applied because, unlike the multivariate normal model in \eqref{eq: lm}, $\bm \Sigma$ and $\bm \beta$ are generally not variationally independent for binary responses; so far, there is no inner envelope estimator for logistic or multinomial multivariate regression. A similar problem occurs when some responses are continuous while others are binary. Or, if the responses form non-linear relationships with the regressors, say the regressors have higher-order polynomial terms, or the regression errors are heteroskedastic, the existing inner envelope estimator may be inconsistent and/or inefficient. In Section \ref{sec:problem}, we show the incosistency of the existing inner envelope estimator when model \eqref{eq: lm} is violated.

While some progress has been made to relax these assumptions in envelope-based methods \cite{cook2015foundations, zhang2018functional}, to the best of our knowledge, there is no work on relaxing the parametric assumptions underlying the inner envelope estimator. This paper aims to resolve this through a semiparametric generalization of the inner envelope.

\subsection{Our Contributions}
At a high level, our semiparametric generalization 
relies on efficiently estimating a more fundamental quantity discussed in the original inner envelope paper called the inner envelope space \cite{su2012inner}. In particular, in the original work, estimation of $\bm \beta$ was done in two steps where the first step involved estimating the inner envelope space and the second step involved a projection of the responses and the regressors onto the estimated inner envelope space; see Section \ref{sec: env} for details. Our main insight in the paper is to realize that (a) the inner envelope space in the original work can actually be defined without assuming a parametric model \eqref{eq: lm} or even a particular type of responses (i.e. continuous, binary, or even a mixture of them), and (b) the inner envelope space can be uniquely parametrized as finite dimensional, semi-orthogonal basis matrices. These insights serve as the basis to develop locally and globally semiparametric efficient estimators, including a locally robust estimator that remains consistent and asymptotically normal even if some of the underlying models are mis-specified.
 
In addition to being robust to parametric modelling assumptions, our semiparametric generalization has two additional byproducts. First, we show a computationally simple procedure to evaluate all regular and asymptotically linear (RAL) estimators based on the generalized method of moments (GMMs). Our procedure is dramatically simpler than the original approach in \cite{su2012inner} based on solving a non-convex objective function over a Grassmannian, a non-convex set. Second, since our semiparametric generalization does not rely on a parametric model between the responses and the regressors, we show how to use the efficiently estimated inner envelope space to improve predictions from supervised machine learning models, such as XGBoost \cite{chen2015xgboost} and random forest \cite{breiman2001random}. In particular, we show that directly using the original responses $\Y_i$ inside a supervised machine learning method may lead to poor, out-of-sample, predictive performance compared to using the dimension-reduced responses generated from our approach. 

Our work is related to some recent works on semiparametric relaxation of sufficient dimension reduction (SDR) in traditional, univariate (i.e. single response) linear regression models
\cite{ma2012semiparametric, ma2013efficient}. Typically, SDR methods such as sliced inverse regression (SIR) \cite{li1991sliced}, sliced average variance estimation (SAVE) \cite{cook1991sliced} and directional regression (DR) \cite{li2007directional} require the univariate response to be linear and have homoskedastic variance. \cite{ma2012semiparametric} relaxed these parametric assumptions by directly estimating the dimension-reduced subspace of the regressors. In contrast, our work deals with the case when both the regressors and the responses are multivariate. In particular, our semiparametric relaxation of the inner envelope  requires estimating two subspaces, one of which summarizes the relationship between $\X$ and $\Y$ and another nested subspace that contains purely relevant information; see Section \ref{Sec: semi} for the exact definitions. The former subspace, which we denote as  $\mathcal{S}_1 \cup \mathcal{S}_2$ below, can be thought of as the generalization of the dimension-reduced subspace of regressors estimated in \cite{ma2012semiparametric} when there are multivariate responses. But, the latter subspace, denote as $\mathcal{S}_3$ below, is new and arises because of the multivariate responses with a dimension-reducing envelope structure.

The outline of the paper is as follows.  Section \ref{Sec: pre} reviews the original envelope \cite{cook2010envelope} and the inner envelope \cite{su2012inner}. The section also illustrates the problem of the existing inner envelope when the parametric modeling assumption is violated. Section \ref{Sec: semi} formalizes our semiparametric generalization of the inner envelope. Sections \ref{Sec: effi} derive the semiparametric efficiency bound of the inner envelope space and show estimators that can achieve this bound. Section \ref{sec: comp} provides implementation details for our estimators. Sections \ref{Sec: simu} and \ref{Sec: real} numerically demonstrate the efficiency gains of the proposed estimators in simulation and a study on the relationship between glycemic control and various risk factors for cardiovascular disease. Section 9 contains the proofs of the key results in the paper.

\section{Preliminaries}\label{Sec: pre}
\subsection{Notation}
For any matrix $\mathbf A\in \mathbb R^{r\times p}$, let $\mathcal A = \text{span}(\mathbf A)$ denote the subspace of $\mathbb R^r$ spanned by the columns of $\mathbf A$ and let $\mathcal A^\perp$ denote the orthogonal complement of this subspace. Let $\mathbf P_{\mathcal A}$ and $\mathbf Q_{\mathcal A} = \mathbf I - \mathbf P_{\mathcal A}$ denote the orthogonal projection matrix that projects a vector $\mathbf v\in\mathbb R^{r}$ onto $\mathcal A$ and $\mathcal A^\perp$, respectively. Let dim($\mathcal{A}$) denote the dimension of the subspace $\mathcal{A}$. Let $\text{vec}(\mathbf A)$ denote the vector formed by stacking columns of the matrix $\mathbf A$ and $\|\mathbf A\|_F = \sqrt{\text{trace}(\mathbf A\T\mathbf A)}$ denote its Frobenius norm. Finally, we say a square matrix $\mathbf A \in \mathbb R^{r \times r}$ is an orthogonal matrix if $\mathbf A\T\mathbf A = \mathbf I_r$. We say a non-square matrix $\mathbf A \in \mathbb R^{r \times p}$ is a semi-orthogonal matrix if $\mathbf A\T \mathbf A = \mathbf I_{p}$. 

For two matrices $\mathbf A\in\mathbb R^{r\times p}$ and $\mathbf B\in\mathbb R^{s \times q}$, let $\mathbf A\otimes\mathbf B\in\mathbb R^{rs\times pq}$ denote their Kronecker product. Relatedly, for any vector $\mathbf{v} \in \mathbb R^r$, let $\mathbf v^{\otimes 2} = \mathbf v \mathbf v\T$. For any two subspaces $\mathcal{S}_1$ and $\mathcal{S}_2$ of $\mathbb R^r$, let $\mathcal S_1\oplus\mathcal S_2 =  \{\x_1+\x_2:\x_1\in\mathcal S_1,\x_2\in\mathcal S_2\}$ denote their direct sum. For a function $\mathbf f(\mathbf x,\y):\mathbb R^{p_1+p_2}\rightarrow\mathbb R^{q}$, let $\partial\mathbf f/\partial\x\T\in\mathbb R^{p_1\times q}$ denote the derivative of $\mathbf f = (f_1,\ldots,f_q)$ with respect of $\x$ where the $(i,j)$-th entry is denoted as $\partial f_i/\partial x_j$. 

Throughout the paper, we use capital letters $\X$, $\Y$ to denote random variables, and lower case letters $\x$, $\y$ to denote observed values. We use $\X\sim\Y$ to denote that the two random variables $\X$ and $\Y$ have the same distributions; also, with a slight abuse of notation, we use the same notation $\sim$ to denote that a random variable follows a known family of distributions, say a multivariate normal distribution. We also use $\X\indep\Y$ to denote that $\X$ and $\Y$ are independent  and $\X \indep \Y \mid \mathbf{Z}$ to denote that $\X$ and $\Y$ are independent conditional on the random variable $\mathbf{Z}$.

\subsection{Review: Parametric Envelope and Inner Envelope Estimators} \label{sec: env}
We review the parametric envelope estimator of \cite{cook2010envelope} and the parametric inner envelope estimator of \cite{su2012inner}. 
The review is not aimed to be comprehensive; rather, the review is geared towards introducing some important concepts, notably the subspaces $\mathcal{S}_1,\mathcal{S}_2, \mathcal{S}_3$, that we will use in our semiparametric generalization. Also,  without loss of generality, we  will assume both $\X$ and $\Y$ are mean $\bm0$ random vectors.  

For the parametric envelope estimator, consider a subspace $\mathcal S\subseteq \mathbb R^r$ that satisfies the following conditions in model \eqref{eq: lm}:
	\begin{condition}\label{assum: env1}
		$\mathbf Q_{\mathcal S}\Y_i\mid \X_i\sim \mathbf Q_{\mathcal S}\Y_i$, 
	\end{condition}
	\begin{condition}\label{assum: env2}
		$\mathbf Q_{\mathcal S}\Y_i\indep \mathbf P_{\mathcal S}\Y_i\mid\X_i$.
	\end{condition} 
\noindent The smallest subspace $\mathcal{S}$ that satisfies the above conditions is called the envelope subspace. 
Broadly speaking, Conditions \ref{assum: env1} and \ref{assum: env2} state that $\mathbf Q_{\mathcal S}\Y_i$ is irrelevant in estimating $\bm\beta$. Consequently, we can regress $\mathbf P_{\mathcal S}\Y_i$ on $\X_i$ to obtain a more efficient, envelope estimator of $\bm\beta$, especially compared to the naive OLS estimator that regresses the entire responses $\Y_i$ on $\X_i$; see \cite{cook2010envelope} for the exact arguments. 

	
However, if the subspace satisfying Conditions \ref{assum: env1} and \ref{assum: env2} happens to be the entire space of the multivariate responses (i.e. $\mathbb R^r$), the envelope estimator has no efficiency gains over the OLS estimator. In such settings, \cite{su2012inner} proposed the inner envelope estimator of $\bm \beta$, which is defined as follows. 
Suppose the responses  can be decomposed as the sum of projections onto three orthogonal subspaces $\mathcal S_1$, $\mathcal S_2$, $\mathcal S_3 \subseteq \mathbb R^r,\mathcal S_1 \oplus \mathcal S_2 \oplus \mathcal S_3 = \mathbb R^r$, with dimensions dim$(\mathcal S_1) = u$, dim$(\mathcal S_2) = p-u$, and dim$(\mathcal S_3) = r - p$. The three subspaces $\mathcal{S}_{1}, \mathcal{S}_2, \mathcal{S}_3$ satisfy the following conditions:
  \begin{condition}\label{assum: inner1}
    $\mathbf P_{\mathcal S_3}\Y_i\mid\X_i\sim\mathbf P_{\mathcal S_3}\Y_i,$
  \end{condition}
  \begin{condition}\label{assum: inner2}
	$\mathbf P_{\mathcal S_1} \Y_i\indep \mathbf Q_{\mathcal S_1}\Y_i\mid\X_i.$
  \end{condition}
  \noindent
The largest subspace $\mathcal S_1$ satisfying Conditions \ref{assum: inner1} and \ref{assum: inner2} is called the inner envelope subspace.
Unlike Conditions \ref{assum: env1} and \ref{assum: env2}, the inner envelope estimator presents a new subspace $\mathcal S_2$ where $\mathbf P_{\mathcal S_2}\Y_i$ is correlated with both $\mathbf P_{\mathcal S_1}\Y_i$ and $\mathbf P_{\mathcal S_3}\Y_i$, but  $\mathbf P_{\mathcal S_1}\Y_i$ is independent with $\mathbf P_{\mathcal S_2}\Y_i$ and $\mathbf P_{\mathcal S_3}\Y_i$ given $\X$. If $\mathcal S_2$ is the null space, Condition \ref{assum: inner2} reduces to Condition \ref{assum: env2} and the inner envelope subspace coincides with the envelope subspace, i.e. $\mathcal{S}_1 = \mathcal{S}$. Critically, Conditions \ref{assum: inner1} and \ref{assum: inner2} do not necessarily rely on a parametric model between $\Y_i$ and $\X_i$, a fact that we exploit in our semiparametric generalization of the inner envelope below.

Under Conditions \ref{assum: inner1}, \ref{assum: inner2}, and model \eqref{eq: lm}, \cite{su2012inner} proposed a two-step, inner envelope estimator of $\bm\beta$ where the first step involves estimating the aforementioned subspaces $\mathcal{S}_1, \mathcal{S}_2, \mathcal{S}_3$ and the second step involves estimating the relationship between the predictors and the projected response.  Specifically, in the first step, let $\bm\Gamma\in\mathbb R^{r\times u}$ and $\bm\Gamma_0\in\mathbb R^{r\times (r-u)}$ denote the semi-orthogonal basis matrices for subspaces $\mathcal S_1$ and $\mathcal S_1^\perp$, respectively (i.e. $\text{span}(\bm\Gamma) = \mathcal S_1$ and $\text{span}(\bm\Gamma_0) = \mathcal S_2\oplus\mathcal S_3$). Also, let $\B\in\mathbb R^{(r-u)\times (p-u)}$ be a semi-orthogonal matrix such that $\bm\Gamma_0\mathbf B$ is a semi-orthogonal basis matrix for $\mathcal S_2$ (i.e. $\text{span}(\bm{\Gamma_0} \B) = \mathcal S_2$). Intuitively, the matrix $\B$ divides the subspace spanned by $\bm \Gamma_0$ (i.e. $\mathcal{S}_1^\perp$) into the subspace specific to $\mathcal{S}_2$ and if $\B_0$ is the orthogonal complement of $\B$, $\mathcal{S}_3 $ can be rewritten as $\mathcal{S}_3 = \text{span}(\bm\Gamma_0\mathbf B_0)$. Based on these orthogonal relationships, estimating $\bm\Gamma$ and $\B$ is sufficient to estimate all the subspaces $\mathcal{S}_1,\mathcal{S}_2$, and $\mathcal{S}_3$. To estimate $\bm \Gamma$, \cite{su2012inner} proposed the following estimator: 
\begin{align}\label{eq: grassmanian2}
&\hat{\bm\Gamma} = \argmin_{\mathbf G}
\{\log|\mathbf G\T\mathbf S_{\text{res}}\mathbf G| + \log|\mathbf G\T\mathbf S_{\text{res}}^{-1}\mathbf G| + \sum_{i=p-u+1}^{r-u}\log (1+\lambda_{(i)})\}\\
	& \text{such that}\hspace{1mm}  \mathbf G\T\mathbf G = \mathbf I_u. \nonumber 
  \end{align}
Here, $\mathbf G_0$ is the orthogonal complement of $\mathbf G$ and $\lambda_{(i)}$, $i=1,\ldots,n$ are the ordered descending eigenvalues of $(\mathbf G_0\T\mathbf S_{\text{res}}\mathbf G_0)^{-1/2}(\mathbf G_0\T\mathbf S_{\text{fit}}\mathbf G_0)(\mathbf G_0\T\mathbf S_{\text{res}}\mathbf G_0)^{-1/2}$ with $\mathbf S_{\text{fit}}$ and $\mathbf S_{\text{res}}$ being the sample covariances of the fitted and residual vectors, respectively, from the OLS regression of $\Y$ on $\X$. Using the orthogonal relationship, we can also arrive at the estimator $\hat{\bm \Gamma}_0$, which is the orthogonal complement of the estimated $\hat{\bm \Gamma}$. The estimator of $\B$, which we denote as $\hat{\B}$, is similar and for brevity, the details are relegated to  Section A.1 of the Supplements. 
  
Once the basis matrices are estimated, \cite{su2012inner} proposed to estimate $\bm \beta$ as follows: 
$$\hat{\bm\beta}_{\rm IE} = \hat{\bm\Gamma}\hat{\bm\zeta}_1 + \hat{\bm\Gamma}_0\hat{\mathbf B}\hat{\bm\zeta}_2. $$ 
Here, $\hat{\bm\zeta_1},\hat{ \bm\zeta_2}$ are the estimated coordinates of the  projections onto the estimated subspaces $\mathcal{S}_1$ and $\mathcal{S}_2$. Under the linear model with normal errors in \eqref{eq: lm}, \cite{su2012inner} showed that the estimator $\hat{\bm\beta}_{\rm IE}$ is not only consistent and asymptotically normal, but also has equal or smaller asymptotic variance than $\hat{\bm \beta}_{\rm OLS}$.

 Finally, we make two additional remarks.  First, the original inner envelope implicitly assumed the dimension of the subspace $\mathcal S_3$ is fixed at $r-p$ and  thus, $\bm\beta$ is required to have full column rank. In our semiparametric generalization below, we relax this assumption and allow $\mathcal S_3$ to have a dimension other than $r -p$. 
  This relaxation introduces one minor condition on the maximum size of the dimensions to ensure identifability and is discussed in Section \ref{sec: targetpar}. Second, similar to other works on semiparametric generalizations of dimension-reduction approaches \cite{ma2012semiparametric,ma2013efficient,ma2014estimation}, we assume $u$ and $d$ are known for our theoretical discussions. But, Section \ref{Sec: dim} discusses how to choose this dimension from data. 

\subsection{Problem: Inconsistency of the Parametric Inner Envelope When Model \eqref{eq: lm} is Wrong} \label{sec:problem}

In this section, we explore the consequences of using the original inner envelope approach based on model \eqref{eq: lm} when that model no longer holds. This exercise intends to mimic a practitioner who may naively conduct dimension reduction using the original inner envelope, but the underlying data generating model may deviate from model \eqref{eq: lm}. Our goal in this section is to not only show that this practice can lead to misleading results, but also to provide a concrete rationale for our semiparametric dimension reduction approach we propose in Section \ref{Sec: semi}.

To start, consider the case where model \eqref{eq: lm} is correct. Specifically, let there be two regressors generated from independent standard normals and three response variables.  The responses are generated from the model $\Y = (X_1+\varepsilon_1, X_2+\varepsilon_2,2X_2+\varepsilon_3)\T$ where $\varepsilon_1, \varepsilon_2,\varepsilon_3$ are independent, standard normals. Some simple algebra reveals that the subspaces $\mathcal S_1, \mathcal S_2$ and $\mathcal S_3$ satisfying the inner envelope Conditions \ref{assum: inner1} and \ref{assum: inner2} are spanned by $(1,0,0)\T$, $(0,1,2)\T$ and $(0,2,-1)\T$, respectively; note that the envelope subspace $\mathcal{S}$ is the entire space $\mathbb R^3$ and thus, this is an example where the envelope offers no gain compared to the inner envelope. Also, if $\mathbf{S}_{\rm res}$ and $\mathbf{S}_{\rm fit}$ in \eqref{eq: grassmanian2} are replaced with the true population covariances, the resulting estimator $\hat{\bm{\Gamma}}$ is a consistent estimator of the basis matrix that spans $\mathcal{S}_1$.


Now, suppose model \eqref{eq: lm} is incorrect and the true model has non-linear regressors, say $\Y =(X_1^2+\varepsilon_1, X_2+\varepsilon_2,2X_2+\varepsilon_3)\T$. The true subspaces $\mathcal S_1$, $\mathcal S_2$ and $\mathcal S_3$ that satisfy Conditions \ref{assum: inner1} and \ref{assum: inner2} are the same as those from the previous paragraph. 
However, the estimated basis matrix from the population version of \eqref{eq: grassmanian2} 
is not $(1,0,0)\T$; that is, there is another estimate $\hat{\bm\Gamma}$ that is not equal to $(1,0,0)\T$ and that achieves the global minimum in \eqref{eq: grassmanian2}. Thus, 
$\hat{\bm\Gamma}$ is an inconsistent estimator of the true inner envelope subspace. Similarly, if the true model for $\Y$ is $\Y = (X_1^2+\varepsilon_1, X_2+\varepsilon_2,X_2\varepsilon_3)\T$ where the third response has heteroskedastic variance, the subspace $\mathcal S_1$ is the same as the previous examples, but its basis matrix is still not the global optimum of \eqref{eq: grassmanian2}.

The main, high-level, reason for the inconsistency lies in the original estimator's strong reliance on parametric modeling assumptions. Specifically, equation \eqref{eq: grassmanian2} only utilizes second order moments between $\X$ and $\Y$ because under the linear model with normal errors in \eqref{eq: lm}, second order moments are ``sufficient'' to identify the underlying subspaces. However, when $Y_1$ contains non-linear terms or when $Y_3$ has heteroskedastic variance, the responses are no longer normal and thus, higher-order moments are necessary to identify the underlying dimension-reduced subspaces that satisfy the inner envelope conditions.
Our proposed semiparametric generalization below aims to remove these dependencies on a particular parametric model and arrive at a more robust and efficient approach to inner envelope.



\section{A Semiparametric Approach to Inner Envelope}\label{Sec: semi}
	
\subsection{Target parameters and assumptions} \label{sec: targetpar}
At a high level, our semiparametric approach focuses on a ``model-free'' estimation the fundamental quantities underlying the original inner envelope estimator, the three orthogonal spaces introduced in the prior section, $\mathcal S_1$, $\mathcal S_2$, and $\mathcal S_3$ with dimensions $u$, $d$ and $r-u-d$. However, a general challenge in estimating subspaces is that the subspace requires an identifiable parametrization. For example, a basis matrix $\mathbf A\in\mathbb R^{k\times \tau}$ for a $\tau$-dimensional subspace $\mathcal A$ in $\mathbb R^k$ is not unique even if the subspace $\mathcal A$ is unique since for any full rank matrix $\mathbf U\in\mathbb R^{\tau\times\tau}$, $\mathbf A\mathbf U$ also spans the same subspace $\mathcal A$. Additionally, the elements of a basis matrix are often variationally dependent, which can complicate the estimation procedure. 

Our solution to remedy these issues are as follows. First, we propose a new, representation of the inner envelope space with the following lemma; see \cite{ma2013efficient} for a related result.
\begin{lemma}\label{lemma: repar} 
Consider a $\tau$-dimensional subspace $\mathcal{A}$ in $\mathbb R^k$ where $\tau \in \{1,\ldots,k\}$.  Then, there exists an one-to-one mapping $\psi_1$ such that
$\mathbf a = \psi_1(\mathcal A) \in\mathbb R^{\tau (k-\tau)}$ is a vector with all variational independent elements. Also, there exists an one-to-one mapping $\psi_2$ such that $\mathbf A = \psi_2(\mathcal A) \in \mathbb{R}^{k\times \tau}$ is a semi-orthogonal matrix. 
The construction of $\psi_1$ and $\psi_2$ are given in Section \ref{sec: proof}.
\end{lemma}

With Lemma \ref{lemma: repar}, we can uniquely parametrize the inner envelope spaces with semi-orthogonal matrices $\bm\Gamma = \psi_2(\mathcal S_1)$ and $\bm\Gamma_0\B = \psi_2(\mathcal S_2)$; note that since $\mathcal{S}_3$ is orthogonal to $\mathcal{S}_1$ and $\mathcal{S}_2$, we can uniquely represent $\mathcal{S}_3$ as $\bm \Gamma_0 \B_0$. Relatedly, given the orthogonality conditions, we can focus on estimating semi-orthogonal matrices $\bm\Gamma$ and $\B$. 

Second, let $\bm \Gamma$ and $\mathbf B$ be the unique parametrizations of the subspaces from Lemma \ref{lemma: repar}. Using these parameters, we can restate Conditions \ref{assum: inner1} and \ref{assum: inner2} as 
\begin{manualtheorem}{$3^*$}\label{eq: semi_inner1}
	$\mathbf B_0\T\bm\Gamma_0\T\Y_i\indep \X_i$,
\end{manualtheorem} 
\begin{manualtheorem}{$4^*$}\label{eq: semi_inner2}
	$\bm\Gamma\T\Y_i\indep \bm\Gamma_0\T\Y_i\mid \X_i$
\end{manualtheorem}
\noindent We additionally require $\mathcal S_3$ to be the largest subspace that satisfies Condition \ref{eq: semi_inner1} and $\mathcal S_1$ to be the largest space that satisfies Condition \ref{eq: semi_inner2} given $\mathcal S_3$. As mentioned in Section \ref{sec: env}, this additional requirement generalizes the implicitly assumed constraint in the original inner envelope where the subspace $\mathcal{S}_3$ has a pre-determined, fixed dimension of $r-p$.

Third, the elements in $\bm{\Gamma}$ and $\B$ are not variationally independent, which makes them difficult to work with for theory and  for some well-known optimization algorithms. Therefore, we collapse the matrices $\bm{\Gamma}$ and $\B$ into variationally independent vectors by using Lemma \ref{lemma: repar} again.
Specifically, we can reparametrized the two matrices as $\bm\gamma = \psi_1(\psi_2^{-1}(\bm{\Gamma}))$, $\mathbf b = \psi_1(\psi_2^{-1}(\mathbf B))$ and $\bm\theta = (\bm\gamma\T,\mathbf b\T)\T\in\mathbb R^{q\times 1}$ with $q = (r-u)\times u + (r-u-d)\times d$.  Also, $\bm \theta$ uniquely parametrizes the subspaces span($\bm\Gamma$) and span($\B$), which in turn uniquely parametrizes the subspaces $\mathcal S_1$, $\mathcal S_2$, and $\mathcal S_3$. Given the uniqueness of the vector representation $\bm \theta$ and the matrix representation $\bm \Gamma, \B$, we use the two representations interchangeability throughout the text.

\subsection{Generalized method of moments estimators}
In this section, we propose a simple estimator of the subspaces parametrized by $\bm \theta$ based on the generalized method of moments (GMM).
To do this, consider the function $\mathbf f(\Y_i,\X_i ; \bm \theta)$ that satisfies the following condition:
\begin{align}\label{eq: GMM1}
	&\E\{ \mathbf f(\Y_i, \X_i;\bm\theta) \} = \bm0, \ \text{ where} \\
	& \mathbf f(\Y_i, \X_i;\bm\theta) = \{\mathbf g(\bm\B_0\T\bm\Gamma_0\T\Y_i) - \E \{\mathbf g(\bm\B_0\T\bm\Gamma_0\T\Y_i)\}\}\{\mathbf h(\bm\Gamma\T\Y_i,\X_i) - \E\{ \mathbf h(\bm\Gamma\T\Y_i,\X_i)\}\}\nonumber.
\end{align}
Here, $\mathbf g:\mathbb R^{r-u-d}\rightarrow \mathbb R^{d_g\times 1}$ and $\mathbf h:\mathbb R^{u+p}\rightarrow\mathbb R^{1\times d_h}$, $d_g d_h \geq q$, are continuously differentiable functions where the covariance matrix $\text{Var}\{ \mathbf f(\Y_i, \X_i;\bm\theta) \}$ exists. 
There are many functions $\mathbf g$ and $\mathbf h$ that satisfy \eqref{eq: GMM1}. For example, if $\mathbf g$, and $\mathbf h$ are identity functions, i.e. $\mathbf g(\bm\B_0\T\bm\Gamma_0\T\Y_i) = \bm\B_0\T\bm\Gamma_0\T\Y_i$ and $\mathbf h(\bm\Gamma\T\Y_i,\X_i) = \{(\bm\Gamma\T\Y_i)\T , \X_i\T\}$, then $\E(\mathbf f) = \bm0$ because
\begin{equation*}
	\begin{aligned}
		 &\E[\{\bm\B_0\T\bm\Gamma_0\T\Y_i - \E(\bm\B_0\T\bm\Gamma_0\T\Y_i)\}\{\bm\Gamma\T\Y_i - \E(\bm\Gamma\T\Y_i)\}\T] \\ =& \text{Cov}(\bm\B_0\T\bm\Gamma_0\T\Y_i,\bm\Gamma\T\Y_i) =\bm 0,\\
		 &\E[\{\bm\B_0\T\bm\Gamma_0\T\Y_i - \E (\bm\B_0\T\bm\Gamma_0\T\Y_i)\}\{\X_i - \E (\X_i)\}\T ]\\ =& \text{Cov}(\bm\B_0\T\bm\Gamma_0\T\Y_i, \X_i) = \bm 0,
	\end{aligned}
\end{equation*}
where the latter equalities use the fact that Conditions \ref{eq: semi_inner1} and \ref{eq: semi_inner2} imply $\B_0\T\bm\Gamma_0\T\Y_i\indep (\bm\Gamma\T\Y_i,\X_i)$. 

Theorem \ref{thm: GMM} shows that solving the empirical counterparts to \eqref{eq: GMM1} will lead to $\sqrt n$-consistent and asymptotically normal estimators of $\bm\theta$.
\begin{theorem}\label{thm: GMM}
Suppose $\mathbf f(\Y_i,\X_i;\bm\theta)$ is a continuously differentiable function that satisfy equation \eqref{eq: GMM1}. Consider the following GMM estimator:
	\[
	\hat{\bm\theta} = \argmin_{\bm\theta}\|\sum_{i=1}^{N}\mathbf f(\Y_i, \X_i;\bm\theta)\|_2^2.
	\]
	Under the regularity conditions (S1)--(S6) in Section A of the Supplements, we have
	
		$$\sqrt n (\hat{\bm\theta} - \bm\theta)\xrightarrow{d} \mathcal N\big\{\bm 0, (\mathbf C_1\T\mathbf C_1)^{-1}\mathbf C_1\T\mathbf D_1\mathbf C_1(\mathbf C_1\T\mathbf C_1)^{-1}\big\},$$
	where $\mathbf C_1 = \E \{\partial\rm{vec}(\mathbf f)/\partial\bm\theta\T\}$ and $\mathbf D_1 = {\rm Var}\{ \mathbf f(\Y_i, \X_i;\bm\theta) \}$.
\end{theorem}
 The regularity conditions (S1)-(S6) concern bounded moments, differentiability of $\mathbf f$, and identifiability of the population GMM equation. These are standard conditions for GMM estimators to achieve consistency and asymptotic normality; see Chapter 2.2.3 in \cite{newey1994large} for a textbook exposition. 
 
Compared to existing methods to estimate the inner envelope (or the envelope) \cite{cook2010envelope,cook2016note,cook2015foundations,cook2015simultaneous}, the GMM approach does not explicitly rely a parametric model between $\mathbf Y$ and $\mathbf X$ nor a parametric, distributional assumption on $\bm\varepsilon$ to achieve consistency and asymptotic normality. Also, solving the estimating equations in Theorem \ref{thm: GMM} is far simpler than solving a non-convex problem, say in \eqref{eq: grassmanian2}, and can be done with existing statistical packages for GMMs \cite{chausse2010computing}.

One limitation of Theorem \ref{thm: GMM} is that it does not specify which functions $\mathbf g$ and $\mathbf h$ to use. That is, while all functions satisfying the conditions in Theorem \ref{thm: GMM} will lead to an esitmator that is consistent and asymptotically normal, some functions may lead to smaller asymptotic variance than others.
The next section explores this question formally by deriving the semiparametric efficient lower bound for $\bm \theta$ and showing an estimator that can achieve this lower bound.


\section{Semiparametric Efficiency}\label{Sec: effi}
While the GMM estimator is simple and computationally tractable compared to existing methods, it may not lead to the most efficient estimator of $\bm \theta$ under Conditions \ref{eq: semi_inner1} and \ref{eq: semi_inner2}. 
In this section, we propose a semiparametric efficient estimator of $\bm\theta$ by directly deriving the semiparametric efficient score under Conditions \ref{eq: semi_inner1} and \ref{eq: semi_inner2}.

\subsection{Efficient Score}\label{Sec: effi_score} 
We start by characterizing the orthogonal nuisance tangent space; the derivations are detailed in the Supplements. 
Under Conditions \ref{eq: semi_inner1} and \ref{eq: semi_inner2},
the joint density of $(\Y, \X)$ can be decomposed as
\begin{equation}\label{likelihood}
\begin{aligned}
\eta(\Y, \X;\bm\theta) 
&=\eta_1(\bm\Gamma\T\Y\mid \X;\bm\theta)\eta_2(\mathbf B\T\bm\Gamma_0\T\Y \mid \mathbf B_0\T\bm\Gamma_0\T\Y, \X;\bm\theta)\eta_3(\mathbf B_0\T\bm\Gamma_0\T\Y;\bm\theta)\eta_4(\X),
\end{aligned}
\end{equation} 
where $\eta_1$ is the conditional probability density function (pdf) of $\bm\Gamma\T\Y$ given $\X$, 
$\eta_2$ is the conditional pdf of $\mathbf B\T\bm\Gamma_0\T\Y$ given $(\mathbf B_0\T\bm\Gamma_0\T\Y, \X)$, $\eta_3$ and $\eta_4$ are the pdfs of $\mathbf B_0\T\bm\Gamma_0\T\Y$ and $\X$, respectively. For simplicity, we use $\eta_{1,2,3}$ to denote the collection of density functions $\eta_1$, $\eta_2$ and $\eta_3$. 

The orthogonal nuisance tangent space of $\bm\theta$ in \eqref{likelihood} is 
\begin{equation*}
	\begin{aligned}
	\Lambda^\perp = &\{\mathbf f(\Y, \X)\in\mathbb R^{q}: \E(\mathbf f\mid \bm\Gamma\T\Y,\X) \text{ is a function of }\X \text{; }\E(\mathbf f\mid\bm\Gamma_0\T\Y, \X)\\
	&\qquad\text{ is a function of } \mathbf B_0\T\bm\Gamma_0\T\Y, \X\text{; }\E(\mathbf f\mid \mathbf B_0\T\bm\Gamma_0\T\Y) = \bm 0\text{; }\E(\mathbf f\mid\X) = \bm 0\}.
	\end{aligned}
\end{equation*}
We remark that the function $\mathbf f(\Y_i,\X_i;\bm\theta)$ inside the GMM estimator and defined in \eqref{eq: GMM1} satisfies $\mathbf M\mathbf f(\Y_i,\X_i;\bm\theta)\in\Lambda^\perp$ for all $\mathbf M\in\mathbb R^{q\times d_gd_h}$.

Let $\bm\Delta_1(\mathbf Y, \mathbf X) = \Y - \E(\Y\mid\X)$, and $\bm\Delta_2( \B_0\T\bm\Gamma_0\T\Y, \X) = \E(\mathbf P_{\bm\Gamma_0\B}\Y\mid \B_0\T\bm\Gamma_0\T\Y,\X) - \E(\mathbf P_{\bm\Gamma_0\B}\Y\mid \B_0\T\bm\Gamma_0\T\Y)$.  Given the orthogonal nuisance tangent space $\Lambda^\perp$, the efficient score of $\bm \theta$ has the form $S_{\rm eff} = (S_{\rm eff, \bm\gamma}, S_{\rm eff, \mathbf b})$ with
{\small\begin{equation*}
	\begin{aligned}
		S_{\rm eff,\bm\gamma} &= \text{vec}\T\left\{\mathbf Q_{\bm\Gamma}\bm\Delta_1\dfrac{\partial \log\eta_1}{\partial (\bm\Gamma\T\Y)\T}\right\}\dfrac{\partial \text{vec}(\bm\Gamma)}{\partial\bm\gamma\T} + \text{vec}\T\bigg\{(\mathbf P_{\bm\Gamma}\Y + \bm\Delta_2)\dfrac{\partial \log\eta_3}{\partial (\mathbf B_0\T\bm\Gamma_0\T\Y)\T}\B_0\T\bigg\}\dfrac{\partial \text{vec}(\bm\Gamma_0)}{\partial\bm\gamma\T}\\
		&\quad+ \text{vec}\T\bigg[\mathbf P_{\bm\Gamma}\bm\Delta_1\bigg\{\dfrac{\partial \log\eta_2}{\partial (\mathbf B\T\bm\Gamma_0\T\Y)\T}\B\T + \dfrac{\partial \log\eta_2}{\partial (\mathbf B_0\T\bm\Gamma_0\T\Y)\T}\B_0\T\bigg\}\bigg]\dfrac{\partial \text{vec}(\bm\Gamma_0)}{\partial\bm\gamma\T},\\
		S_{\rm eff,\mathbf b} &= \text{vec}\T\bigg\{\bm\Gamma_0\T\bm\Delta_{2}\dfrac{\partial \log\eta_3}{\partial (\mathbf B_0\T\bm\Gamma_0\T\Y)\T}\bigg\}\dfrac{\partial \text{vec}(\B_0)}{\partial \mathbf b\T}.
	\end{aligned}
\end{equation*}
}
Interestingly, the efficient score of $S_{\rm eff,\mathbf b}$ only relies on the density of $\eta_3$. This implies that once  $\bm \gamma$ is known, we only need the marginal density $\eta_3$ to estimate $\mathbf b$ and the additional information contained in $\eta_2$ does not improve the efficiency of $\B$. In contrast, the efficient score of $\bm \gamma$ involves all three densities $\eta_{1,2,3}$.

If all the nuisance functions, specifically the partial derivatives of the log likelihoods, i.e. ${\partial\log\eta_1}/{\partial (\bm\Gamma\T\Y)\T}$, ${\partial\log\eta_2}/{\partial(\B\T\bm\Gamma_0\T\Y)\T}$, ${\partial\log\eta_2}/{\partial(\B_0\T\bm\Gamma_0\T\Y)\T}$, ${\partial\log\eta_3}/{\partial(\B_0\T\bm\Gamma_0\T\Y)\T}$, and $\bm\Delta_1, \bm\Delta_2$ are known, we can obtain an efficient estimator of $\bm\theta$ by solving  the score equation $\E\{S_{\rm eff}(\Y, \X;\bm\theta)\} = \bm 0$ for $\bm\theta$. 
However, in practice, they are unknown and must be estimated. The next two subsections discuss different approaches of estimating these nuisance functions, specifically the aforementioned partial derivatives,
which will lead to either a globally or locally efficient estimator of $\bm \theta$. 

    \subsection{Globally efficient estimator} \label{sec: global}
    The globally efficient estimator uses kernel-based, nonparametric estimators of $\eta_{1,2,3}$, $\bm \Delta_1$, and $\bm \Delta_2$ as plug-ins estimators and solves for $\bm \theta$. Specifically, let $K_h$ be a kernel function that satisfies Assumption \ref{assump: kernel}, say a uniform kernel or an Epanechnikov kernel. The density $\eta_1$ is estimated by a kernel density estimator of the form 
    \begin{equation*}
    	\hat\eta_1(\bm\Gamma^{\mathsf T}\y,\x) =  \dfrac{\sum_{i=1}^nK_h(\bm\Gamma^{\mathsf T}\Y_i-\bm\Gamma\T\y)K_h(\X_i-\x)}{\sum_{i=1}^nK_h(\X_i-\x)}.
    \end{equation*}
    Based on the estimator $\hat\eta_1(\bm\Gamma^{\mathsf T}\y,\x)$, ${\partial\log\hat\eta_1}/{\partial (\bm\Gamma\T\Y)\T}$ can be calculated in closed form, i.e.
    \begin{equation}\label{kdde}
    	\frac{\partial \log\hat\eta_1}{\partial (\bm\Gamma\T\y)\T} = \dfrac{\hat\eta_1'(\bm\Gamma^{\mathsf T}\y,\x)}{\hat\eta_1(\bm\Gamma^{\mathsf T}\y,\x)} =  \dfrac{\sum_{i=1}^nK_h'(\bm\Gamma^{\mathsf T}\Y_i-\bm\Gamma\T\y)K_h(\X_i-\x)}{\sum_{i=1}^nK_h(\bm\Gamma^{\mathsf T}\Y_i-\bm\Gamma\T\y)K_h(\X_i-\x)},
    \end{equation}
    where $K_h'$ is the first order derivative of $K_h$. Estimation of the other two densities $\eta_2$ and $\eta_3$ and their partial derivates are similar to $\eta_1$ and stated in Section B.1 of the Supplements. 
    
    Similarly, to estimate $\bm\Delta_1$,
    we use nonparametric kernel regressions of the form
    \begin{align}\label{eq: kernel_regression}
    	\hat{\bm\Delta}_1(\y,\x) =& \y - \frac{\sum_{i=1}^N \Y_iK_{h}(\X_i-\x)}{\sum_{i=1}^NK_{h}(\X_i-\x)}.
    \end{align} The estimation of $\bm\Delta_2$ is similar and given in Section B.1 of the Supplements. For additional details on kernel-based nonparametric estimators, see \cite{hayfield2008nonparametric}.

Let $\hat{S}_{\rm eff}(\Y_i, \X_i, \hat\eta_{1,2,3}, \hat{\bm \Delta}_1, \hat{\bm \Delta}_2; \bm\theta)$ denote the sample version of the efficient score where we replace the unknown nuisance functions with their nonparametrically estimated counterparts. The globally efficient estimator is  $\hat{\bm \theta}$ that solves the following estimating equation:
\begin{equation} \label{eq:Shat_global}
	\frac{1}{n}\sum_{i = 1}^n \hat{S}_{\rm eff}(\Y_i, \X_i, \hat\eta_{1,2,3}, \hat{\bm \Delta}_1, \hat{\bm \Delta}_2; \hat{\bm\theta}) = \bm 0.
	\end{equation}
Theorem \ref{thm: global} shows that under Assumptions \ref{assump: true density}-\ref{assump: kernel}, the estimator $\hat{\bm\theta}$ obtained from equation \eqref{eq:Shat_global} is $\sqrt{n}$-consistent and efficient. 
\begin{assump}\label{assump: true density}
	(The true conditional densities $\eta_{1,2,3}$) The true conditional densities $\eta_{1,2,3}$ are bounded away from 0 and $\infty$. The third order derivatives of $\log\eta_{1,2,3}$ around $\bm\theta$ are locally Lipschitz-continuous. Also, there exists a compact set $\bm\Theta$ which contains the true parameter value $\bm\theta$.
\end{assump}
\begin{assump}\label{assump: global_iden}
    (Identifiability of $S_{\rm eff}$) The solution to the estimating equation $\E\{{S}_{\rm eff}(\Y, \X, \eta_{1,2,3}, {\bm \Delta}_1, {\bm \Delta}_2; {\bm\theta})\}=\bm 0$ is unique.
\end{assump}
\begin{assump}\label{assump: rate}
	(Estimation of $\hat\eta_{1,2,3}$) The estimators $\hat\eta_{1,2,3}$ are obtained through a kernel density estimator with a common bandwidth $h$. The bandwidth satisfies $nh^8\rightarrow0$ and $nh^{2(1+s)}\rightarrow\infty$ as $n\rightarrow\infty$, where $s = \max(r+p-u, p+u)$.
\end{assump}
\begin{assump} \label{assump: finite}
	(\textit{Smoothness of $\bm\Delta_1$ and $\bm\Delta_2$}) The second order derivatives of $\E(\Y\mid\X)$ and $\E(\B\T\bm\Gamma_0\T\Y\mid\B_0\T\bm\Gamma_0\T\Y,\X)$ with respect to $\X$ and $(\B_0\T\bm\Gamma_0\T\Y,\X)$ are continuous for all $\bm\theta\in\bm\Theta$. 
\end{assump}
\begin{assump}\label{assump: kernel}
	(\textit{Kernel function}) Consider a univariate kernel $K_h(x)$, $0 \neq \int x^2K_h(x) < \infty$, which is bounded, symmetric, has a compact support on $[-h, h]$ and has a bounded second derivative. The multivariate kernel $K_h(\mathbf x)$ for a $d$-dimensional $\mathbf x$ has the form $K_h(\mathbf x) = \prod_{j=1}^d K_h(x_j)$ for $\mathbf x = (x_1,\ldots,x_d)\T$. 
\end{assump}
\begin{theorem}\label{thm: global}
	Under Assumptions \ref{assump: true density}--\ref{assump: kernel} and Conditions \ref{eq: semi_inner1}--\ref{eq: semi_inner2}, the estimator $\hat{\bm\theta}$ obtained from solving \eqref{eq:Shat_global} achieves the semiparametric efficiency bound $
		\mathcal V_{\bm\theta} = \E(S_{\rm eff}S_{\rm eff}\T)^{-1}
	$, i.e. as $n \rightarrow \infty$, $\hat{\bm\theta}$ satisfies
	$$\sqrt n(\hat{\bm\theta} - \bm\theta)\xrightarrow{d}\mathcal N\big(\bm 0, \mathcal V_{\bm\theta}). $$
\end{theorem}
 \noindent Assumptions \ref{assump: true density}--\ref{assump: kernel} are used to achieve consistency of kernel-based nonparametric estimators at sufficiently fast rates. These assumptions are not new and have been used in past works; see Chapter 1 in \cite{li2007nonparametric} for one textbook example.  Also, under the original inner envelope framework where the linear model in \eqref{eq: lm} is the true model, Assumptions \ref{assump: true density}, \ref{assump: global_iden}, and \ref{assump: kernel} concerning the data generating model would hold \cite{su2012inner}.

\subsection{Local efficiency and a robust score $S_{\rm eff}^{*}$} \label{sec:local_eff}
In some settings, investigators may have working models of the densities $\eta_{1,2,3}$, denoted as  $\eta_{1,2,3}^*$. 
For example, suppose the investigator imposes the working models for $\eta_{1,2,3}^*$ based on the multivariate normal model in \eqref{eq: lm}. 
Then, under \eqref{eq: lm}, the densities of $\eta_{1,2,3}^*$ follow multivariate normal distributions
 \begin{align*}
\eta_1^*(\bm\Gamma\T\Y_i,\X_i) &\sim\mathcal N(\bm\Gamma\T\Y_i - \bm\zeta_1\T\X_i, \bm\Omega), \\
\eta_2^*(\bm\Gamma_0\T\Y_i,\X_i) &\sim\mathcal N(\B\T\bm\Gamma_0\T\Y_i - \bm\zeta_2\T\X_i-\bm\mu_2\T\B_0\T\bm\Gamma_0\T\Y_i, \bm\Sigma_2), \\
\eta_3^*(\B_0\T\bm\Gamma_0\T\Y_i) &\sim\mathcal N(\B_0\T\bm\Gamma_0\T\Y_i, \B_0\T\bm\Omega_0\B_0),
\end{align*}
where $\bm\mu_2 = (\B_0\T\bm\Omega_0\B_0)^{-1}\B_0\T\bm\Omega_0\B$ and $\bm\Sigma_2 = \B\T\bm\Omega_0^{\frac{1}{2}}\mathbf P_{\bm\Omega_0^{{1}/{2}}\B}\bm\Omega_0^{\frac{1}{2}}\B$; see \cite{su2012inner} for details. Note that the terms $\bm\zeta_1\in\mathbb R^{u\times p}$ and $\bm\zeta_2\in\mathbb R^{(p-u)\times p}$ are the coordinates of $\bm\beta$ under the basis matrices $\bm\Gamma$ and $\bm\Gamma_0\B$, respectively, and the terms $\bm\Omega$ and $\bm\Omega_0$ are the covariance matrices of $\bm\Gamma\T\Y_i$ and $\bm\Gamma\T_0\Y_i$ respectively.  Also, based on the form of $\eta_{1,2,3}^*$, their partial derivatives in $S_{\rm eff}$ can be written in closed-form as
{\small\begin{align*}
    &\frac{\partial \log\eta_1^*}{\partial (\bm\Gamma\T\Y_i)\T} = -(\bm\Gamma\T\Y_i - \bm\zeta_1\T\X_i)\T\bm\Omega^{-1},\\
    &\frac{\partial \log\eta_2^*}{\partial (\B\T\bm\Gamma_0\T\Y_i)\T} = -(\B\T\bm\Gamma_0\T\Y_i - \bm\zeta_2\T\X_i-\bm\mu_2\T\B_0\T\bm\Gamma_0\T\Y_i)\T\bm\Sigma_2^{-1},\\
    &\frac{\partial \log\eta_2^*}{\partial (\B_0\T\bm\Gamma_0\T\Y_i)\T} = (\B\T\bm\Gamma_0\T\Y_i - \bm\zeta_2\T\X_i-\bm\mu_2\T\B_0\T\bm\Gamma_0\T\Y_i)\T\bm\Sigma_2^{-1}\bm\mu_2\T,\\
    &\frac{\partial \log\eta_3^*}{\partial (\B_0\T\bm\Gamma_0\T\Y_i)\T} = -(\B_0\T\bm\Gamma_0\T\Y_i)\T(\B_0\T\bm\Omega_0\B_0)^{-1}.
\end{align*}}
Then, under the multivariate normal working model, we only have to estimate finite-dimensional nuisance parameters $\bm\zeta_1$, $\bm\zeta_2$, $\bm\Omega$, and $\bm\Omega_0$ to characterize the densities $\eta_{1,2,3}$ and these nuisance parameters can be estimated at $\sqrt{n}$- rates. Also, if the working models for $\eta_{1,2,3}^*$ are correct, $S_{\rm eff}(\Y, \X, \eta_{1,2,3}^*, \bm \Delta_1, \bm \Delta_2;\bm\theta)\in\Lambda^\perp$  and the estimator based on solving for $\bm\theta$ in the estimating equation $\E\{S_{\rm eff}(\Y,\X,\eta_{1,2,3}^*,\bm \Delta_1, \bm \Delta_2 ;\bm\theta)\} = \bm 0$ is semiparametrically efficient. 

However, more often than not, the working models $\eta_{1,2,3}^*$ are likely incorrect, which implies $S_{\rm eff}( \Y,\X, \eta_{1,2,3}^*, \bm \Delta_1, \bm \Delta_2;\bm\theta)\notin\Lambda^\perp$ and the solution to the estimating equation $\E\{S_{\rm eff}(\Y,\X,\eta_{1,2,3}^*, \bm \Delta_1, \bm \Delta_2;\bm\theta)\} = \bm 0$ will lead to an inconsistent estimator of $\bm\theta$. To overcome the sensitivity of the efficient score to potential mis-specifications of the working models, we propose an alternative, robust score in $\Lambda^\perp$, denoted as $S_{\rm eff}^* = (S_{\rm eff,\gamma}^{*\mathsf T}, S_{\rm eff,\mathbf b}^{*\mathsf T})\T$ and stated below:

{\small\begin{equation*}
	\begin{aligned}
		S_{\rm eff,\bm\gamma}^* &= \text{vec}\T\bigg[\mathbf Q_{\bm\Gamma}\bm\Delta_1\bigg\{\dfrac{\partial \log\eta_1^*}{\partial (\bm\Gamma\T\Y)\T} - \E\bigg(\dfrac{\partial \log\eta_1^*}{\partial (\bm\Gamma\T\Y)\T}\bmid\X\bigg)\bigg\}\bigg]\dfrac{\partial \text{vec}(\bm\Gamma)}{\partial\bm\gamma\T} \\
		&\quad+\text{vec}\T\bigg[(\mathbf P_{\bm\Gamma}\Y + \bm\Delta_2)\bigg\{\dfrac{\partial \log\eta_3^*}{\partial (\mathbf B_0\T\bm\Gamma_0\T\Y)\T} - \E\bigg(\dfrac{\partial \log\eta_3^*}{\partial (\mathbf B_0\T\bm\Gamma_0\T\Y)\T}\bigg)\bigg\}\B_0\T\bigg]\dfrac{\partial \text{vec}(\bm\Gamma_0)}{\partial\bm\gamma\T}\\
		&\quad+ \text{vec}\T\bigg[\mathbf P_{\bm\Gamma}\bm\Delta_1\bigg\{\dfrac{\partial \log\eta_2^*}{\partial (\mathbf B\T\bm\Gamma_0\T\Y)\T}-\E\bigg(\dfrac{\partial \log\eta_2^*}{\partial (\mathbf B\T\bm\Gamma_0\T\Y)\T}\bmid \B_0\T\bm\Gamma_0\T\Y,\X\bigg) \bigg\}\B\T\bigg]\dfrac{\partial \text{vec}(\bm\Gamma_0)}{\partial\bm\gamma\T},\\
        &\quad+ \text{vec}\T\bigg[\mathbf P_{\bm\Gamma}\bm\Delta_1\bigg\{\dfrac{\partial \log\eta_2^*}{\partial (\mathbf B_0\T\bm\Gamma_0\T\Y)\T}-\E\bigg(\dfrac{\partial \log\eta_2^*}{\partial (\mathbf B_0\T\bm\Gamma_0\T\Y)\T}\bmid \B_0\T\bm\Gamma_0\T\Y,\X\bigg) \bigg\}\B_0\T\bigg]\dfrac{\partial \text{vec}(\bm\Gamma_0)}{\partial\bm\gamma\T},\\
        S_{\rm eff,\mathbf b}^* &= \text{vec}\T\bigg[\bm\Gamma_0\T\bm\varepsilon_{2}\bigg\{\dfrac{\partial \log\eta_3^*}{\partial (\mathbf B_0\T\bm\Gamma_0\T\Y)\T} - \E\bigg(\dfrac{\partial \log\eta_3^*}{\partial (\mathbf B_0\T\bm\Gamma_0\T\Y)\T}\bigg) \bigg\}\bigg]\dfrac{\partial \text{vec}(\B_0)}{\partial \mathbf b\T}.
	\end{aligned}
\end{equation*}}
An appealing feature of the new score $S_{\rm eff}^*(\Y, \X, \eta_{1,2,3}^*, \bm\Delta_1, \bm\Delta_2;\bm\theta)$ is that when the working models of $\eta_{1,2,3}^*$ are correctly specified, we have $S_{\rm eff}^* = S_{\rm eff}$ and the resulting estimator of $\bm\theta$ based on solving $\E\{S_{\rm eff}^*(\Y,\X,\eta_{1,2,3}^*, \bm\Delta_1, \bm\Delta_2;\bm\theta)\} = \bm 0$ is semiparametrically efficient. But, if any of the working models are misspecified, $S_{\rm eff}^*$ is still an element in $\Lambda^\perp$ and thus, the estimator will still be consistent and asymptotically normal. In other words, $S_{\rm eff}^{*}$ will always guarantee a consistent estimator of $\bm\theta$ irrespective of the choice of the working models and the estimator will be efficient if the working models are correct. 

The new, robust score $S_{\rm eff}^{*}$ contains new nuisance parameters in the form of conditional expectations of the partial log of $\eta_{1,2,3}^*$. Depending on the choice of the working models, these conditional expectations may or may not have to be estimated. For example, under the aforementioned working models based on the multivariate normal distribution, $S_{\rm eff}^{*}$ simplifies to
\begin{align*}
    S_{\rm eff,\gamma}^* =& -\text{vec}\T(\mathbf Q_{\bm\Gamma}\bm\Delta_1\bm\Delta_1\T\bm\Gamma\bm\Omega^{-1})\dfrac{\partial \text{vec}(\bm\Gamma)}{\partial\bm\gamma\T} - \text{vec}\T\big\{\mathbf P_{\bm\Gamma}\bm\Delta_1\bm\Delta_1\T\bm\Gamma_0\B\bm\Sigma_2^{-1}(\B\T-\bm\mu_2\T\B_0\T)\\
    &\quad+(\mathbf P_{\bm\Gamma}\bm\Delta_1 + \bm\Delta_2)\Y\bm\Gamma_0\B_0(\B_0\T\bm\Omega\B_0)^{-1}\B_0\T\big\}\dfrac{\partial \text{vec}(\bm\Gamma_0)}{\partial\bm\gamma\T},\\
    S_{\rm eff,b}^* =& -\text{vec}\T\{\bm\Gamma_0\T\bm\Delta_2\Y\bm\Gamma_0\B_0(\B_0\T\bm\Omega\B_0)^{-1}\}\dfrac{\partial \text{vec}(\bm\B_0)}{\partial\mathbf b\T}. 
\end{align*}
Notice that the conditional expectations of the partial logs of $\eta_{1,2,3}^*$ are not present and the only unknown quantities in $S_{\rm eff}^*$ are $\bm\Delta_1$ and $\bm\Delta_2$, which can be estimated using the same nonparametric regression estimators from the globally efficient estimator in Section \ref{sec: global}.

However, if the working models are more complex, investigators may have to estimate these conditional expectations. Here, we propose to estimate them in the same manner as estimating $\bm \Delta_1$ adn $\bm \Delta_2$ with a kernel regression estimator. For example, for $\E\{\partial \log \eta_1^*/\partial (\bm\Gamma\T\Y)\T\mid\X=\x\}$, we propose the following nonparametric estimator
\begin{equation*}
    \hat{\E}\left(\dfrac{\partial \log\eta_1^*}{\partial (\bm\Gamma^{\mathsf T}\Y_i)\T}\bmid\x\right) = \frac{\sum_{i=1}^N {\partial \log\eta_1^*}/{\partial (\bm\Gamma^{\mathsf T}\Y_i)\T} K_{h}(\X_i-\x)}{\sum_{i=1}^NK_{h}(\X_i-\x)}.
\end{equation*}
For additional details on the nonparametric estimators of ${\E}\{{\partial \log\eta_2^*}/{\partial (\B^{\mathsf T}\bm\Gamma_0^{\mathsf T}\Y)\T}\mid\B_0^{\mathsf T}\bm\Gamma_0^{\mathsf T}\y,\X\}$, ${\E}\{{\partial \log\eta_2^*}/{\partial (\B_0^{\mathsf T}\bm\Gamma_0^{\mathsf T}\Y)\T}\mid\B_0^{\mathsf T}\bm\Gamma_0^{\mathsf T}\Y,\X\}$, see Section B.1 of the Supplements. Practically speaking, we recommend investigators start with the multivariate normal working model as it not only mirrors the original, parametric inner envelope estimator, but also it leads to a simpler $S_{\rm eff}^{*}$.


 Let $\hat{S}_{\rm eff}^*(\Y, \X, \hat\eta_{1,2,3}^*, \hat{\bm \Delta}_1, \hat{\bm \Delta}_2; \bm\theta)$ denote the sample version of the score with the posited densities $\eta_{1,2,3}^*$ and the estimated $\bm \Delta_1, \bm \Delta_2$ from the globally efficient estimator in \eqref{eq: kernel_regression}. We define $\hat{\bm \theta}$ to be the solution to the following estimating equation:
\begin{equation} \label{eq:Shat_local}
\frac{1}{n}\sum_{i = 1}^n \hat{S}_{\rm eff}^*(\Y_i, \X_i, \hat\eta_{1,2,3}^*, \hat{\bm \Delta}_1, \hat{\bm \Delta}_2; \hat{\bm\theta}) = \bm 0.
\end{equation}
Theorem \ref{thm: local} shows that under Assumptions \ref{assump: Conditional density}-\ref{assump: data density},  the solution to \eqref{eq:Shat_local} will always be asymptotically normal and be locally efficient if all the working densities are correctly specified.
\begin{assumpB}\label{assump: Conditional density}
	(\textit{Working models $\eta_{1,2,3}^*$}) The working models $\eta_{1,2,3}^*$ are bounded away from 0 and infinity. The Hessians of $\log\eta_1^*$, $\log\eta_2^*$, $\log\eta_3^*$ are positive definite and bounded on a compact set $\bm\Theta$ that contains $\bm\theta$. The third order derivatives of $\log\eta_{1,2,3}^*$ around $\bm\theta$ are locally Lipschitz-continuous.
\end{assumpB}
\begin{assumpB}\label{assump: local_iden}
    (Identifiability of $S_{\rm eff}^*$) The solution to the estimating equation $\E\{{S}_{\rm eff}^*(\Y, \X, \eta_{1,2,3}^*, {\bm \Delta}_1, {\bm \Delta}_2; {\bm\theta})\}=\bm 0$ is unique.
\end{assumpB}
\begin{assumpB}\label{assump: bandwidth}
    (\textit{Smoothness of conditional expectations}) The second order derivatives of ${\E}\{{\partial \log\eta_1^*}/{\partial (\bm\Gamma\T\Y)\T}\mid \mathbf{X} \}$, ${\E}\{{\partial \log\eta_2^*}/{\partial (\B\T\bm\Gamma_0\T\Y)\T}\mid\B_0\T\bm\Gamma_0\T \mathbf{Y},\mathbf{X}\}$ and ${\E}\{{\partial \log\eta_2^*}\\/{\partial(\B_0\T\bm\Gamma_0\T\Y)\T}\mid\B_0\T\bm\Gamma_0\T \mathbf{Y},\mathbf{X} \}$ are uniformly continuous for all $\bm\theta\in\bm\Theta$. If they are estimated via kernel regressions in Section B.1 of the Supplements, the bandwidth $h$ satisfies $nh^8\rightarrow 0$ and $nh^{2(p + r-u-d)}\rightarrow\infty$.
\end{assumpB}
\begin{assumpB}\label{assump: data density}
    (\textit{Data density}) The data density for $(\X, \Y)$ is bounded away from 0 and infinity and has twice continuously differentiable derivatives.
\end{assumpB}
\begin{theorem}\label{thm: local}
	Suppose Assumptions \ref{assump: Conditional density}--\ref{assump: data density}, \ref{assump: finite}--\ref{assump: kernel} and Conditions \ref{eq: semi_inner1}--\ref{eq: semi_inner2} hold. Then, the estimator in \eqref{eq:Shat_local} is consistent and asymptotically normal, i.e.
		$$\sqrt n(\hat{\bm\theta} - \bm\theta)\xrightarrow{d}\mathcal N\{\bm 0, \mathbf C_2^{-1}\mathbf D_2(\mathbf C_2^{-1})\T\},$$ where
	$$\mathbf C_2 = \E\left\{\dfrac{\partial S_{\rm eff}^*(\Y_i,\X_i, \eta_{1,2,3}^*, \bm\Delta_1, \bm\Delta_2; \bm\theta)}{\partial\bm\theta}\right\}, \hskip .5cm  \mathbf D_2 = \E\left\{S_{\rm eff}^*(\Y_i,\X_i, \eta_{1,2,3}^*, \bm\Delta_1, \bm\Delta_2; \bm\theta)^{\otimes2}\right\}. $$
Additionally, if the working models are correctly specified, i.e. $\eta_1^*=\eta_1$, $\eta_2^*=\eta_2$ and $\eta_3^*=\eta_3$, the estimator is locally efficient.
\end{theorem}
We make some remarks about the assumptions underlying Theorem \ref{thm: local}. First,
the aforementioned multivariate normal model automatically satisfy Assumptions \ref{assump: Conditional density}--\ref{assump: bandwidth} and there is no need to nonparametrically estimate the conditional expectations of the partial log likelihoods. Second, the consistency and asymptotic normality of the locally efficient estimator requires weaker regularity condition for the bandwidth rate than the globally efficient estimator. Specifically, the bandwidth rate $nh^{2(1+s)}\rightarrow \infty$ from Assumption \ref{assump: rate} implies $nh^{2(p+r-u-d)}\rightarrow\infty$ in Assumption \ref{assump: bandwidth}.


\section{Computational and Other Considerations} \label{sec: comp}
\subsection{An Alternating Algorithm to Solve Estimating Equations}  \label{sec:algo}
Both the globally and locally efficient estimators in equations \eqref{eq:Shat_global} and \eqref{eq:Shat_local}, respectively, require solving potentially non-linear estimating equations. While there are many general-purpose, optimization procedures to solve such equations (e.g., Chapter 11 in \cite{nocedal2006numerical}), we propose a procedure based on an alternating algorithm The optimization algorithm for the globally efficient estimator is detailed in Algorithm \ref{algo: global} and the optimization algorithm for the locally efficient estimator is detailed in Algorithm \ref{algo: local}.

\begin{algorithm}
    \caption{An alternating algorithm to solve equation \eqref{eq:Shat_global}}

    \label{algo: global}
    \begin{algorithmic}
		\State Inputs: (1) data $(\Y_i,\X_i)$, $i=1,\ldots,n$; (2) kernel function $K_h$; (3) convergence threshold $\delta > 0$.
		\State 1. Initialize $\bm\theta$ using the GMM estimator in Section \ref{Sec: semi} and denote it as $\bm\theta^{(0)}$. 
		\State 2. Estimate $\bm\Delta_1$ using equation \eqref{eq: kernel_regression}.
		\While{$|\bm\theta^{(t+1)} - \bm\theta^{(t)}| > \delta$}
		\State 3. Obtain $\bm\Gamma^{(t)}$, $\bm\Gamma^{(t)}_0$, $\B^{(t)}$, $\B^{(t)}_0$ from $\bm\theta^{(t)}$ using Lemma \ref{lemma: repar}.
		\State 4. Using the estimates in step 3, use the kernel estimators to nonparametrically estimate $\eta_{1,2,3}$ and $\bm \Delta_2$. 
		\State 5. Using the estimates in step 4, solve the estimating equation in \eqref{eq:Shat_global} for $\bm \theta$ and denote it as $\bm{\theta}^{(t+1)}$.
		\EndWhile
		\State \textbf{Output:} Final estimator $\hat{\bm\theta} = \bm\theta^{(t+1)}$.
    \end{algorithmic}
\end{algorithm}
\begin{algorithm}
	\caption{An alternating algorithm to solve equation \eqref{eq:Shat_local}}
	\label{algo: local}
    \begin{algorithmic}
		\State Inputs:  (1) data $(\Y_i,\X_i)$, $i=1,\ldots,n$; (2) working models of $\eta_{1,2,3}^*$; (3) kernel function $K_h$; (4) convergence threshold $\delta > 0$.
		\State 1. Initialize $\bm\theta$ using the GMM estimator in Section \ref{Sec: semi} and denote it as $\bm\theta^{(0)}$. 
		\State 2. Estimate $\bm\Delta_1$ using equation \eqref{eq: kernel_regression}.
		\While{$|\bm\theta^{(t+1)} - \bm\theta^{(t)}| > \delta$}
		\State 2. Obtain $\bm\Gamma^{(t)}$, $\bm\Gamma^{(t)}_0$, $\B^{(t)}$, $\B^{(t)}_0$ from $\bm\theta^{(t)}$ using Lemma \ref{lemma: repar}.
		\State 3. Using the estimates in step 2, use the kernel estimators to estimate $\bm\Delta_2$.
		\State 4. Using the estimates in step 2, estimate relevant parameters in the working models $\eta_{1,2,3}^*$. 
		\State 5. If necessary, estimate ${\E}\{{\partial \log\eta_1^*}/{\partial (\bm\Gamma^{(t)\mathsf T}\Y)\T}\mid\X\}$, ${\E}\{{\partial \log\eta_2^*}/{\partial (\B^{(t)\mathsf T}\bm\Gamma_0^{(t)\mathsf T}\Y)\T}\mid\B_0^{(t)\mathsf T}\bm\Gamma_0^{(t)\mathsf T}\y,\X\}$, ${\E}\{{\partial \log\eta_2^*}/{\partial (\B_0^{(t)\mathsf T}\bm\Gamma_0^{(t)\mathsf T}\Y)\T}\mid\B_0^{(t)\mathsf T}\bm\Gamma_0^{(t)\mathsf T}\Y,\X\}$ via equations (S7)--(S10).
		\State 6. Using the estimates from steps 3 to steps 5, solve the estimating equation in \eqref{eq:Shat_local} and denote it as $\bm \theta^{(t+1)}$.
		\EndWhile
		\State \textbf{Output:} The final estimator $\hat{\bm\theta} = \bm\theta^{(t+1)}$.
		
	\end{algorithmic}
\end{algorithm}
We make some remarks about both algorithms. First, both algorithms alternate between estimating $\bm\theta$ and the relevant nuisance parameters in the scores $S_{\rm eff}$ and $S_{\rm eff}^*$. The rationale behind alternating between these two estimation steps is that investigators can directly use popular, off-the-shelf software for  nonparametric kernel-based estimation. For example, for nonparametric, kernel regression estimators of $\bm \Delta_1$ and $\bm \Delta_2$, investigators can use the R function ``npreg" in the R package ``np" \cite{hayfield2008nonparametric}. For the densities $\eta_{1,2,3}$ and their derivatives, investigators can use the R function ``kde" and ``kdde" in the R package ``ks" \cite{duong2007ks}. For choosing the bandwidth parameter, investigators can directly use the R functions ``npregbw" and ``Hpi" in the above two packages, which use cross validation to choose the bandwidth parameters; note that while there is a rate-driven choice of the bandwidth parameter $h$ from our theory (see Theorem \ref{thm: global} and Theorem \ref{thm: local}), in practice, cross validation is used. Second, to evaluate the partial derivatives of the basis matrices with respect to their vector representations (e.g. $\partial \text{vec}(\bm\Gamma_0) /\partial\bm\gamma\T)$, we can use Lemma \ref{lemma: repar}, which provides a continuous mapping between $(\bm$ $\bm\Gamma,\B)$ and $\bm\gamma, \mathbf b$, and finite differencing \cite{nocedal2006numerical}. Third, both algorithms initialize with the GMM estimator, but other initializations are possible. In practice, we recommend using the GMM estimator with the identity functions $\mathbf g$ and $\mathbf h$ for simplicity.

\subsection{Dimension Selection of the Inner Envelope}\label{Sec: dim}
An important step in using any dimension-reduction procedure is selecting the reduced dimension subspace. For the inner envelope, we need to estimate the dimensions of 
$\mathcal S_1$, $\mathcal S_2$ and $\mathcal S_3$ and a common approach to estimating the dimensions involves using Bayesian information criterion. But, this approach relies on a parametic model between the responses and the outcome whereas our semiparametric generalization does not impose such a model. Instead, we use a nonparametric, bootstrapped-based procedure in \cite{dong2010dimension, ye2003using} adapted to our inner envelope setting to estimate the dimension of the subspaces.

Formally, for each possible dimension of $\mathcal S_i$,  we calculate the GMM estimators of the corresponding $\bm \theta$ with Theorem \ref{thm: GMM}. We also take nonparametric bootstrapped samples of $(\mathbf{Y}, \mathbf{X})$ and obtain $B$ bootstrapped estimates of $\bm \theta$. 
For each subspace $i=1,2,3$, let $\hat{\mathcal S}_{i,k}$ denote the estimated space $\mathcal S_i$ under dimension $k$ and $\hat{\mathcal S}_{i,k}^b$ denote the estimate based on the $b$-th bootstrap sample for $b = 1,\ldots,B$. We then jointly estimate the dimensions of $\mathcal{S}_1, \mathcal{S}_2$, and $\mathcal{S}_3$ by solving 
\begin{equation} \label{eq: dim_select}
	(\hat u, \hat d) = \argmax_{u,d } \dfrac{1}{B}\sum_{b = 1}^B\{q^2(\hat{\mathcal S}_{1,u}, \hat{\mathcal S}_{1,u}^b)+q^2(\hat{\mathcal S}_{2,d}, \hat{\mathcal S}_{2,d}^b)+q^2(\hat{\mathcal S}_{3,r-u-d}, \hat{\mathcal S}_{3,r-u-d}^b)\},
\end{equation}where we assume the dimensions of $\mathcal S_1$, $\mathcal S_2$ and $\mathcal S_3$ are greater or equal to 1. Here, $q^2$ is the Hotelling's vector correlation coefficient  \cite{harold1936relations} between two spaces. That is, for any  subspaces $\mathcal A$ and $\mathcal B\subset\mathbb R^{r}$, $q^2(\mathcal{A},\mathcal{B})$ is defined as
\begin{equation*}
	q^2(\mathcal A,\mathcal B) = \det (\mathbf B\T\mathbf A\mathbf A\T\mathbf B),
\end{equation*}
where $\mathbf A$ and $\mathbf B$ are any basis matrices for the spaces $\mathcal A$ and $\mathcal B$.
The correlation coefficient $q^2$ is bounded between $0$ and $1$ where higher value of $q^2$ indicates higher correlation between the two subspaces. In the extreme case where $q^2 (\mathbf A,\mathbf B)= 1$, $\text{span}(\mathbf A) = \text{span}(\mathbf B)$. 

Roughly speaking, equation \eqref{eq: dim_select} selects the dimensions where the intra-correlation between the bootstrapped estimates of the subspaces and the estimated subspace are small. If the true dimension of a subspace is $k^*$, but we estimated the subspace assuming it is of dimension $k \neq k^*$, the intra-correlation between the bootstrapped estimates assuming $k$ will be larger than that from assuming $k = k^*$ and thus, the selected dimension from \eqref{eq: dim_select} will be closer to the true dimension. For additional details behind  using intra-correlation of bootstrapped estimates of dimension-reduced subspaces to estimate dimensions, see \cite{dong2010dimension, ye2003using}.

\section{Simulations Study}\label{Sec: simu}
We conduct three simulation studies to numerically assess our proposed methods. The first two simulation studies concern linear and non-linear models with $r = 4$ responses and $p = 2$ regressors. 
We set the true dimension of  $\mathcal S_1$,$\mathcal S_2$ and $\mathcal S_3$ to be 1, 1, 2 respectively, and consider five different sample sizes,  $n=100, 300, 500, 750, 1000$. We compare our methods with existing approaches 
and compare the performance of each method by computing ${\rm dist}(\mathcal S_j, \hat{\mathcal S}_j) = \|\mathbf P_{\mathcal S_j} - \mathbf P_{\hat{\mathcal S}_j}\|_F$, with a smaller distance implying a better method. Additionally, for the linear simulation model in Section \ref{sec: simu_1}, we compute $\|\hat{\bm\beta} - \bm\beta\|^2_F$ and for the non-linear simulation model in Section \ref{sec: simu_2}, we compare the out-of-sample predictive root mean squared error. Finally, for the last simulation study, we generate our simulation model based on the well-known iris data \cite{fisher1936use}.

\subsection{Linear model with normal errors}\label{sec: simu_1}
Consider the following model 
\begin{align} \label{eq:sim_model}
\Y_i &= \bm\Gamma f_1(\X_i) + \bm\Gamma_0\B f_2(\X_i) + \bm\varepsilon_i, \quad{} X_{1i}, X_{2i} \sim U[-5, 5] 
\end{align}
Let $\bm\Gamma = (0.50,0.50,0.50,0.50)\T \in\mathbb R^{4\times 1}$, $\B = (1, 2, 3)\T/\sqrt{14}\in\mathbb R^{3\times 1}$, 
$\bm\Gamma_0 = (\bm\Gamma_{01},\bm\Gamma_{02}, \bm\Gamma_{03}) \in\mathbb R^{4\times 3}$ with $\bm\Gamma_{01} = (0.50,-0.83,0.17,0.17)\T$, $\bm\Gamma_{02} = (0.50, 0.17, -0.83, 0.17)\T$, $\bm\Gamma_{03} = (0.50, \\0.17,0.17,-0.83)\T$. For this simulation study, we set $f_1$, $f_2$, and $\bm \varepsilon_i$ as follows 
\begin{equation*}
	\begin{aligned}
	&f_1(\X_i) = X_{1i}, \quad f_{2}(\X_i) = X_{1i}+X_{2i}, \quad\bm\varepsilon_i = \bm\Gamma\varepsilon_{1i} + \bm\Gamma_0\B(0.2,0.2)\bm\varepsilon_{2i}  + \bm\Gamma_0\B_0\bm\varepsilon_{2i}, \\
	&\varepsilon_{1i}\sim\mathcal N(0, 1),\quad \bm\varepsilon_{2i}\sim\mathcal N\left(\bm 0, 100\mathbf I_2\right), \quad i = 1,\ldots, n.
	\end{aligned}
\end{equation*}
With the above specifications, the simulation model is equivalent to the linear model in equation \eqref{eq: lm} where $\bm\beta = \bm\Gamma\bm\eta_1\T+\bm\Gamma_0\B\bm\eta_2\T$, $\bm\eta_1\T = (1, 0)$, $\bm\eta_2\T = (1, 1)$, $\bm\Omega_1 = 72.5$, $\bm\Omega_0 = (\bm\omega_{01}, \bm\omega_{02},\bm\omega_{03})$, $\bm\omega_{01} = (0.25, 0.25, 0.25)\T$, $\bm\omega_{02} = (0.25,68.25, -31.75)\T$, and $\bm\omega_{03} = (0.25,\\-31.75, 68.25)\T$.  
Also, Conditions \ref{eq: semi_inner1} and \ref{eq: semi_inner2} hold with the subspace dimensions $u=1$ and $d = 1$. We apply our proposed methods in Sections \ref{Sec: semi} and \ref{Sec: effi} as well as the original inner envelope estimator developed under a normal linear model. We remark that for the GMM estimator, the functions $\mathbf g$ and $\mathbf h$ are chosen to be identity functions. 

Figure \ref{fig: dist_lm} shows the results for ${\rm dist}(\mathcal S_j, \hat{\mathcal S}_{j})$ across different methods; for detailed numerical results, see Table \ref{tb: simu2} of the Supplements. We see that the locally efficient estimator with correctly specified density performs better than the globally efficient estimator and the GMM estimator performs the worst. Also, as expected, the original inner envelope estimator performed the best since the simulation model is linear and the errors are normally distributed. Finally, we remark that across different sample sizes, the correct inner envelope dimension (i.e. $u=1$, $d = 1$) was selected 97\% of the time.

\begin{figure}[!h]
		\centering
		\includegraphics[width=.7\linewidth]{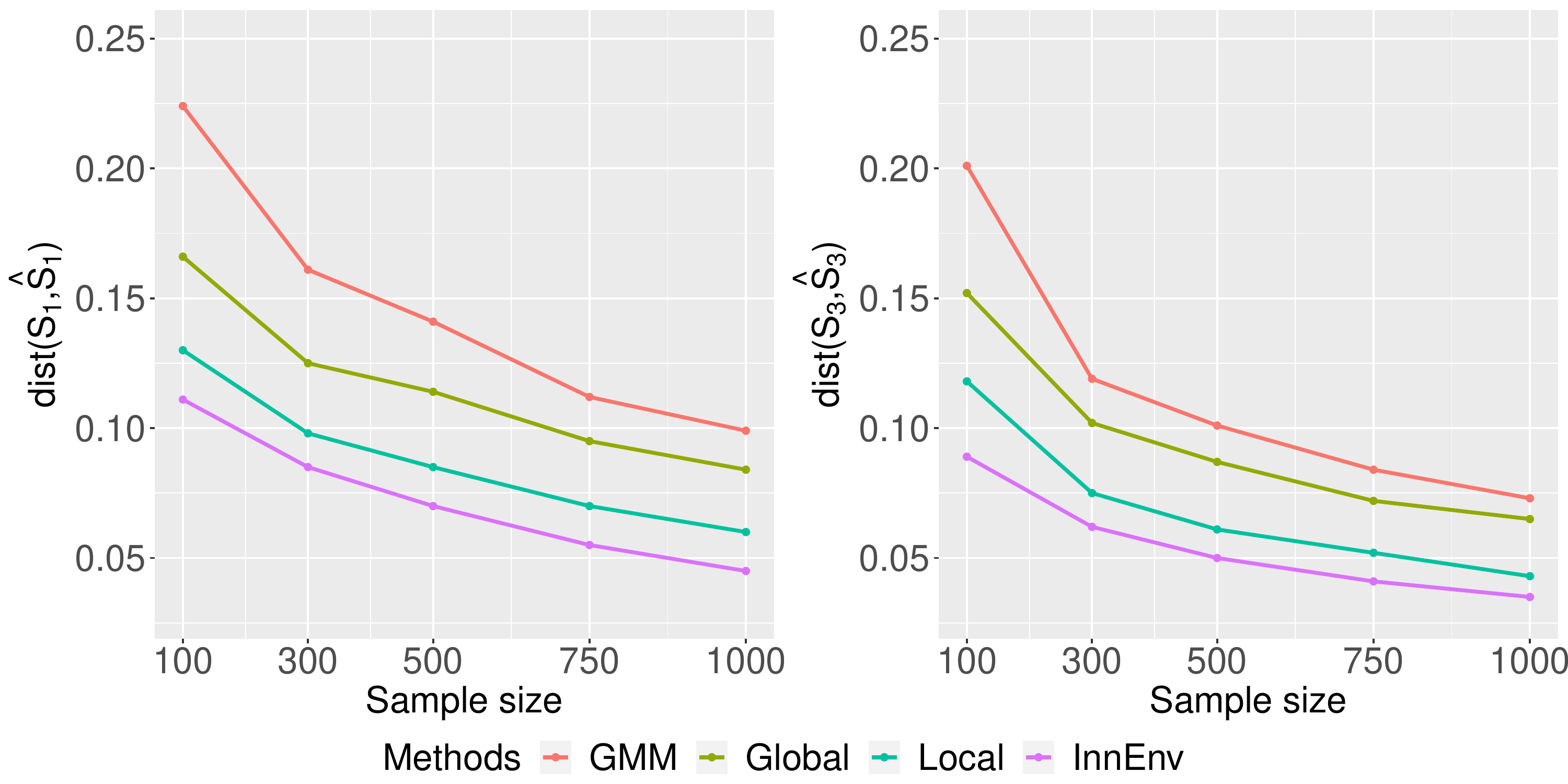}  
	\caption{Distance between the true space $\mathcal S_j$ and the estimated space $\hat{\mathcal S}_{j}$, denoted as ${\rm dist}(\mathcal S_j, \hat{\mathcal S}_{j})$, from the linear simulation model in Section \ref{sec: simu_1}. GMM represents the GMM estimator, global represents the globally efficient estimator, local represents the locally efficient estimator, and InnEnv represents the original inner envelope estimator.}
	\label{fig: dist_lm}
\end{figure}

Next, we compare the estimates of the regression parameter $\bm\beta$. As discussed in Section \ref{Sec: pre}, once we have an estimate of the basis matrices $\hat{\bm\Gamma}$ and $\hat\B$, we can estimate $\hat{\bm\beta}$ via projections; see 
 Section A.1 of the Supplements for details.
Also, since the model is linear, we include four additional estimators for $\bm\beta$: the OLS estimator, the original envelope estimator, the response partial least squares (PLS) estimator \cite{cook2018introduction} and the oracle estimator that assumes that the inner envelope spaces are known a priori. 
Figure \ref{fig: lm_mse} shows $\|\hat{\bm\beta} - \bm\beta\|_F^2$ across different methods; see Table \ref{tb: simu2} in the Supplements for additional details. Compared to the OLS estimator, there are gains from using either the original inner envelope estimator or our proposed estimators. Specifically, the GMM estimator has a smaller $\|\hat{\bm\beta} - \bm\beta\|_F^2$ than the OLS estimator, the PLS estimator, and the original envelope estimator. But, the GMM estimator is worse than the original inner envelope estimator, the locally efficient estimator and the globally efficient estimator. These observations agree with what's expected from theory as the GMM estimator utilizes the inner envelope space compared to the OLS estimator, the PLS estimator, and the original envelope estimator, leading to better performance. But, the GMM estimator is not designed to be semiparametrically efficient compared to the local, global, and the original inner envelope estimator. 

Finally, in Section \ref{sec: numerical_results} of the supplements, we calculate the out-of-sample predictive root mean squared error (RMSE) when $n = 500$. In short, because the underlying data generating model is linear, the predictive performance of the different methods mirror the performance of estimating $\bm\beta$. 

\begin{figure}[!h]
	\centering
	\includegraphics[width=.55\linewidth]{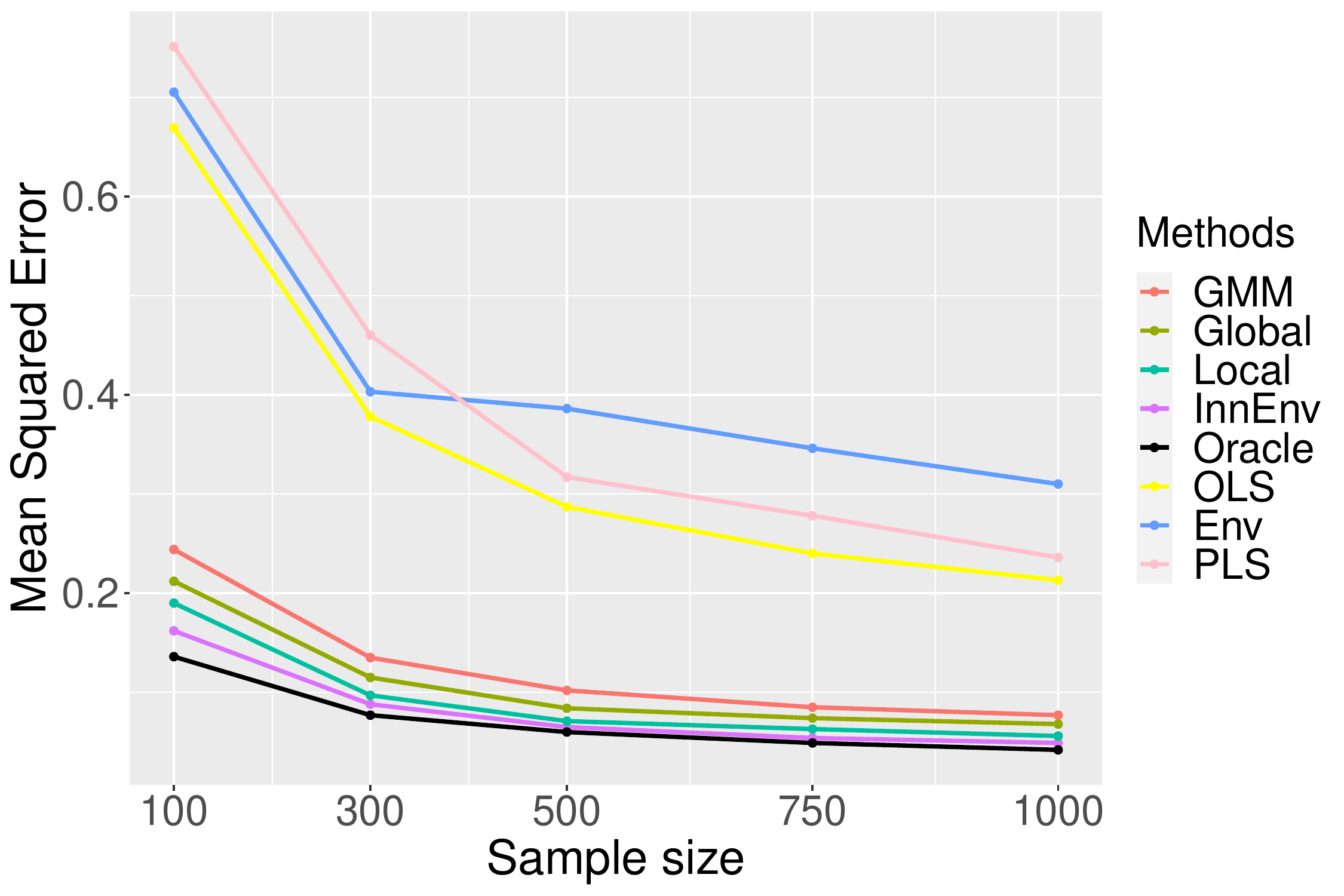}  
\caption{Mean squared error of estimating $\bm\beta$ (i.e. $\|\hat{\bm\beta} - \bm\beta\|_F^2$) in Section \ref{sec: simu_1}. GMM represents the GMM estimator, global represents the globally efficient estimator, local represents the locally efficient estimator, InnEnv represents the original inner envelope estimator, oracle represents the oracle OLS estimator where the inner envelope structure is known a priori, OLS represents the naive ordinary least squares estimator, Env represents the original envelope estimator, and PLS represents the partial least squares estimator.}
\label{fig: lm_mse}
\end{figure}

\subsection{Non-linear model with non-normal errors}\label{sec: simu_2}
For this simulation study, we consider the same general model in \eqref{eq:sim_model} and the subspace matrices $\bm \Gamma, \mathbf{B}$, except we consider the following $f_1$, $f_2$, and $\bm \varepsilon_i$:
\begin{equation*}
\begin{aligned}
&f_1(\X_i) = X_{1i}^2\text{sgn}(X_{2i}), \qquad f_{2}(\X_i) = 20\sin\{0.5(X_{1i}+X_{2i})\},\\
&\bm\varepsilon_{i} = \bm\Gamma\varepsilon_{1i} + \bm\Gamma_0\B\cdot\bm{0.1}\T\bm\varepsilon_{2i} + \bm\Gamma_0\B_0\bm\varepsilon_{2i}, \\
&\varepsilon_{1i}\sim t_5(0,1), \quad\bm\varepsilon_{2i}\sim t_5(\bm 0, 100\mathbf I_2)
\end{aligned}
\end{equation*}
Here, $t_\nu(\bm\mu,\bm\Sigma)$ denotes a multivariate $t$ distribution with mean $\bm\mu$, covariance matrix $\bm\Sigma$ and degrees of freedom $\nu$, and $\bm{0.1}$ denotes the vector $\bm 1 \times 0.1$. Overall, the above specification creates a non-linear, non-normal, and heteroskedatsic model between the responses and the regressors. We also remark that Conditions \ref{eq: semi_inner1} and \ref{eq: semi_inner2} are satisfied for the above non-linear model.

Figure \ref{fig: dist_nlm} shows the results of ${\rm dist}(\mathcal S_j,\hat {\mathcal S}_j)$ under different methods; for additional details, see Table \ref{tb: simu2} in the Supplements. Because the underlying model is non-linear and has non-normal errors, the original inner envelope estimator has a high {\rm dist} even as $n$ increases. In contrast, our proposed estimators which do not rely on linearity or normality show that the underlying subspaces are being estimated correctly with ${\rm dist}$ shrinking towards zero as $n$ increases. Also, between the locally efficient estimator and the globally efficient estimator, the locally efficient estimator uses wrong normal working models for the densities $\eta_1$, $\eta_2$, and $\eta_3$ leading to worse performance than the globally efficient estimator, which estimate these densities nonparametrically via kernel regression. Finally, we remark that the correct inner envelope dimension $u = 1$, $d = 1$ was selected $95\%$ of the time.

\begin{figure}[!h]
	  \centering
	  \includegraphics[width=.7\linewidth]{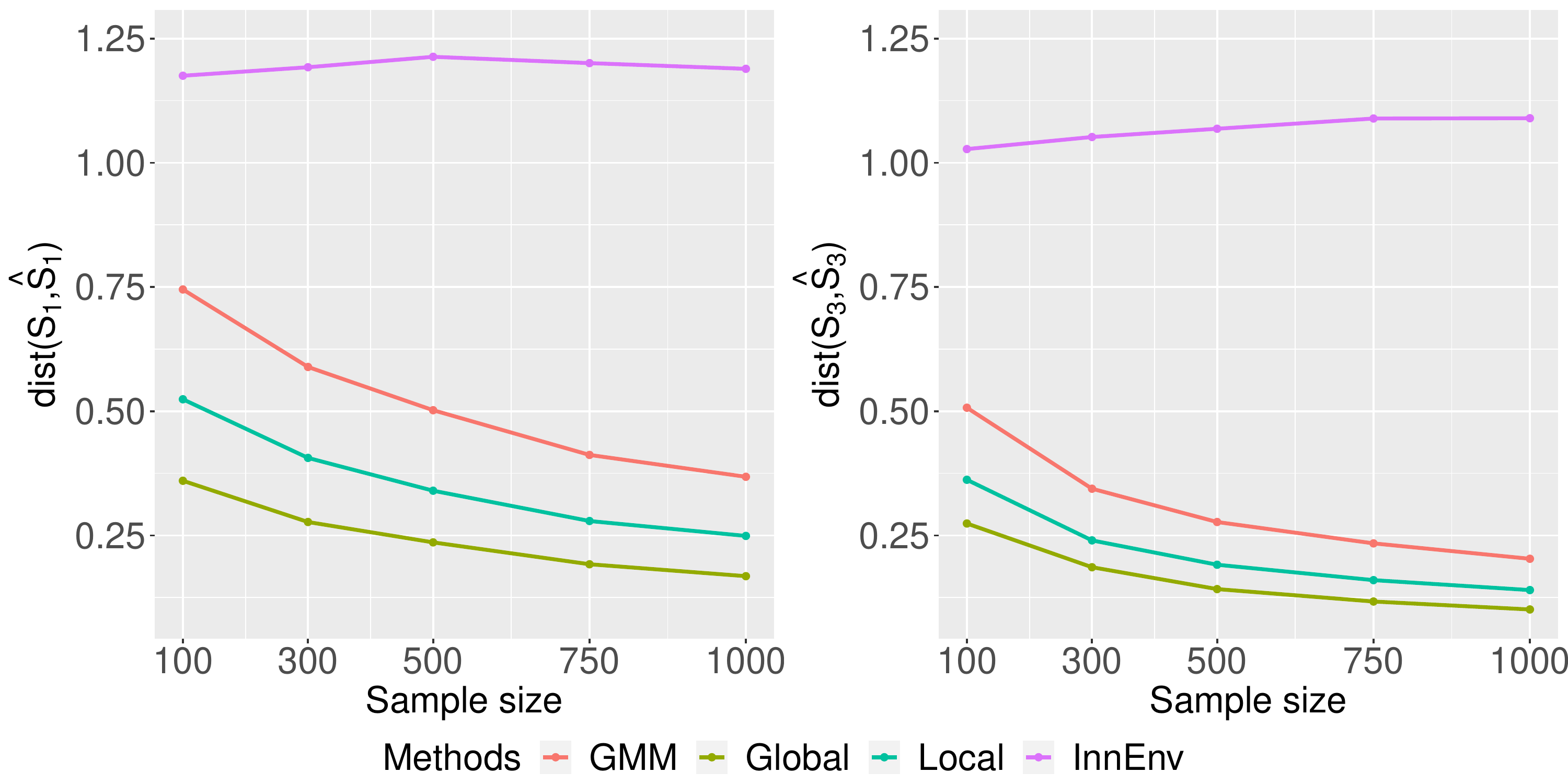}  
  \caption{Distance between the true space $\mathcal{S}_j$ and the estimated space $\hat{\mathcal{S}}_j$, denoted as ${\rm dist}(\mathcal{S}_j, \hat{\mathcal{S}}_j)$ for the non-linear model in Section \ref{sec: simu_2}. GMM represents the GMM estimator, global represents the globally efficient estimator, local represents the locally efficient estimator, and InnEnv represents the original inner envelope estimator.}
  \label{fig: dist_nlm}
\end{figure}

Next, we evaluate the out-of-sample predictive RMSE when $n = 500$.  To do this, we create pseudo-outcomes $\tilde \Y_i = (\hat{\bm\Gamma}\T\Y_i, \hat{\B}\T\hat{\bm\Gamma}_0\T\Y_i)$ based on the estimated subspaces from above and run a supervised learning algorithm with $\tilde \Y_i$ as the outcome and $\X_i$ as the predictor. 
Note that after training the supervised learning model and obtaining the predictions for the pseudo-outcome, we can transform the pseudo-outcome back to the original outcome via $\hat\Y_i = (\hat{\bm\Gamma}, \hat{\bm\Gamma}_0\hat\B)\tilde\Y_i$. We randomly split 80\% of the data into training data $(\Y_{train},\X_{train})$ ($i =1,\ldots,m$) and the rest into test data $(\Y_{test}, \X_{test})$ ($i = m+1,\ldots,n$) and the predictive RMSE is evaluated on the test data. Finally, for our supervised learning algorithms, we use XGBoost \cite{chen2016xgboost} with the default hyperparameters in the R package \cite{chen2015xgboost}.

The predictive RMSE for ``Local", ``Global", and ``GMM" methods in Figure \ref{fig: dist_nlm} are 18.94, 18.83, and 19.21, respectively. Also, the oracle preditive RMSE, which is the RMSE from predictions that use the true subspaces, is 18.64 and the naive predictive RMSe, which is the RMSE from predictions that use the original responses $\Y$ is 22.75. 
We see that incorporating the inner envelope structure into supervised learning methods reduces out-of-sample predictive RMSE compared to the naive method that directly use the original outcome. Also, the locally efficient and globally efficient methods have a predictive RMSE that is close to the oracle method, and broadly speaking, the predictive performance roughly follows the performance of estimating the subspaces in Figure \ref{fig: dist_nlm}. 
\subsection{Synthetic dataset based on the iris data}\label{sec: iris}
Our third simulation study mirrors is based on the 
the classic \textit{iris} dataset by Fisher \cite{fisher1936use}, which has been used by other envelope-based method \cite{su2012inner}. Briefly, the data contains 150 samples of iris species (setosa, versicolor, and virginica) along with their flower characteristics (sepal length, sepal width, and petal length). We take the species, dichotomized to two dummy variables, as predictors We also standardize the flower characteristics to mean zero and one standard deviation. 
Also, as a sanity check, we added two, random artificial responses $( Z_1, Z_2)\overset{i.i.d}{\sim}\mathcal N(\bm0,\mathbf I_2)$ to the original set of responses. Our algorithms should identify these two responses as part of the subspace $\mathcal S_3$. In total, we have six responses $\Y\in\mathbb R^6$ and two predictors $\mathbf X\in\mathbb R^2$. 


We first fit a multivariate linear regression of $\Y$ against $\X$. We also conduct a Shapiro-Wilk normality test for the error terms based on the linear model. The test suggests that the responses are non-normal, with  $p$-value less than $10^{-12}$, and the original inner envelope method may be inappropriate in this setting. 

Next, we run our globally efficient estimator using the selected dimension of $u = \text{dim}(\mathcal S_1) = 1$ and $d = \text{dim}(\mathcal S_3) = 4$ from the nonparametric bootstrap method;  note that the original envelope method selected $u=4$ as the dimension, suggesting that the envelope does not discover structure beyond the two noise responses we added artifically.  
Also, we verified whether the estimated $\mathcal S_3$ contains the subspace spanned by the two artificially added random responses. At a high level, this verification involves checking the distances between the estimated $S_3$ and the space spanned by the two responses and we found that our estimated $\mathcal S_3$ contains the subspace; see Section B.4 in the Supplement for details.  

Next, we compare the estimated regression parameter $\bm\beta$ under different methods. We also conduct nonparametric bootstrap 100 times to obtain estimates of the standard errors of $\hat{\bm\beta}_{ij,\rm ols}$, $\hat{\bm\beta}_{ij,\rm env}$, $\hat{\bm\beta}_{ij,\rm InnEnv}$, $\hat{\bm\beta}_{ij,\rm global}$ and $\hat{\bm\beta}_{ij,\rm local}$, where we posit the normal densities for the locally efficient estimator. The point estimates, bootstrap standard errors and $p$-values are given in Table \ref{tb: iris} in the Supplements. On average, the mean ratio of the standard error of the globally efficient estimator over that of the OLS estimator is 1.83. The mean ratios of the standard errors for the locally efficient estimator and the original inner envelope estimator compared to that of the OLS estimator are 1.70 and 1.76. Finally, the mean ratio of the standard error of the envelope estimator over that of the OLS estimator is 1.12. Figure \ref{fig: ecdf1} visualizes these results by comparing the empirical cumulative distribution functions (ECDF) of ${\rm sd}(\hat{\bm\beta}_{ij,\rm ols})/{\rm sd}(\hat{\bm\beta}_{ij,\rm env})$ and ${\rm sd}(\hat{\bm\beta}_{ij,\rm ols})/{\rm sd}(\hat{\bm\beta}_{ij,\rm global})$ for each element of $\bm \beta$ matrix. Roughly speaking, these results imply that to achieve the same power to test the null hypothesis that $\beta_{ij}$ is zero versus a fixed, alternative hypothesis, a Wald test based on the globally efficient estimator only requires about half of the original sample size (54.6\%) compared to the Wald test based on the OLS estimator. Section B.4 of the supplements contains additional results from the analysis.


\begin{figure}[!h]
	\centering
	\begin{subfigure}{.35\textwidth}
	  \centering
	  \includegraphics[width=\linewidth]{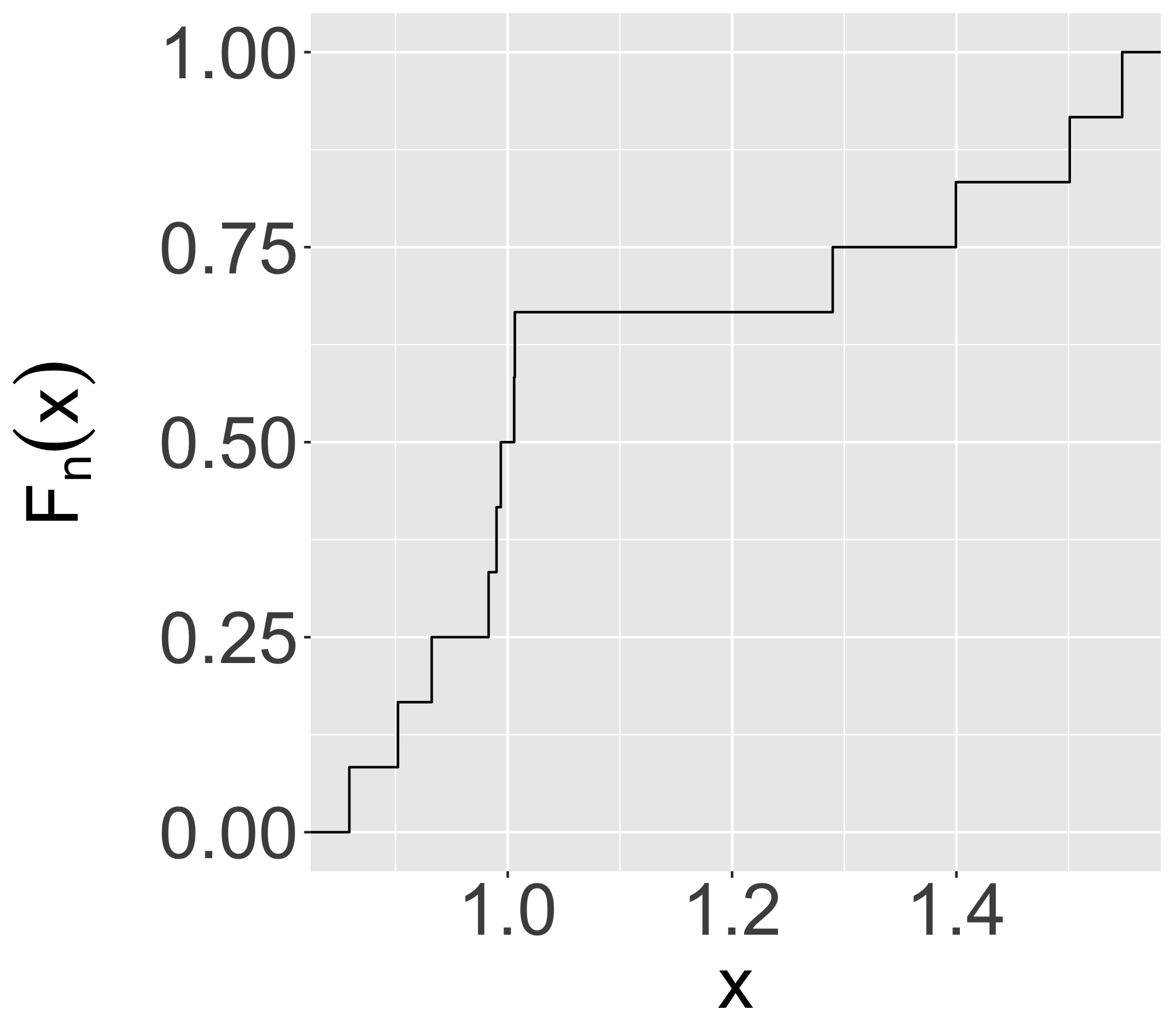}  
	  \caption{Envelope}
	\end{subfigure}
	\hspace{2mm}
	\begin{subfigure}{.35\textwidth}
	  \centering
	  \includegraphics[width=\linewidth]{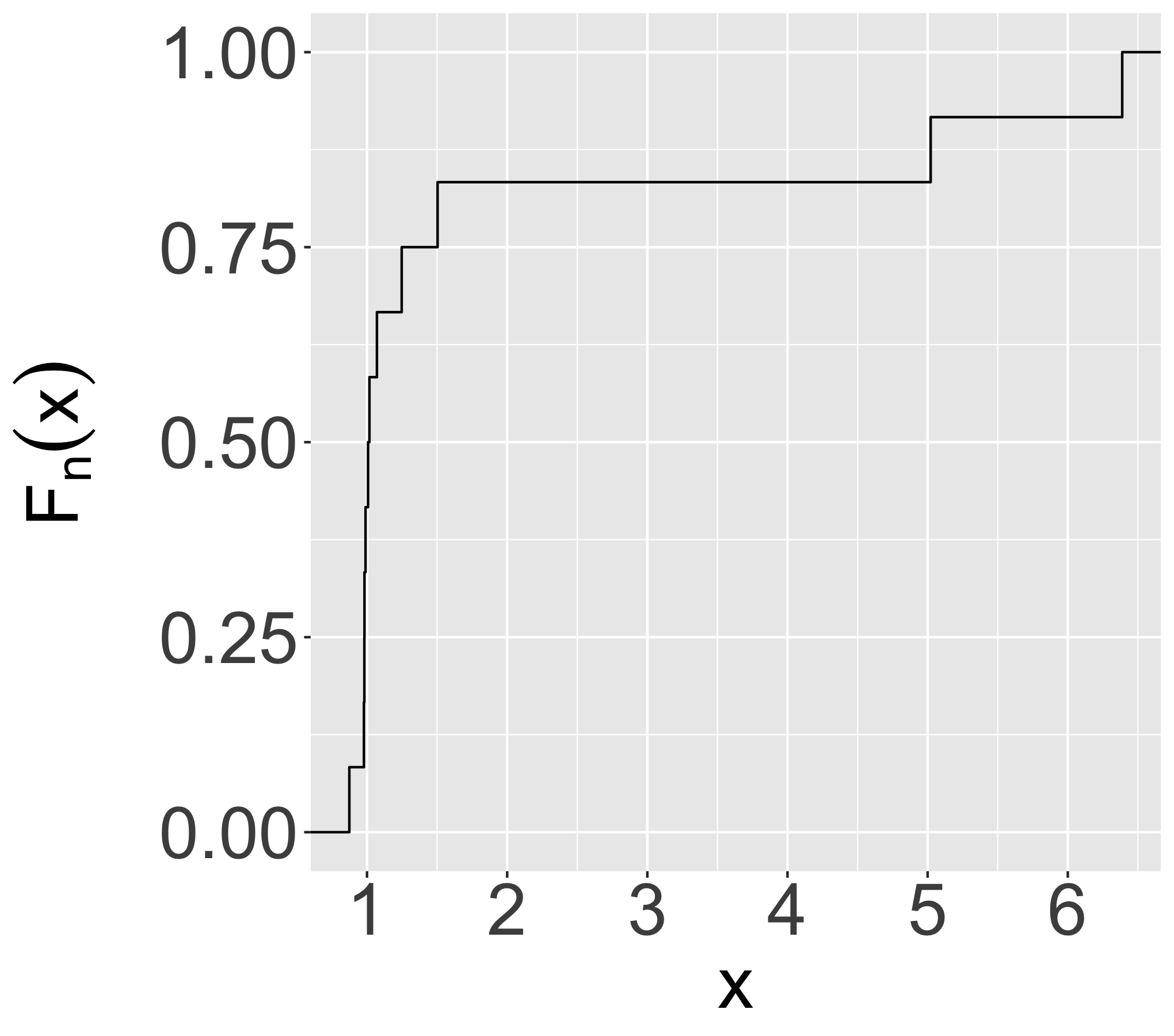}  
	  \caption{Inner envelope}
	\end{subfigure}
	\caption{The empirical cumulative distribution function of ${\rm sd}(\hat{\bm\beta}_{ij,\rm ols})/{\rm sd}(\hat{\bm\beta}_{ij,\rm env})$ (left plot) and ${\rm sd}(\hat{\bm\beta}_{ij,\rm ols})/{\rm sd}(\hat{\bm\beta}_{ij,\rm global})$ (right plot) for each element of $\bm \beta$ matrix from the \textit{iris} dataset.}
	\label{fig: ecdf1}
\end{figure}



\section{Real Data Analysis}\label{Sec: real}

Cardiovascular disease is a major cause of morbidity and mortality in patients with type 2 diabetes. The Action to Control Cardiovascular Risk in Diabetes (ACCORD) study in 2008 aimed to determine if the rate of cardiovascular disease (CVD) can be reduced in people with type 2 diabetes using intensive glycemic control, intensive blood pressure control, and multiple lipid management. Specifically, the ACCORD study randomized patients to either intensive glycemic control (HbA1c <6\%) or standard glycemic control (HbA1c 7-7.9\%). Participant were between the age of 40 to 82 who have been involved in the study for 2 to 7 years. Apart from diabetes, they also had a high risk of heart attack and stroke where each participant either has at least two risk factors for CVD and diabetes or has been diagnosed with CVD before the start of the study.

For our analysis, we are interested in investigating whether intensive glycemic control is associated with better outcomes after adjusting for baseline covariates. There are in total 6766 observations, 9 responses $\Y\in\mathbb R^9$ measuring the efficacy of intensive glycemic contorl, and 3 regressors $\X\in\mathbb R^3$. The response variables are treatment satisfaction, depression scale, aggregate physical activity score, aggregate mental score, symptom and distress score, systolic blood pressure (SBP), diastolic blood pressure, and heart rate. Similar to the simulation studies, we standardized the responses. The predictors are body weight, age and the treatment indicator. 

We first build a multivariate linear model to check if the residuals follows a normal distribution using the Shapiro–Wilk normality test. The $p$-values for each of the responses are significant under both tests, which suggests the normality assumption is violated. Next, we apply our proposed methods and the original envelope method. We remark that the envelope method chose a dimension of $u=9$, which equals the total number of responses and the envelope method may have limited practical value. In contrast, our method found a non-trivial inner envelope structure with dimensions $u = \text{dim}(\mathcal S_1) = 2$, and $d = \text{dim}(\mathcal S_3) = 3$. 

Once the dimension are selected, we assess the out-of-sample RMSE where we split 80\% of the dataset into training data and 20\% into test data, with 5416 and 1350 observations. Like Section \ref{Sec: simu}, We calculate the predictive RMSE of the test data using different machine learning algorithms where the responses are either the original responses or the responses utilizing the inner envelope structure. The results are shown in Table \ref{tb: real2}.

\begin{table}[!htb]
	\centering
	\begin{tabular}{||c c  c c ||} 
	\hline
	 & XGBoost & Random Forest & Linear \\ [0.5ex] 
	\hline\hline
	Using inner envelope structure  &0.9918 & 0.9854 &  0.9834  \\ 
	\hline
	Using original response & 0.9949 & 0.9856 & 0.9835 \\
	\hline
   \end{tabular}
   \caption{Predictive root-mean squared error (RMSE) for the test data in the ACCORD study}
   \label{tb: real2}
\end{table}

Comparing the rows of Table \ref{tb: real2}, the predictive RMSE is slightly smaller for all the methods if we use the inner envelope structure. 
Also, comparing the columns of Table \ref{tb: real2}, the linear model seems to have a better predictive performance compared to more complex methods. Given this, we compare the estimation of the regression parameter $\bm\beta$ supposing that the underlying model is linear. Specifically, we use the nonparametric bootstrap to obtain estimates of the standard errors of the OLS estimator and our globally efficient semiparametric inner envelope estimator of $\bm \beta$. The point estimates, bootstrap standard errors and $p$-values for the parameters in a linear model corresponding to the treatment are given in Table \ref{real_data1}.

\begin{table}[!htb]
	\centering
	\caption{The point estimates, bootstrap standard errors and $p$-values for the regression parameter corresponding treatment for the ACCORD study. Asterisks correspond to p-values that are less than or equal to 0.05.}
		\begin{tabular}{c|ccc|ccc}
			\hline
			&\multicolumn{3}{c}{Our Method}&\multicolumn{3}{c}{OLS} \\
			\hline
			Corresponding to Treatment& $\bm{\hat\beta}$ & $\hat {\mathrm{SE}}$   & $p$-value & $\bm{\hat\beta}$ & $\hat {\mathrm{SE}}$  &  $p$-value \\
			\hline
			Treatment Satisfaction &  0.029 & 0.023  & 0.22 & 0.027 & 0.023 & 0.24 \\ 
			Depression Scale & 0.054 & 0.026 &  0.04$^*$ & 0.054 & 0.026 & 0.04$^*$ \\  
      		Physical Score & 0.016 & 0.024  &  0.50 & 0.014 & 0.024 & 0.55 \\ 
			Mental Score &  -0.045 & 0.020  &  0.03$^*$ & -0.043 & 0.026 & 0.10 \\ 
			Interference Score & -0.011 & 0.015 & 0.46 & -0.009 & 0.027 & 0.72 \\ 
			Symptom \& Distress Score & 0.039 & 0.024 & 0.11 & 0.039 & 0.025 & 0.13 \\ 
			SBP & -0.012 & 0.005  & 0.02$^*$ & -0.034 & 0.022 & 0.12 \\ 
			DBP & 0.004 & 0.024 & 0.86 & 0.002 & 0.025 & 0.94 \\ 
			Heart Rate & 0.039 & 0.020  &  0.05$^*$ & 0.038 & 0.023 & 0.11 \\ 
			\hline	
			
		\end{tabular}
	\label{real_data1}
\end{table}

Our method found that the treatment had a significant effect on mental score, SBP and heart rate. In contrast, OLS found that the treatment was not significantly associated with these responses. This may be because our method is more efficient than the OLS, leading to more power to reject the null hypothesis of no effect. Indeed,  the estimated standard errors of our method are generally smaller compared to that from the OLS estimator. For example, the mean ratio of the standard error using our method over that using the OLS estimator is 1.55. 


\section{Summary and Discussion}
In this paper, we proposed a semiparametric approach to the inner envelope, a dimension reduction method proposed by \cite{su2012inner} for the linear multivariate regression models. We derived the orthogonal nuisance tangent space, score function and efficient scores to estimate the inner envelope, all without having to make parametric modeling assumptions between the response, the covariates, and or the error terms. 
We also prooposed a simple GMM estimator from a set of estimating equations based on moment conditions. 


We briefly take a moment to highlight some key limitations of our work. First, we assume throughout the theoretical results that the dimension of the inner envelope is fixed, even though in practice it is estimated from data. While we have shown that our bootstrap-based approach can reliably select the dimension numerically, we leave it as future work to derive the statistical properties of the estimated inner envelope with an estimated dimension. Second, we have not explored joint dimension reductions of both the responses and the predictors, which may bring further efficiency gains on the parameters of interest.

\section{Proofs}\label{sec: proof}
\subsection{Proof of Lemma \ref{lemma: repar}}
Take any matrix $\mathbf A_1\in\mathbb R^{k\times \tau}$ that is full column rank and spans a $\tau$-dimensional subspace $\mathcal A$ in a $k$-dimensional space. We can represent $\mathbf A_1$ by $\mathbf A_1=\left(\mathbf A_{u}^{T}, \mathbf A_{l}^{T}\right)^{T},$ where $\mathbf A_{u} \in \mathbb{R}^{\tau \times \tau}$ and ${\mathbf A}_{l} \in \mathbb{R}^{(k-\tau) \times \tau}$.
Without loss of generality, we can assume that $\mathbf A_{u}$ is invertible. Then, $\mathcal A$ is uniquely represented by $\mathbf A_{\rm{r e p r}}=\mathbf A_1 \mathbf A_{u}^{-1}=\big\{\mathbf{I}_{u},\left(\mathbf A_{l} \mathbf A_{u}^{-1}\right)^{T}\big\}^{T}$, and the lower $(p-u) \times u$ submatrix $\mathbf A_{l} \mathbf A_{u}^{-1}$ uniquely parameterizes $\mathcal A$. Let $\mathbf a=\text{vec}(\mathbf A_{l} \mathbf A_{u}^{-1})$ denote the vector concatenation of the lower part of $\mathbf A_{\rm repr}$. Since the mapping between $\mathbf A_{repr}$ and $\mathcal A$ is one-to-one, there exists a one-to-one mapping $\psi_1$ such that $\mathbf a = \psi_1(\mathcal A)$. 

Because ${\mathbf A}_{\rm{r e p r}}$ is not a semi-orthogonal matrix, we use Gram-Schmidt procedure to obtain a unique semi-orthogonal matrix $\mathbf A$ from $\mathbf{A}_{\rm r e p r}$. Hence, there exists $\tilde\psi$ such that $\mathbf A = \psi_2(\mathcal A)$. One can show that $\psi_2$ is also a one-to-one mapping. The relationship between $\mathbf a$, $\mathbf A_{\rm repr}$ and $\mathbf A$ are shown in the Figure \ref{fig: relationship}. Notice that if we decompose ${\mathbf A}$ as $\mathbf {A} = ({\mathbf A}_{\text{up}}\T, \mathbf{A}_{\text{low}}\T)\T$ as before, we also have $\mathbf A_{\text{repr}} = \mathbf{A}{\mathbf A}_{\text{up}}^{-1}$. Therefore, we have unique representations of the space $\mathcal A$ by a variational independent parameter $\mathbf a$ and by an orthogonal matrix $\mathbf A$.
\begin{figure}[!htb]
 \small{\[\mathbf a = \begin{pmatrix}
 a_{\tau+1,1}\\
 \vdots\\
 a_{k,1}\\
 \vdots\\
 a_{\tau+1,\tau}\\
 \vdots\\
 a_{k,\tau}
\end{pmatrix}\xrightleftharpoons[\text{Vectorize lower submatrix}]{\text{Stack }\mathbf a\text{ and concatenate } \mathbf I_\tau }\mathbf A_{\rm{repr}}=\begin{pmatrix}
 1 & \cdots  & 0 \\
 \vdots &  \ddots & \vdots \\
 0 & \cdots & 1\\

 a_{\tau+1,1} &  \cdots &  a_{\tau+1,\tau}\\
 \vdots & \ddots & \vdots\\
 a_{k,1} &  \cdots &  a_{k,\tau}
\end{pmatrix}\xrightleftharpoons[\mathbf{A}_{{\rm r e p r}}=\mathbf{A} \mathbf{A}_{\text{up}}^{-1}]{\text{Gram-Schmidt}}{\mathbf A} = \begin{pmatrix}
 {\mathbf A}_{\text{up}}\\
{\mathbf A}_{\text{low}}
\end{pmatrix}\]} 
\caption{Unique parameterization of any space $\mathcal A$.}\label{fig: relationship}
\end{figure}

\subsection{Proof of Theorem \ref{thm: local}}
\begin{proof}
	We prove the consistency and asymptotic normality by checking conditions in Theorem 2.1 (Lemma 6 in the Supplements) and 3.1 (Lemma 7 in the Supplements) from \cite{newey1994large}, and the efficiency by checking the asymptotic variance achieves the semiparametric efficiency bound. We write $X_n = o_p(1)$ if $X_n\overset{p}{\rightarrow}0$ and $X_n = O_p(1)$ if for all $\epsilon > 0$ there exists $M$ such that $\sup_n \mathrm P(\|X_n >M\|)<\epsilon$. Throughout the proof of this theorem, we let $\bm\theta_0$ denote the true value of $\bm\theta$.
	
	
	We firstly prove the $\hat{\bm\theta}$ obtained by solving equation \eqref{eq:Shat_local} satisfies $\hat{\bm\theta}\overset{p}{\rightarrow}\bm\theta_0$. 
	
	We prove by checking conditions of Theorem 2.1 in \cite{newey1994large}. Consider the following two functions: 
	{\small $$Q_0(\bm\theta) = \dfrac{1}{2}\left\{\frac{1}{n}\sum_{i = 1}^n {S}_{\rm eff}^*(\Y_i, \X_i, \eta_{1,2,3}^*, {\bm \Delta}_1, {\bm \Delta}_2; {\bm\theta})\right\}^2$$
	$$\hat Q_n(\bm\theta) = \dfrac{1}{2}\left\{\frac{1}{n}\sum_{i = 1}^n \hat{S}_{\rm eff}^*(\Y_i, \X_i, \hat\eta_{1,2,3}^*, \hat{\bm \Delta}_1, \hat{\bm \Delta}_2; {\bm\theta})\right\}^2.$$}
	Let $\hat{\bm\theta}_0$ be a minimizer of $Q_0(\bm\theta)$. By regularity condition (B2), $\hat{\bm\theta}_0$ is unique and $\hat{\bm\theta}_0\overset{p}{\rightarrow}\bm\theta_0$. Also, $\hat{\bm\theta}$ is the minimizer of $\hat Q_n(\bm\theta)$. Because the parameter space $\bm\Theta$ is compact and $Q_0(\bm\theta)$ is continuous, in order to apply Theorem 2.1 in \cite{newey1994large}, we only need to show condition (iv) holds. That is, $\hat Q_n(\bm\theta)$ converges uniformly in probability to $Q_0(\bm\theta)$ for $\bm\theta\in\bm\Theta$. By Lemma 2.8 in \cite{newey1994large}, we only need to show (1) $\hat Q_n(\bm\theta)\overset{p}{\rightarrow}Q_0(\bm\theta)$ for any $\bm\theta\in\bm\Theta$; and (2) $\hat Q_n(\bm\theta)$ is stochastic equicontinuous. By the continuous mapping theorem, we only need to show $$\hat{S}_{\rm eff}^*(\Y_i, \X_i, \eta_{1,2,3}^*, \hat{\bm \Delta}_1, \hat{\bm \Delta}_2; {\bm\theta})\overset{p}{\rightarrow}{S}_{\rm eff}^*(\Y_i, \X_i, \eta_{1,2,3}^*, {\bm \Delta}_1, {\bm \Delta}_2; {\bm\theta}).$$ 
	Since $\hat{S}_{\rm eff}^*(\Y_i, \X_i, \eta_{1,2,3}^*, \hat{\bm \Delta}_1, \hat{\bm \Delta}_2)$ is a continuous function of $\hat{\bm \Delta}_1, \hat{\bm \Delta}_2, \hat{\E}\{{\partial \log{\eta}_1^*}/{\partial ({\bm\Gamma}\T\Y_i)\T}\mid \X_i\}, \hat{\E}\{{\partial \log{\eta}_2^*}/{\partial ({\bm\Gamma_0\T}\Y_i)\T}\mid \B_0\T\bm{\Gamma}_0\T\Y_i, \X_i\}$ and $\hat{\E}\{{\partial \log{\eta}_3^*}/{\partial ({\B_0\T\bm\Gamma_0\T}\Y_i)\T}\}$; and by the properties of nonparametric regression, $\hat{\bm\Delta}_1\overset{p}{\rightarrow}\bm\Delta_1,$ $\hat{\bm\Delta}_2(\bm\theta)\overset{p}{\rightarrow}\bm\Delta_2(\bm\theta),$ where $\hat{\bm\Delta}_{i1} = \Y_i - \hat\E(\Y_i\mid\X_i)$, $\hat{\bm\Delta}_{i2}({\bm\theta}) = \mathbf P_{{\bm\Gamma}_0\B}\{\hat\E(\Y_i\mid {\B}_0{\bm\Gamma}_0\T\Y_i,\X_i) - \hat\E(\Y_i\mid {\B}_0{\bm\Gamma}_0\T\Y_i)\}$, and 
	{\small\begin{align*}
		&\hat{\E}\left\{\frac{\partial \log{\eta}_1^*}{\partial ({\bm\Gamma}\T\Y_i)\T}\bmid \X_i\right\}\overset{p}{\rightarrow} {\E}\left\{\frac{\partial \log{\eta}_1^*}{\partial ({\bm\Gamma}\T\Y_i)\T}\bmid \X_i\right\},\\
		&\hat{\E}\left\{\frac{\partial \log{\eta}_2^*}{\partial ({\bm\Gamma_0\T}\Y_i)\T}\bmid \B_0\T\bm{\Gamma}_0\T\Y_i, \X_i\right\}\overset{p}{\rightarrow}\E\left\{\frac{\partial \log{\eta}_2^*}{\partial ({\bm\Gamma_0\T}\Y_i)\T}\bmid \B_0\T\bm{\Gamma}_0\T\Y_i, \X_i\right\},\\
		&\hat{\E}\left\{\frac{\partial \log{\eta}_3^*}{\partial ({\B_0\T\bm\Gamma_0\T}\Y_i)\T}\right\}\overset{p}{\rightarrow}{\E}\left\{\frac{\partial \log{\eta}_3^*}{\partial ({\B_0\T\bm\Gamma_0\T}\Y_i)\T}\right\};
	\end{align*}}
	by continuous mapping theorem, $$\hat{S}_{\rm eff}^*(\Y_i, \X_i, \eta_{1,2,3}^*, \hat{\bm \Delta}_1, \hat{\bm \Delta}_2; {\bm\theta})\overset{p}{\rightarrow}{S}_{\rm eff}^*(\Y_i, \X_i, \eta_{1,2,3}^*, {\bm \Delta}_1, {\bm \Delta}_2; {\bm\theta})$$ 
	for any $\bm\theta\in\bm\Theta$. To prove the stochastic equicontinuity of $\hat Q_n(\bm\theta)$, by Lemma 2.9 of  \cite{newey1994large}, we only need to show $\forall \bm\theta_1,\bm\theta_2\in\bm\Theta$, there exists $B_n = O_p(1)$ such that $|\hat Q_n(\bm\theta_1) - \hat Q_n(\bm\theta_2)|\leq B_n\|\bm\theta_1-\bm\theta_2\|$. By regularity conditions (A5) and (B1), $\hat Q_n(\bm\theta)$ is differentiable. By mean value theorem, there exists $\bar{\bm\theta}\in\bm\Theta$ such that $|\hat Q_n(\bm\theta_1) - \hat Q_n(\bm\theta_2)|\leq \|\hat Q_n'(\bar{\bm\theta})\|\|\bm\theta_1-\bm\theta_2\|$. By Theorem 2.9 in \cite{li2007nonparametric} and regularity condition (B3), $\partial \hat{\bm\Delta}_2(\bm\theta)/\partial \bm\theta - \partial{\bm\Delta}_2(\bm\theta)/\partial \bm\theta = O_p(n^{1/2}h^{1+(p+r-u-d)/2}) = o_p(1)$. Because $\log \eta_{1,2,3}^*$ are twice differentiable, we have  
	{\small\begin{align*}
		&\hat{\E}\left\{\frac{\partial^2 \log{\eta}_1^*}{\partial\bm\theta\partial ({\bm\Gamma}\T\Y_i)\T}\bmid \X_i\right\}\overset{p}{\rightarrow} {\E}\left\{\frac{\partial^2 \log{\eta}_1^*}{\partial\bm\theta\partial ({\bm\Gamma}\T\Y_i)\T}\bmid \X_i\right\},\\
		&\hat{\E}\left\{\frac{\partial^2 \log{\eta}_2^*}{\partial\bm\theta\partial ({\bm\Gamma_0\T}\Y_i)\T}\bmid \B_0\T\bm{\Gamma}_0\T\Y_i, \X_i\right\}\overset{p}{\rightarrow}\E\left\{\frac{\partial^2 \log{\eta}_2^*}{\partial\bm\theta\partial ({\bm\Gamma_0\T}\Y_i)\T}\bmid \B_0\T\bm{\Gamma}_0\T\Y_i, \X_i\right\},\\
		&\hat{\E}\left\{\frac{\partial^2 \log{\eta}_3^*}{\partial\bm\theta\partial ({\B_0\T\bm\Gamma_0\T}\Y_i)\T}\right\}\overset{p}{\rightarrow}{\E}\left\{\frac{\partial^2 \log{\eta}_3^*}{\partial\bm\theta\partial ({\B_0\T\bm\Gamma_0\T}\Y_i)\T}\right\}.
	\end{align*}}
	Therefore, by continuous mapping theorem, $\hat{Q}'_n(\bm\theta)\overset{p}{\rightarrow}{Q}'_0(\bm\theta)$. Hence, $B_n = \sup_{\bm\theta\in\bm\Theta}\hat Q_n'({\bm\theta}) \leq \sup_{\bm\theta\in\bm\Theta} Q_0'({\bm\theta})+ \sup_{\bm\theta\in\bm\Theta}|Q_0'({\bm\theta})-\hat Q_n'({\bm\theta})|$. Since $Q_0'({\bm\theta})$ is a continuous function in a compact set, $\sup_{\bm\theta\in\bm\Theta} Q_0'({\bm\theta})$ is $O_p(1)$. Suppose $\sup_{\bm\theta\in\bm\Theta}|Q_0'({\bm\theta})-\hat Q_n'({\bm\theta})| = \infty$, then there exists a sequence $\{\bm\theta_{k}\}$ such that $|Q_0'({\bm\theta_k})-\hat Q_n'({\bm\theta_k})|\rightarrow\infty$. That is, for any $C > 0$, there exists $N\in\mathbb N^{+}$ such that for any $n\geq N$, $|Q_0'({\bm\theta_n})-\hat Q_n'({\bm\theta_n})|> C$, which contradicts with the fact that $\hat Q_n'({\bm\theta_n})\overset{p}{\rightarrow}Q_0'({\bm\theta_n})$. Hence, $\sup_{\bm\theta\in\bm\Theta}|Q_0'({\bm\theta})-\hat Q_n'({\bm\theta})| = o_p(1)$. Therefore, all the conditions in Theorem 2.1 from \cite{newey1994large} hold, and we have $\hat{\bm\theta}\overset{p}{\rightarrow}\bm\theta_0$.
	
	Then, in order to show the asymptotic normality of $\hat{\bm\theta}$, we only need to verify conditions (i)--(v) in Theorem 3.1 from \cite{newey1994large}. Conditions (i)--(ii) are already satisfied. We then prove $\sqrt n{\partial \hat Q_n(\bm\theta_0)}/{\partial \bm\theta}$ is asymptotically normal.
	
	Because
	{\small\begin{align*}
		&\sqrt n\frac{\partial \hat Q_n(\bm\theta_0)}{\partial \bm\theta} = \frac{1}{\sqrt n}\sum_{i = 1}^n \hat{S}_{\rm eff}^*(\Y_i, \X_i, \hat\eta_{1,2,3}^*, \hat{\bm \Delta}_1, \hat{\bm \Delta}_2; {\bm\theta_0})\cdot \frac{1}{n}\sum_{i=1}^n\dfrac{\partial}{\partial\bm\theta}\hat S_{\rm eff}^*(\Y_i,\X_i, \hat\eta_{1,2,3}^*, \hat{\bm \Delta}_1, \hat{\bm \Delta}_2;\bm\theta_0),
	\end{align*}}
	and 
	{\small$$\frac{1}{n}\sum_{i=1}^n\dfrac{\partial}{\partial\bm\theta}\hat S_{\rm eff}^*(\Y_i,\X_i, \hat\eta_{1,2,3}^*, \hat{\bm \Delta}_1, \hat{\bm \Delta}_2;\bm\theta_0)\overset{p}{\rightarrow}\E\left\{\dfrac{\partial}{\partial\bm\theta}S_{\rm eff}^*(\Y,\X, \eta_1^*, \eta_2^*, \eta_3^*;\bm\theta_0)\right\},$$}
	by Slutsky's Theorem, we only need to show $n^{-1/2}\sum_{i = 1}^n \hat{S}_{\rm eff}^*(\Y_i, \X_i, \hat\eta_{1,2,3}^*, \hat{\bm \Delta}_1, \hat{\bm \Delta}_2; {\bm\theta_0})$ converges to a normal distribution. Also, because 
	{\small$$\frac{1}{\sqrt n}\sum_{i = 1}^n {S}_{\rm eff}^*(\Y_i, \X_i, \eta_{1,2,3}^*, {\bm \Delta}_1, {\bm \Delta}_2; {\bm\theta_0}) \overset{d}{\rightarrow}\mathcal N\bigg[\bm 0, \E\big\{S_{\rm eff}^*(\Y_i,\X_i, \eta_1^*, \eta_2^*, \eta_3^*;\bm\theta_0)^{\otimes2}\big\}\bigg],$$}
    we only need to show 
	{\small$$\frac{1}{\sqrt n}\sum_{i = 1}^n \hat{S}_{\rm eff}^*(\Y_i, \X_i, \hat\eta_{1,2,3}^*, \hat{\bm \Delta}_1, \hat{\bm \Delta}_2; {\bm\theta_0})-\frac{1}{\sqrt n}\sum_{i = 1}^n {S}_{\rm eff}^*(\Y_i, \X_i, \eta_{1,2,3}^*, {\bm \Delta}_1, {\bm \Delta}_2; {\bm\theta_0})\overset{p}{\rightarrow}0.$$}
	The score function $S_{\text{eff}}^* = (S_{\text{eff}, \bm\gamma}^*, S_{\text{eff}, \mathbf b}^*)$ has two components. For simplicity, we only show the convergence of the second component $S_{\text{eff}, \mathbf b}^*$. The convergence of the first component can be proved using the same technique. Let $\bm\Gamma$ and $\mathbf B$ denote the orthogonal basis derived from $\bm\theta_0$. Then,
	{\small\begin{align*}
	    &\frac{1}{\sqrt n}\sum_{i = 1}^n \hat{S}_{\rm{eff},\mathbf b}^*(\Y_i, \X_i, \hat\eta_{1,2,3}^*, \hat{\bm \Delta}_1, \hat{\bm \Delta}_2; {\bm\theta_0}) - \frac{1}{\sqrt n}\sum_{i = 1}^n {S}_{\rm{eff},\mathbf b}^*(\Y_i, \X_i, \eta_{1,2,3}^*, {\bm \Delta}_1, {\bm \Delta}_2; {\bm\theta_0})\\
	    =&\dfrac{1}{\sqrt n}\sum_{i=1}^{n}\text{vec}\T\bigg[{\bm\Gamma}_0\T\hat{\bm\Delta}_{i2}({\bm\theta_0})\bigg\{\dfrac{\partial \log{\eta}_3^*}{\partial (\B_0\T{\bm\Gamma}_0\T\Y_i)\T} - \hat\E\bigg(\dfrac{\partial \log\eta_3^*}{\partial (\B_0\T{\bm\Gamma}_0\T\Y_i)\T}\bigg) \bigg\}\bigg] \dfrac{\partial \text{vec}({\B}_0)}{\partial {\mathbf b}\T}\\
	    &\quad -\dfrac{1}{\sqrt n}\sum_{i=1}^{n}\text{vec}\T\bigg[{\bm\Gamma}_0\T{\bm\Delta}_{i2}({\bm\theta_0})\bigg\{\dfrac{\partial \log{\eta}_3^*}{\partial (\B_0\T{\bm\Gamma}_0\T\Y_i)\T} - \E\bigg(\dfrac{\partial \log\eta_3^*}{\partial (\B_0\T{\bm\Gamma}_0\T\Y_i)\T}\bigg) \bigg\}\bigg] \dfrac{\partial \text{vec}({\B}_0)}{\partial {\mathbf b}\T}\\
	    =& \dfrac{1}{\sqrt n}\sum_{i=1}^{n}\text{vec}\T\bigg[{\bm\Gamma}_0\T\left\{\hat{\bm\Delta}_{i2}({\bm\theta_0}) - {\bm\Delta}_{i2}({\bm\theta_0})\right\}\bigg\{\dfrac{\partial \log\eta_3^*}{\partial (\B_0\T{\bm\Gamma}_0\T\Y_i)\T} - \E\bigg(\dfrac{\partial \log\eta_3^*}{\partial (\B_0\T{\bm\Gamma}_0\T\Y_i)\T}\bigg) \bigg\}\bigg] \dfrac{\partial \text{vec}({\B}_0)}{\partial {\mathbf b}\T}\\
	&\quad+ \dfrac{1}{\sqrt n}\sum_{i=1}^{n}\text{vec}\T\bigg[{\bm\Gamma}_0\T\left\{\hat{\bm\Delta}_{i2}({\bm\theta_0}) - {\bm\Delta}_{i2}({\bm\theta_0})\right\}\bigg\{\E\bigg(\dfrac{\partial \log\eta_3^*}{\partial (\B_0\T{\bm\Gamma}_0\T\Y_i)\T}\bigg) - \hat\E\bigg(\dfrac{\partial \log\eta_3^*}{\partial (\B_0\T{\bm\Gamma}_0\T\Y_i)\T}\bigg) \bigg\}\bigg] \dfrac{\partial \text{vec}({\B}_0)}{\partial{\mathbf b}\T}\\
	&\quad + \dfrac{1}{\sqrt n}\sum_{i=1}^{n}\text{vec}\T\bigg[{\bm\Gamma}_0\T{\bm\Delta}_{i2}({\bm\theta_0})\bigg\{\E\bigg(\dfrac{\partial \log\eta_3^*}{\partial (\B_0\T{\bm\Gamma}_0\T\Y_i)\T}\bigg) - \hat\E\bigg(\dfrac{\partial \log\eta_3^*}{\partial (\B_0\T{\bm\Gamma}_0\T\Y_i)\T}\bigg) \bigg\}\bigg] \dfrac{\partial \text{vec}({\B}_0)}{\partial{\mathbf b}\T} 
	\end{align*}}
By Lemma 5, the first term can be bounded by $O_p(h^2 + n^{-1/2}h^{p + r-u-d}\log n + n^{1/2}h^4)$. Under regularity condition \eqref{assump: bandwidth}, the second term is $o_p(1)$.

By Theorem 2.6 in \cite{li2007nonparametric}, under regularity condition \eqref{assump: finite},  
{\small $$\sup_{\X_i,\Y_i}|\hat{\bm\Delta}_{i2}({\bm\theta_0}) - {\bm\Delta}_{i2}({\bm\theta_0})| = O_p\bigg\{\left(\frac{\log n}{nh^{p+r-u-d}}\right)^{1/2}+h^2\bigg\}$$}
is $o_p(1)$ under regularity condition \eqref{assump: bandwidth}.

By the central limit theorem,
{\small \begin{align*}
	&\dfrac{1}{\sqrt n}\sum_{i=1}^{n}\text{vec}\T\bigg[\left\{\hat{\bm\Delta}_{i2}({\bm\theta_0}) - {\bm\Delta}_{i2}({\bm\theta_0})\right\}\bigg\{\E\bigg(\dfrac{\partial \log\eta_3^*}{\partial (\B_0\T{\bm\Gamma}_0\T\Y_i)\T}\bigg) - \hat\E\bigg(\dfrac{\partial \log\eta_3^*}{\partial (\B_0\T{\bm\Gamma}_0\T\Y_i)\T}\bigg) \bigg\}\bigg]\dfrac{\partial \text{vec}({\B}_0)}{\partial{\mathbf b}\T}\\
	=&\text{vec}\T\bigg[\frac{1}{ n}\sum_{i=1}^{ n}\big\{\hat{\bm\Delta}_{i2}({\bm\theta_0}) - {\bm\Delta}_{i2}({\bm\theta_0})\big\}\cdot\sqrt n \bigg\{\E\bigg(\dfrac{\partial \log\eta_3^*}{\partial (\B_0\T{\bm\Gamma}_0\T\Y_i)\T}\bigg) - \hat\E\bigg(\dfrac{\partial \log\eta_3^*}{\partial (\B_0\T{\bm\Gamma}_0\T\Y_i)\T}\bigg) \bigg\}\bigg]\dfrac{\partial \text{vec}({\B}_0)}{\partial{\mathbf b}\T}\\
	=&o_p(1)\cdot O_p(1) = o_p(1).
\end{align*}}
Hence, the second term is also $o_p(1)$.

Also,
{\small \begin{align*}
	&\frac{1}{\sqrt n}\sum_{i=1}^{n}{\bm\Delta}_{i2}({\bm\theta}_0)\bigg\{\E\bigg(\dfrac{\partial \log\eta_3^*}{\partial (\B_0\T{\bm\Gamma}_0\T\Y_i)\T}\bigg) - \hat\E\bigg(\dfrac{\partial \log\eta_3^*}{\partial (\B_0\T{\bm\Gamma}_0\T\Y_i)\T}\bigg) \bigg\}\\
	=&o_p(1)\cdot n^{-{1}/{2}}\sum_{i=1}^{n}{\bm\Delta}_{i2}({\bm\theta_0})\\
	=&o_p(1),
\end{align*}}
where the last equation is because ${\bm\Delta}_{i2}$ are i.i.d. mean $\bm 0$ random variables. Hence, the third term is also $o_p(1)$.

Therefore, {\small$$\frac{1}{\sqrt n}\sum_{i = 1}^n \hat{S}_{\rm{eff},\mathbf b}^*(\Y_i, \X_i, \hat\eta_{1,2,3}^*, \hat{\bm \Delta}_1, \hat{\bm \Delta}_2; {\bm\theta_0})-\frac{1}{\sqrt n}\sum_{i = 1}^n {S}_{\rm{eff},\mathbf b}^*(\Y_i, \X_i, \eta_{1,2,3}^*, {\bm \Delta}_1, {\bm \Delta}_2; {\bm\theta_0}) \overset{p}{\rightarrow}0 .$$}
Hence, by Slutsky's Theorem, {\small$$\sqrt n\frac{\partial \hat Q_n(\bm\theta_0)}{\partial \bm\theta}\overset{d}{\rightarrow}\mathcal N\{\bm 0, \mathbf C_2\T\mathbf D_2\mathbf C_2\},$$
	where $$\mathbf C_2 = \E\left\{\dfrac{\partial}{\partial\bm\theta}S_{\rm eff}^*(\Y,\X, \eta_1^*, \eta_2^*, \eta_3^*; \bm\theta_0)\right\}, \hskip .5cm  \mathbf D_2 = \E\left\{S_{\rm eff}^*(\Y_i,\X_i, \eta_1^*, \eta_2^*, \eta_3^*;\bm\theta_0)^{\otimes2}\right\}.$$}
Next, we verify the conditions (iv)--(v) that are related to $\nabla_{\bm\theta\bm\theta}\hat{Q}_n(\bm\theta)$. Notice that 
	\begin{align*}
	   &\nabla_{\bm\theta\bm\theta}\hat{Q}_n(\bm\theta) 
	   \overset{p}{\rightarrow}\E\left\{\dfrac{\partial}{\partial\bm\theta}S_{\rm eff}^*(\Y_i,\X_i, \eta_{1,2,3}^*, {\bm \Delta}_1, {\bm \Delta}_2;\bm\theta)\right\}\cdot\E\left\{\dfrac{\partial}{\partial\bm\theta\T}S_{\rm eff}^*(\Y_i,\X_i, \eta_{1,2,3}^*, {\bm \Delta}_1, {\bm \Delta}_2;\bm\theta)\right\}\\
	   &+\E\left\{S_{\rm eff}^*(\Y_i,\X_i, \eta_{1,2,3}^*, {\bm \Delta}_1, {\bm \Delta}_2;\bm\theta)\right\}\cdot\E\left\{\dfrac{\partial^2}{\partial\bm\theta\partial\bm\theta\T}S_{\rm eff}^*(\Y_i,\X_i, \eta_{1,2,3}^*, {\bm \Delta}_1, {\bm \Delta}_2;\bm\theta)\right\}
	   = H(\bm\theta).
	\end{align*}
	Following the same argument as proving the uniform convergence in probability for $\partial\hat{Q}_n(\bm\theta)/\partial\bm\theta$, because the third order derivative of $\log\eta_{1,2,3}^*$ exists, we have $\nabla_{\bm\theta\bm\theta}\hat{Q}_n(\bm\theta)\overset{p}{\rightarrow}H(\bm\theta)$ uniformly for $\bm\theta\in\bm\Theta$.
	Hence, condition (iv) holds. Because $$\dfrac{1}{n}\sum_{i=1}^n S_{\rm eff}^*(\Y_i,\X_i, \eta_{1,2,3}^*, {\bm \Delta}_1, {\bm \Delta}_2;\bm\theta_0)\overset{p}{\rightarrow}\bm 0,$$we have $H(\bm\theta_0) = \mathbf C_2\mathbf C_2\T$. Since $\mathbf C_2$ is nonsingular, $H(\bm\theta_0)$ is nonsingular. Hence, condition (v) holds. Therefore, by Theorem 3.1 in \cite{newey1994large},
	$$\sqrt n (\bm\theta_0-\bm\theta)\overset{p}{\rightarrow}\mathcal N\{\bm 0, H(\bm\theta_0)^{-1}\mathbf C_2\T\mathbf D_2\mathbf C_2H(\bm\theta_0)^{-1}\} = \mathcal N\{\bm0, \mathbf C_2^{-1}\mathbf D_2(\mathbf C_2\T)^{-1}\}.$$

\end{proof}

\section*{Acknowledgements}
We would like to thank Professor Yanyuan Ma for the helpful discussion for the proofs. The authors were supported by NSF-DMS 1916013 and NIH U24-DK-060990.

\begin{appendix}

	\section{Proof of Lemmas and Theorems}\label{appB}
	From equation \eqref{likelihood}, the nuisance tangent space can be written as
	\begin{equation*}
		\begin{aligned}
		\Lambda &= \{s_1(\bm\Gamma\T\Y\mid \X) + s_2(\mathbf B\T\bm\Gamma_0\T\Y \mid \mathbf B_0\T\bm\Gamma_0\T\Y, \X) + s_3(\mathbf B_0\T\bm\Gamma_0\T\Y) + s_4(\X)\}\\
		&=\Lambda_1 \cup\Lambda_2\cup\Lambda_3\cup\Lambda_4,
		\end{aligned}
		\end{equation*} 
		 where $s_1$, $s_2$, $s_3$ and $s_4$ are measurable functions satisfying 
		 \[\int s_1(\bm\Gamma\T\y\mid \x)\eta_1(\bm\Gamma\T\y\mid\x)d\y = \bm 0 \text{, \hspace{2mm} for any }\x,\]
		 \[\int s_2(\B\T\bm\Gamma_0\T\y\mid \B_0\T\bm\Gamma_0\T\y,\x)\eta_2(\B\T\bm\Gamma_0\T\y\mid\B_0\T\bm\Gamma_0\T\y,\x)d\y = \bm 0, \quad\text{for any }(\B_0\T\bm\Gamma_0\T\y,\x),\]
		\[\int s_3(\B_0\T\bm\Gamma_0\T\y)d\y = \bm 0,\hspace{1cm} \int s_4(\x)d\x = \bm 0,\]
		and $\Lambda_i$ is the Hilbert space spanned by $s_i$ for $i = 1,\ldots,4$. Therefore, the structure of $\Lambda_i$ has the following representation
		\begin{equation*}
	\begin{aligned}
	&\Lambda_1 = \{\mathbf f(\bm\Gamma\T\Y, \X): \forall \mathbf f\text{ such that } \E(\mathbf f\mid \X) = \bm 0 \},\\
	&\Lambda_2 = \{\mathbf f(\mathbf B\T\bm\Gamma_0\T\Y, \mathbf B_0\T\bm\Gamma_0\T\Y, \X): \forall \mathbf f\text{ such that } \E(\mathbf f\mid \mathbf B_0\T\bm\Gamma_0\T\Y, \X) = \bm 0 \},\\
	&\Lambda_3 = \{\mathbf f(\mathbf B_0\T\bm\Gamma_0\T\Y): \forall \mathbf f\text{ such that } \E(\mathbf f) = \bm 0 \},\\
	&\Lambda_4 = \{\mathbf f(\X): \forall \mathbf f\text{ such that } \E (\mathbf f) = \bm 0 \},
	\end{aligned}
	\end{equation*}
	where $\mathbf f$ is a square integrable function.
	The following lemmas provide an orthogonal decomposition of the nuisance tangent space and the form of the orthogonal nuisance tangent space.
	\begin{lemma}\label{Prop: structure tangent space}
		Under Condition (\ref{eq: semi_inner1}) and (\ref{eq: semi_inner2}), the nuisance tangent space can be written as $$\Lambda = \Lambda_1\oplus\Lambda_2\oplus\Lambda_3\oplus\Lambda_4.$$
	\end{lemma}
	\begin{proof}[Proof of Lemma \ref{Prop: structure tangent space}]
		We prove that $\Lambda_1\perp\Lambda_2\perp\Lambda_3\perp\Lambda_4$. Obviously, $\Lambda_4 \perp (\Lambda_1\cup\Lambda_2\cup\Lambda_3)$. We only prove that $\Lambda_1$, $\Lambda_2$ and $\Lambda_3$ are orthogonal to each other.
	For any $f_1\in \Lambda_1$ and $f_2\in\Lambda_2$,
	\begin{equation*}
	\begin{aligned}
	&\E\{f_1\T(\bm\Gamma\T\Y,\X)f_2(\mathbf B\T\bm\Gamma_0\T\Y, \mathbf B_0\T\bm\Gamma_0\T\Y, \X)\}\\
	=&\E[ \E\{f_1\T(\bm\Gamma\T\Y,\X)f_2(\mathbf B\T\bm\Gamma_0\T\Y, \mathbf B_0\T\bm\Gamma_0\T\Y, \X)\mid \X\}]\\
	=&\E[ \E\{f_1\T(\bm\Gamma\T\Y,\X)\mid \X \}\E\{f_2(\mathbf B\T\bm\Gamma_0\T\Y, \mathbf B_0\T\bm\Gamma_0\T\Y, \X)\mid \X\}]=\bm 0.\\
	\end{aligned}
	\end{equation*} Hence, $\Lambda_1\perp \Lambda_2$. The second equation to the third equation is because $\bm\Gamma\T\Y\indep \bm\Gamma_0\T\Y\mid \X$. The third equation equals to $\bm0$ because $\E(f_1\mid\X) = \bm 0$. For any $f_1\in \Lambda_1$ and $f_3\in \Lambda_3$, \begin{equation*}
	  \begin{aligned}
	  &\E\{f_1\T(\bm\Gamma\T\Y,\X)f_3(\mathbf B_0\T\bm\Gamma_0\T\Y)\}\\
	  =&\E[\E\{f_1\T(\bm\Gamma\T\Y,\X)f_3(\mathbf B_0\T\bm\Gamma_0\T\Y)\mid \X\}]\\
	  =&\E[\E\{f_1\T(\bm\Gamma\T\Y,\X)\mid \X\}\E\{f_3(\mathbf B_0\T\bm\Gamma_0\T\Y)\mid \X\}=\bm 0.
	  \end{aligned}
	  \end{equation*} 
	  Thus, $\Lambda_1\perp \Lambda_3$. The second to the third equation is again because of Condition \ref{eq: semi_inner1}. For any $f_2\in \Lambda_2$ and$f_3\in\Lambda_3$, 
	  \begin{equation*}
	  \begin{aligned}
	  &\E\{f_2\T(\mathbf B\T\bm\Gamma_0\T\Y, \mathbf B_0\T\bm\Gamma_0\T\Y, \X)f_3(\mathbf B_0\T\bm\Gamma_0\T\Y)\}\\
	  =&\E[ \E\{f_2\T(\mathbf B\T\bm\Gamma_0\T\Y, \mathbf B_0\T\bm\Gamma_0\T\Y, \X)f_3(\mathbf B_0\T\bm\Gamma_0\T\Y)\mid\mathbf B_0\T\bm\Gamma_0\T\Y, \X\}]\\
	  =&\E[ \E\{f_2\T(\mathbf B\T\bm\Gamma_0\T\Y, \mathbf B_0\T\bm\Gamma_0\T\Y, \X)\mid\mathbf B_0\T\bm\Gamma_0\T\Y, \X\}f_3(\mathbf B_0\T\bm\Gamma_0\T\Y)] = \bm 0
	  \end{aligned}
	  \end{equation*} The last equation equals to 0 because $\E(f_2\mid\mathbf B_0\T\bm\Gamma_0\T\Y, \X) = 0$. Hence, $\Lambda_i$, $i = 1,\ldots,4$ are orthogonal to each other. Therefore, $$\Lambda =\Lambda_1\oplus\Lambda_2\oplus\Lambda_3\oplus\Lambda_4.$$
	\end{proof}
	\begin{lemma}\label{prop: orthogonal tangent space}
		The orthogonal nuisance tangent space	$\Lambda^\perp = \Lambda_1^\perp\cap\Lambda_2^\perp\cap\Lambda_3^\perp\cap\Lambda_4^\perp$ where 
		\begin{equation*}
		\begin{aligned}
		&\Lambda_{1}^\perp = \{\mathbf f(\Y,\X): \E(\mathbf f\mid \bm\Gamma\T\Y,\X) \text{ is a function of }\X \},\\
		&\Lambda_{2}^\perp = \{\mathbf f(\Y,\X): \E(\mathbf f\mid\mathbf B\T\bm\Gamma_0\T\Y, \mathbf B_0\T\bm\Gamma_0\T\Y, \X)\text{ is a function of } \mathbf B_0\T\bm\Gamma_0\T\Y, \X\},\\
		&\Lambda_{3}^\perp = \{\mathbf f(\Y,\X): \E(\mathbf f\mid \mathbf B_0\T\bm\Gamma_0\T\Y) = \bm 0 \},\\
		&\Lambda_{4}^\perp = \{\mathbf f(\Y,\X): \E(\mathbf f\mid\X) = \bm 0\}.
		\end{aligned}
		\end{equation*}
	\end{lemma}
	\begin{proof}[Proof of Lemma \ref{prop: orthogonal tangent space}]
		From the structure of $\Lambda$, we have  Our conjecture for the orthogonal complements for $\Lambda_i$, $i = 1,\ldots,4$ are 
	  \begin{equation*}
	  \begin{aligned}
	  &\Lambda_{1,\text{conj}}^\perp = \{f(\Y,\X): \E(f\mid \bm\Gamma\T\Y,\X) \text{ is a function of }\X \},\\
	  &\Lambda_{2,\text{conj}}^\perp = \{f(\Y,\X): \E(f\mid\mathbf B\T\bm\Gamma_0\T\Y, \mathbf B_0\T\bm\Gamma_0\T\Y, \X)\text{ is a function of } \mathbf B_0\T\bm\Gamma_0\T\Y, \X\},\\
	  &\Lambda_{3,\text{conj}}^\perp = \{f(\Y,\X): \E(f\mid \mathbf B_0\T\bm\Gamma_0\T\Y)=\bm 0 \}\\
	  &\Lambda_{4,\text{conj}}^\perp = \{f(\Y,\X): \E(f\mid\X) = \bm 0\}.
	  \end{aligned}
	  \end{equation*}
	  We only prove our conjecture is true for $\Lambda_1^\perp$. 
	
	  For any $f_1^\perp \in \Lambda_{1,\text{conj}}^\perp$ and $f_1 \in \Lambda_1$,
	\begin{equation*}
	\begin{aligned}
	\E\{f_1\T f_1^\perp\} &= \E[\E\{f_1\T f_1^\perp\mid \bm\Gamma\T\Y,\X \}]\\
	&= \E[f_1\T\E\{f_1^\perp\mid \bm\Gamma\T\Y,\X\}]\\
	&=\E\{f_1\T a(\X)\}\\
	&= \E\{\E(f_1\T\mid\X)a(\X)\} = 0.
	\end{aligned}
	\end{equation*}
	Hence, we have $\Lambda_{1,\text{conj}}^\perp \perp \Lambda_1$. Then, we show that for any $f\in \Lambda_1^\perp$, $f$ satisfies $\E(f\mid\bm\Gamma\T\Y,\X) = a(\X)$ for some function $a(\X)$.
	
	Firstly, define $g = \E(f\mid\bm\Gamma\T\Y,\X) - \E(f\mid \X)$. Clearly, $g\in\Lambda_1$. Thus, 
	\begin{equation*}
	\begin{aligned}
	0 = \E(g\T f) &= \E\{g\T\E(f\mid \bm\Gamma\T\Y,\X)\}\\
	&= \E(g\T g) + \E\{g\T\E(f\mid \X)\} = \E(g\T g).
	\end{aligned}
	\end{equation*}
	Therefore, $g = \bm 0$, which implies $\E(f\mid\bm\Gamma\T\Y,\X) = \E(f\mid \X) = a(\X)$. We proved that $\Lambda_{1,\text{conj}} = \Lambda_1$.
	
	Validity of the other three follows the same technique.
	\end{proof}
	 \begin{customthm}{3}\label{prop: GMM}
		The functions $\mathbf f_1(\bm\B_0\T\bm\Gamma_0\T\Y, \X)$ and $\mathbf f_2(\bm\Gamma\T\Y, \bm\B_0\T\bm\Gamma_0\T\Y)$ defined above belongs to the orthogonal nuisance tangent space $\Lambda^\perp$. 
	\end{customthm}
	\begin{proof}[Proof of Lemma \ref{prop: GMM}]
		Recall that 
	 \begin{equation*}
		f_1(\bm\B_0\T\bm\Gamma_0\T\Y, \X) = \{g_1(\bm\B_0\T\bm\Gamma_0\T\Y) - \E g_1\}\{h_1(\X) - \E h_1\} = \bm 0,
	\end{equation*}
	\begin{equation*}
		f_2(\bm\Gamma\T\Y, \bm\B_0\T\bm\Gamma_0\T\Y) = \{g_2(\bm\Gamma\T\Y) - \E g_2\}\{h_2(\bm\B_0\T\bm\Gamma_0\T\Y) - \E h_2\} = \bm 0.
	\end{equation*}
	We want to show that $f_1(\bm\B_0\T\bm\Gamma_0\T\Y, \X)\in\Lambda^\perp$ and $f_2(\bm\Gamma\T\Y, \bm\B_0\T\bm\Gamma_0\T\Y)\in\Lambda^\perp$.
	\begin{equation*}
		\begin{aligned}
			\E(f_1\mid \bm\Gamma\T\Y, \X) &= \E[\{g_1(\bm\B_0\T\bm\Gamma_0\T\Y) - \E g_1\}\{h_1(\X) - \E h_1\}\mid \bm\Gamma\T\Y, \X]\\
			&= \E\{g_1(\bm\B_0\T\bm\Gamma_0\T\Y) - \E g_1\mid \bm\Gamma\T\Y, \X\}\{h_1(\X) - \E h_1\}\\
			&=\E\{g_1(\bm\B_0\T\bm\Gamma_0\T\Y) - \E g_1\}\{h_1(\X) - \E h_1\}\\
			&= \bm0.
		\end{aligned}
	\end{equation*}
	
	Hence, $f_1\in\Lambda_1^\perp$. The second to the third equation is because $\bm\B_0\T\bm\Gamma_0\T\Y \indep (\bm\Gamma\T\Y, \X)$.
	\begin{equation*}
		\begin{aligned}
			\E(f_1\mid \bm\Gamma_0\T\Y, \X) &= \E[\{g_1(\bm\B_0\T\bm\Gamma_0\T\Y) - \E g_1\}\{h_1(\X) - \E h_1\}\mid \bm\Gamma_0\T\Y, \X]\\
			&= \{g_1(\bm\B_0\T\bm\Gamma_0\T\Y) - \E g_1\}\{h_1(\X) - \E h_1\}
		\end{aligned}
	\end{equation*}
	is a function of $\bm\B_0\T\bm\Gamma_0\T\Y$ and $\X$. Therefore, $f_1\in\Lambda_2^\perp$. Notice that we also have 
	\begin{equation*}
		\begin{aligned}
			\E(f_1\mid \B_0\T\bm\Gamma_0\T\Y) &= \E[\{g_1(\bm\B_0\T\bm\Gamma_0\T\Y) - \E g_1\}\{h_1(\X) - \E h_1\}\mid \B_0\T\bm\Gamma_0\T\Y]\\
			&= \{g_1(\bm\B_0\T\bm\Gamma_0\T\Y) - \E g_1\}\E\{h_1(\X) - \E h_1\mid \B_0\T\bm\Gamma_0\T\Y\}\\
			&= \{g_1(\bm\B_0\T\bm\Gamma_0\T\Y) - \E g_1\}\E\{h_1(\X) - \E h_1\}\\
			&= \bm 0,
		\end{aligned}
	\end{equation*}
	and
	\begin{equation*}
		\begin{aligned}
			\E(f_1\mid \X) &= \E[\{g_1(\bm\B_0\T\bm\Gamma_0\T\Y) - \E g_1\}\{h_1(\X) - \E h_1\}\mid \X]\\
			&= \E[\{g_1(\bm\B_0\T\bm\Gamma_0\T\Y) - \E g_1\mid \X\}\{h_1(\X) - \E h_1\}\\
			&= \E\{g_1(\bm\B_0\T\bm\Gamma_0\T\Y) - \E g_1\}\{h_1(\X) - \E h_1\}\\
			&= \bm 0.
		\end{aligned}
	\end{equation*}
	Therefore, $f_1\in\Lambda_1^\perp\cap\Lambda_2^\perp\cap\Lambda_3^\perp\cap\Lambda_4^\perp = \Lambda^\perp$.
	Then we prove that $f_2(\bm\Gamma\T\Y, \B_0\T\bm\Gamma_0\T\Y)\in\Lambda^\perp$.
	\begin{equation*}
		\begin{aligned}
			\E(f_2\mid\bm\Gamma\T\Y, \X) &= \E[\{g_2(\bm\Gamma\T\Y) - \E g_2\}\{h_2(\bm\B_0\T\bm\Gamma_0\T\Y) - \E h_2\}\mid\bm\Gamma\T\Y, \X]\\
			&= \{g_2(\bm\Gamma\T\Y) - \E g_2\}\E\{h_2(\bm\B_0\T\bm\Gamma_0\T\Y) - \E h_2\mid\bm\Gamma\T\Y, \X\}\\
			&= \{g_2(\bm\Gamma\T\Y) - \E g_2\}\E\{h_2(\bm\B_0\T\bm\Gamma_0\T\Y) - \E h_2\}\\
			&= \bm 0.
		\end{aligned}
	\end{equation*}
	\begin{equation*}
		\begin{aligned}
			\E(f_2\mid\bm\Gamma_0\T\Y, \X) &= \E[\{g_2(\bm\Gamma\T\Y) - \E g_2\}\{h_2(\bm\B_0\T\bm\Gamma_0\T\Y) - \E h_2\}\mid\bm\Gamma_0\T\Y, \X]\\
			&= \E\{g_2(\bm\Gamma\T\Y) - \E g_2\mid\bm\Gamma_0\T\Y, \X\}\{h_2(\bm\B_0\T\bm\Gamma_0\T\Y) - \E h_2\}\\
			&= \E\{g_2(\bm\Gamma\T\Y) - \E g_2\mid\X\}\{h_2(\bm\B_0\T\bm\Gamma_0\T\Y) - \E h_2\}\\
		\end{aligned}
	\end{equation*}
	is a function of $(\bm\B_0\T\bm\Gamma_0\T\Y, \X)$.
	\begin{equation*}
		\begin{aligned}
			\E(f_2\mid\B_0\T\bm\Gamma_0\T\Y) &= \E[\{g_2(\bm\Gamma\T\Y) - \E g_2\}\{h_2(\bm\B_0\T\bm\Gamma_0\T\Y) - \E h_2\}\mid\B_0\T\bm\Gamma_0\T\Y]\\
			&= \E\{g_2(\bm\Gamma\T\Y) - \E g_2\mid\B_0\T\bm\Gamma_0\T\Y\}\{h_2(\bm\B_0\T\bm\Gamma_0\T\Y) - \E h_2\}\\
			&= \E\{g_2(\bm\Gamma\T\Y) - \E g_2\}\{h_2(\bm\B_0\T\bm\Gamma_0\T\Y) - \E h_2\}\\
			&= \bm 0.
		\end{aligned}
	\end{equation*}
	\begin{equation*}
		\begin{aligned}
			\E(f_2\mid\X) &= \E[\{g_2(\bm\Gamma\T\Y) - \E g_2\}\{h_2(\bm\B_0\T\bm\Gamma_0\T\Y) - \E h_2\}\mid\X]\\
			&= \E\{g_2(\bm\Gamma\T\Y) - \E g_2\mid \X\}\E\{h_2(\bm\B_0\T\bm\Gamma_0\T\Y) - \E h_2\mid\X\}\\
			&= \E\{g_2(\bm\Gamma\T\Y) - \E g_2\mid \X\}\E\{h_2(\bm\B_0\T\bm\Gamma_0\T\Y) - \E h_2\}\\
			&=\bm 0.
		\end{aligned}
	\end{equation*}
	The first to the second equation is because $\bm\Gamma\T\Y \indep \bm\B_0\T\bm\Gamma_0\T\Y \mid \X$. Hence, we also have $f_2\in\Lambda^\perp$.
	\end{proof}

	\begin{customthm}{4}\label{prop: score}
		The tangent space generated by the score vector with respect to the parameter of interest $\bm\theta$ is $\mathscr T_{\bm\theta} = \{\mathbf MS_{\bm\theta} \text{ for all } \mathbf M\in\mathbb{R}^{q\times q}\}$. The score vector $S_{\bm\theta} = (S_{\bm\gamma}\T, S_{\mathbf b}\T)\T$,
		\begin{equation*}
			\begin{aligned}
			S_{\bm\gamma} =& \text{vec}\T\bigg\{\Y\dfrac{\partial \log\eta_1}{\partial (\bm\Gamma\T\Y)\T}\bigg\}\dfrac{\partial \text{vec}(\bm\Gamma)}{\partial\bm\gamma\T} + \text{vec}\T\bigg\{\Y\dfrac{\partial \log\eta_3}{\partial (\mathbf B_0\T\bm\Gamma_0\T\Y)\T}\B_0\T\bigg\}\dfrac{\partial \text{vec}(\bm\Gamma_0)}{\partial\bm\gamma\T} \\
			&\quad+ \text{vec}\T\bigg\{\Y\dfrac{\partial \log\eta_2}{\partial (\mathbf B\T\bm\Gamma_0\T\Y)\T}\B\T + \Y\dfrac{\partial \log\eta_2}{\partial (\mathbf B_0\T\bm\Gamma_0\T\Y)\T}\B\T_0\bigg\}\dfrac{\partial \text{vec}(\bm\Gamma_0)}{\partial\bm\gamma\T}
			\end{aligned} 
		\end{equation*}
		and 
		\begin{align*}
			S_{\mathbf b} =& \text{vec}\T\bigg\{\bm\Gamma_0\T\Y\dfrac{\partial \log\eta_2}{\partial (\mathbf B\T\bm\Gamma_0\T\Y)\T}\bigg\}\dfrac{\partial \text{vec}(\B)}{\partial \mathbf b\T} +  \text{vec}\T\bigg\{\bm\Gamma_0\T\Y\dfrac{\partial \log\eta_2}{\partial (\mathbf B_0\T\bm\Gamma_0\T\Y)\T}\bigg\}\dfrac{\partial \text{vec}(\B_0)}{\partial \mathbf b\T} \\ &\quad + \text{vec}\T\bigg\{\bm\Gamma_0\T\Y\dfrac{\partial \log\eta_3}{\partial (\mathbf B_0\T\bm\Gamma_0\T\Y)\T}\bigg\}\dfrac{\partial \text{vec}(\B_0)}{\partial \mathbf b\T}.
		\end{align*}
		
	\end{customthm}
	
	\begin{proof}[Proof of Lemma \ref{prop: score}]
		Define $l(\bm\theta\mid\X, \Y)$ as the log-likelihood of $(\X,\Y)$, i.e., $$l(\bm\theta\mid\X, \Y) = \log\eta_1(\bm\Gamma\T\Y, \X) + \log\eta_2(\mathbf{B}\T\bm\Gamma_0\T\Y, \mathbf{B}_0\T\bm\Gamma_0\T\Y, \X) + \log\eta_3(\mathbf{B}_0\T\bm\Gamma_0\T\Y) + \log \eta_4(\X).$$ 
	Then, the score with respect to $\bm\gamma$ is as follows:
	\begin{equation*}
	S_{\bm\gamma} = \dfrac{\partial l(\bm\theta\mid\X,\Y)}{\partial\bm\gamma\T} = \dfrac{\partial \log\eta_1}{\partial\bm\gamma\T} + \dfrac{\partial \log\eta_2}{\partial\bm\gamma\T} + \dfrac{\partial \log\eta_3}{\partial\bm\gamma\T}.
	\end{equation*}
	Firstly, 
	\begin{equation*}
	\dfrac{\partial \log\eta_1}{\partial\bm\gamma\T} =  \dfrac{\partial \log\eta_1}{\partial (\bm\Gamma\T\Y)\T}\dfrac{\partial \bm\Gamma\T\Y}{\partial\text{vec}(\bm\Gamma)\T}\dfrac{\partial\text{vec}(\bm\Gamma)}{\partial\bm\gamma\T}.
	\end{equation*}
	It's easy to verify that $$\dfrac{\partial \bm\Gamma\T\Y}{\partial\text{vec}(\bm\Gamma)\T} = \mathbf I_u\otimes \Y\T.$$ Hence, 
	\begin{equation*}
	\dfrac{\partial \log\eta_1}{\partial\bm\gamma\T} = \dfrac{\partial \log\eta_1}{\partial (\bm\Gamma\T\Y)\T}(\mathbf I_u\otimes \Y\T)\dfrac{\partial\text{vec}(\bm\Gamma)}{\partial\bm\gamma\T} = \text{vec}\T\bigg\{\Y\dfrac{\partial \log\eta_1}{\partial (\bm\Gamma\T\Y)\T}\bigg\}\dfrac{\partial \text{vec}(\bm\Gamma)}{\partial\bm\gamma\T}.
	\end{equation*}
	Secondly,  
	\begin{align*}
	\dfrac{\partial \log\eta_2}{\partial\bm\gamma\T} =& \dfrac{\partial \log\eta_2}{\partial (\mathbf B\T\bm\Gamma_0\T\Y)\T}\dfrac{\partial \mathbf B\T\bm\Gamma_0\T\Y}{\partial\bm\gamma\T} + \dfrac{\partial \log\eta_2}{\partial (\mathbf B_0\T\bm\Gamma_0\T\Y)\T}\dfrac{\partial \mathbf B_0\T\bm\Gamma_0\T\Y}{\partial\bm\gamma\T}.
	\end{align*}
	We know that $\mathbf B\T\bm\Gamma_0\T\Y = \text{vec}(\mathbf B\T\bm\Gamma_0\T\Y) = \text{vec}(\Y\T\bm\Gamma_0\B) = (\B\T\otimes\Y\T)\text{vec}(\bm\Gamma_0)$. Thus,
	\begin{equation*}
	\dfrac{\partial \mathbf B\T\bm\Gamma_0\T\Y}{\partial\bm\gamma\T} = (\mathbf B\T\otimes\Y\T) \dfrac{\partial \text{vec}(\bm\Gamma_0)}{\partial\bm\gamma\T}. 
	\end{equation*}
	Therefore,
	\begin{align*}
		\dfrac{\partial \log\eta_2}{\partial\bm\gamma\T} =& \left\{\dfrac{\partial \log\eta_2}{\partial (\mathbf B\T\bm\Gamma_0\T\Y)\T} (\mathbf B\T\otimes\Y\T) + \dfrac{\partial \log\eta_2}{\partial (\mathbf B_0\T\bm\Gamma_0\T\Y)\T}(\mathbf B_0\T\otimes\Y\T)\right\}\dfrac{\partial \text{vec}(\bm\Gamma_0)}{\partial\bm\gamma\T}\\ 
	=&\text{vec}\T\bigg\{\Y\dfrac{\partial \log\eta_2}{\partial (\mathbf B\T\bm\Gamma_0\T\Y)\T}\B\T + \Y\dfrac{\partial \log\eta_2}{\partial (\mathbf B_0\T\bm\Gamma_0\T\Y)\T}\B\T_0\bigg\}\dfrac{\partial \text{vec}(\bm\Gamma_0)}{\partial\bm\gamma\T}
		\end{align*}
	Similarly, we can show that 
	\begin{equation*}
		\dfrac{\partial \log\eta_3}{\partial\bm\gamma\T} = \dfrac{\partial \log\eta_3}{\partial (\mathbf B_0\T\bm\Gamma_0\T\Y)\T}(\mathbf B_0\T\otimes\Y\T)\dfrac{\partial \text{vec}(\bm\Gamma_0)}{\partial\bm\gamma\T} = \text{vec}\T\bigg\{\Y\dfrac{\partial \log\eta_3}{\partial (\mathbf B_0\T\bm\Gamma_0\T\Y)\T}\B_0\T\bigg\}\dfrac{\partial \text{vec}(\bm\Gamma_0)}{\partial\bm\gamma\T} .
	\end{equation*}
	Consequently, the score for $\bm\gamma$ is
	\begin{equation*}
			\begin{aligned}
			S_{\bm\gamma} =& \text{vec}\T\bigg\{\Y\dfrac{\partial \log\eta_1}{\partial (\bm\Gamma\T\Y)\T}\bigg\}\dfrac{\partial \text{vec}(\bm\Gamma)}{\partial\bm\gamma\T} + \text{vec}\T\bigg\{\Y\dfrac{\partial \log\eta_3}{\partial (\mathbf B_0\T\bm\Gamma_0\T\Y)\T}\B_0\T\bigg\}\dfrac{\partial \text{vec}(\bm\Gamma_0)}{\partial\bm\gamma\T} \\
			&\quad+ \text{vec}\T\bigg\{\Y\dfrac{\partial \log\eta_2}{\partial (\mathbf B\T\bm\Gamma_0\T\Y)\T}\B\T + \Y\dfrac{\partial \log\eta_2}{\partial (\mathbf B_0\T\bm\Gamma_0\T\Y)\T}\B\T_0\bigg\}\dfrac{\partial \text{vec}(\bm\Gamma_0)}{\partial\bm\gamma\T}
			\end{aligned} 
		\end{equation*}
	Next, we calculate the score with respect to $\mathbf b$.
	\begin{equation*}
	\dfrac{\partial l(\bm\gamma\mid\X,\Y)}{\partial \mathbf b\T} = \dfrac{\partial \log\eta_2}{\partial \mathbf b\T} + \dfrac{\partial \log\eta_3}{\partial \mathbf b\T}.
	\end{equation*}
	By the chain rule, we have 
	\begin{equation*}
	\dfrac{\partial \log\eta_2}{\partial \mathbf b\T} = \dfrac{\partial \log\eta_2}{\partial (\mathbf B\T\bm\Gamma_0\T\Y)\T}\dfrac{\partial \mathbf B\T\bm\Gamma_0\T\Y}{\partial \mathbf b\T} + \dfrac{\partial \log\eta_2}{\partial (\mathbf B_0\T\bm\Gamma_0\T\Y)\T}\dfrac{\partial \mathbf B_0\T\bm\Gamma_0\T\Y}{\partial  \mathbf b\T}.
	\end{equation*}
	Recall that $\mathbf B\in \mathbb R^{(r-u)\times d}$ and $\mathbf B_0\in \mathbb R^{(r-u)\times (r-u-d)}$, where $d$ is the dimension for $\mathcal S_2$. By the same ``vec" trick, we have 
	\begin{equation*}
		\B\T\bm\Gamma_0\T\Y = \text{vec}(\B\T\bm\Gamma_0\T\Y) = \text{vec}(\Y\T\bm\Gamma_0\B\mathbf I_d) = (\mathbf I_d\otimes\Y\T\bm\Gamma_0)\text{vec}(\B).
	\end{equation*}
	Hence, 
	\begin{align*}
		\dfrac{\partial \log\eta_2}{\partial \mathbf b\T} =& \dfrac{\partial \log\eta_2}{\partial (\mathbf B\T\bm\Gamma_0\T\Y)\T}(\mathbf I_d\otimes\Y\T\bm\Gamma_0)\dfrac{\partial \text{vec}(\B)}{\partial \mathbf b\T} + \dfrac{\partial \log\eta_2}{\partial (\mathbf B_0\T\bm\Gamma_0\T\Y)\T}(\mathbf I_{r-u-d}\otimes\Y\T\bm\Gamma_0) \dfrac{\partial \text{vec}(\B_0)}{\partial \mathbf b\T}\\
		=&\text{vec}\T\bigg\{\bm\Gamma_0\T\Y\dfrac{\partial \log\eta_2}{\partial (\mathbf B\T\bm\Gamma_0\T\Y)\T}\bigg\}\dfrac{\partial \text{vec}(\B)}{\partial \mathbf b\T} +  \text{vec}\T\bigg\{\bm\Gamma_0\T\Y\dfrac{\partial \log\eta_2}{\partial (\mathbf B_0\T\bm\Gamma_0\T\Y)\T}\bigg\}\dfrac{\partial \text{vec}(\B_0)}{\partial \mathbf b\T}.
	\end{align*}
	Similarly, 
	\begin{align*}
		\dfrac{\partial \log\eta_3}{\partial \mathbf b\T} =& \dfrac{\partial \log\eta_3}{\partial (\mathbf B_0\T\bm\Gamma_0\T\Y)\T}(\mathbf I_{r-u-d}\otimes\Y\T\bm\Gamma_0) \dfrac{\partial \text{vec}(\B_0)}{\partial \mathbf b\T}\\
		=&\text{vec}\T\bigg\{\bm\Gamma_0\T\Y\dfrac{\partial \log\eta_3}{\partial (\mathbf B_0\T\bm\Gamma_0\T\Y)\T}\bigg\}\dfrac{\partial \text{vec}(\B_0)}{\partial \mathbf b\T}.
	\end{align*}
	Therefore, 
	\begin{align*}
			S_{\mathbf b} =& \text{vec}\T\bigg\{\bm\Gamma_0\T\Y\dfrac{\partial \log\eta_2}{\partial (\mathbf B\T\bm\Gamma_0\T\Y)\T}\bigg\}\dfrac{\partial \text{vec}(\B)}{\partial \mathbf b\T} +  \text{vec}\T\bigg\{\bm\Gamma_0\T\Y\dfrac{\partial \log\eta_2}{\partial (\mathbf B_0\T\bm\Gamma_0\T\Y)\T}\bigg\}\dfrac{\partial \text{vec}(\B_0)}{\partial \mathbf b\T} \\ &\quad + \text{vec}\T\bigg\{\bm\Gamma_0\T\Y\dfrac{\partial \log\eta_3}{\partial (\mathbf B_0\T\bm\Gamma_0\T\Y)\T}\bigg\}\dfrac{\partial \text{vec}(\B_0)}{\partial \mathbf b\T}.
		\end{align*}
	The score vector can be formulated as $S_{\bm\theta} = (S_{\bm\gamma}, S_{\mathbf b})$.
	
	\end{proof}
	\begin{proof}[Regularity Conditions and Proof of Theorem \ref{thm: GMM}]
		Let $\bm\Theta$ denote the parameter space that contains $\bm\theta$, and let $\bm\theta_0$ denote the true parameter value. 
		Firstly, we state the regularity conditions needed for Theorem \ref{thm: GMM}: 
		
		\noindent(S1) $\E\{\mathbf f(\Y,\X;\bm\theta)\}\neq \bm0$ for all $\bm\theta\in\bm\Theta$ such that $\bm\theta\neq\bm\theta_0$. Also, $\bm\theta_0$ is an interior point in $\bm\Theta$.
		
		\noindent (S2) The parameter space $\bm\Theta$ is a compact set.
		
		\noindent (S3) $\E\{\sup_{\bm\theta\in\bm\Theta}\|\mathbf f(\Y,\X;\bm\theta)\|_2\}<\infty.$
		
		\noindent(S4)
		$\mathbf C_1$ is continuous on some neighborhood $\bm\Theta_\epsilon$ of $\bm\theta_0$, where $\mathbf C_1 = \E \{\partial\rm{vec}(\mathbf f)/\partial\bm\theta\T\}$.
		
		\noindent (S5) The matrix $\mathbf C_1$ is of full column rank.
		
		\noindent (S6) $\sup_{\bm\theta\in\bm\Theta_\epsilon}\|n^{-1}\sum_{i=1}^N \partial\text{vec}(\mathbf f)/\partial\bm\theta\T - \mathbf C_1\|\overset{p}{\rightarrow}\bm 0$.

		Regularity conditions (S1)--(S5) are standard conditions for the GMM estimators. Based on these equations, by Theorem 3.2 in \cite{hall2004generalized}, we have $$\sqrt n (\hat{\bm\theta} - \bm\theta)\xrightarrow{d} \mathcal N\big\{\bm 0, (\mathbf C_1\T\mathbf C_1)^{-1}\mathbf C_1\T\mathbf D_1\mathbf C_1(\mathbf C_1\T\mathbf C_1)^{-1}\big\},$$
		where $\mathbf D_1 = {\rm Var}\{ \mathbf f(\Y_i, \X_i;\bm\theta) \}$.
	\end{proof}
	\begin{proof}[Derivation of the efficient score $S_{eff}$]
		Since $S_{\bm\theta}$ satisfies $\E(S_{\bm\theta}) = 0$, we only need to consider the projection of any mean zero function $h(\Y, \X)$. Let $\Pi(h\mid \Lambda^\perp)$ denote the projection of $h$ onto $\Lambda^\perp$. Then, by Lemma \ref{Prop: structure tangent space}, 
	\begin{equation*}
	\Pi(h\mid \Lambda^\perp) = h - \Pi(h\mid \Lambda) = h - \Pi(h\mid \Lambda_1)
	- \Pi(h\mid \Lambda_2) - \Pi(h\mid \Lambda_3) - \Pi(h\mid \Lambda_4).
	\end{equation*}
	Notice that $ \Pi(h\mid \Lambda_1) = \E(h\mid \bm\Gamma\T\Y,\X) - \E(h\mid\X)$, $ \Pi(h\mid \Lambda_2) = \E(h\mid \bm\Gamma_0\T\Y,\X) - \E(h\mid \mathbf B_0\T\bm\Gamma_0\T\Y, \X)$, $\Pi(h\mid \Lambda_3) = \E(h\mid \mathbf B_0\T\bm\Gamma_0\T\Y)$ and $\Pi(h\mid \Lambda_4) = \E(h\mid\X)$. Therefore, 
	$$S_{eff} = S_{\bm\theta} - \E(S_{\bm\theta}\mid \bm\Gamma\T\Y,\X) - \E(S_{\bm\theta}\mid \bm\Gamma_0\T\Y,\X) + \E(S_{\bm\theta}\mid \mathbf B_0\T\bm\Gamma_0\T\Y, \X) -  \E(S_{\bm\theta}\mid \mathbf B_0\T\bm\Gamma_0\T\Y).$$
	Notice that $$S_{\bm\theta} = \dfrac{\partial \log\eta_1}{\partial \bm\theta} + \dfrac{\partial \log\eta_2}{\partial \bm\theta}+ \dfrac{\partial \log\eta_3}{\partial \bm\theta}.$$
	Firstly, since $\eta_1(\cdot)$ is the conditional distribution of $\bm\Gamma\T\Y\mid\X$, and the model assumption $\bm\Gamma^T\Y\perp\bm\Gamma_0^T\Y\mid\X$ holds, we must have 
	$$\E\left(\dfrac{\partial \log\eta_1}{\partial \bm\theta}\bmid \bm\Gamma_0\T\Y,\X\right) = \bm 0.$$
	As a consequence, $$\E\left(\dfrac{\partial \log\eta_1}{\partial \bm\theta}\bmid \mathbf B_0\T\bm\Gamma_0\T\Y,\X\right) = \bm 0,$$
	and 
	$$\E\left(\dfrac{\partial \log\eta_1}{\partial \bm\theta}\bmid \mathbf B_0\T\bm\Gamma_0\T\Y\right) = \bm 0.$$
	Similarly, for $\eta_2(\cdot)$, we have 
	$$\E\left(\dfrac{\partial \log\eta_2}{\partial \bm\theta}\bmid \bm\Gamma\T\Y, \X\right) = \E\left(\dfrac{\partial \log\eta_2}{\partial \bm\theta}\bmid \mathbf B_0\T\bm\Gamma_0\T\Y, \X\right) = \E\left(\dfrac{\partial \log\eta_2}{\partial \bm\theta}\bmid \mathbf B_0\T\bm\Gamma_0\T\Y\right) = \bm 0.$$
	Also, because $\B_0\T\bm\Gamma_0\T\Y\indep(\bm\Gamma\T\Y, \X)$, 
	$$\E\left(\dfrac{\partial \log\eta_3}{\partial \bm\theta}\bmid \bm\Gamma\T\Y, \X\right) = \bm 0.$$
	Hence, 
	\begin{equation*}
		\begin{aligned}
			S_{eff,\bm\gamma} &= \dfrac{\partial \log\eta_1}{\partial \bm\gamma\T} + \dfrac{\partial \log\eta_2}{\partial \bm\gamma\T} + \dfrac{\partial \log\eta_3}{\partial \bm\gamma\T} - \E\left(\dfrac{\partial \log\eta_1}{\partial\bm\gamma}\bmid\bm\Gamma\T\Y,\X\right) - \E\left(\dfrac{\partial \log\eta_2}{\partial\bm\gamma}\bmid\bm\Gamma_0\T\Y,\X\right)\\
			&\quad - \E\left(\dfrac{\partial \log\eta_3}{\partial\bm\gamma}\bmid\bm\Gamma_0\T\Y,\X\right) + \E\left(\dfrac{\partial \log\eta_3}{\partial\bm\gamma}\bmid\B_0\T\bm\Gamma_0\T\Y,\X\right) - \E\left(\dfrac{\partial \log\eta_3}{\partial\bm\gamma}\bmid\B_0\T\bm\Gamma_0\T\Y\right).
		\end{aligned}
	\end{equation*}
	Notably, 
	\begin{align*}
		&\dfrac{\partial \log\eta_1}{\partial \bm\gamma\T} - \E\left(\dfrac{\partial \log\eta_1}{\partial\bm\gamma}\bmid\bm\Gamma\T\Y,\X\right)\\
		=&\text{vec}\T\left[\{\Y - \E(\Y\mid\bm\Gamma\T\Y,\X)\}\dfrac{\partial \log\eta_1}{\partial (\bm\Gamma\T\Y)\T}\right]\dfrac{\partial \text{vec}(\bm\Gamma)}{\partial\bm\gamma\T},\\
		& \dfrac{\partial \log\eta_2}{\partial \bm\gamma\T} - \E\left(\dfrac{\partial \log\eta_2}{\partial\bm\gamma}\bmid\bm\Gamma_0\T\Y,\X\right)\\
		=&\text{vec}\T\bigg[\left\{\Y - \E(\Y\mid\bm\Gamma_0\T\Y,\X)\right\}\bigg\{\dfrac{\partial \log\eta_2}{\partial (\mathbf B\T\bm\Gamma_0\T\Y)\T}\B\T + \dfrac{\partial \log\eta_2}{\partial (\mathbf B_0\T\bm\Gamma_0\T\Y)\T}\B_0\T\bigg\}\bigg]\dfrac{\partial \text{vec}(\bm\Gamma_0)}{\partial\bm\gamma\T},\\
		& \dfrac{\partial \log\eta_3}{\partial \bm\gamma\T} - \E\left(\dfrac{\partial \log\eta_3}{\partial\bm\gamma}\bmid\bm\Gamma_0\T\Y,\X\right) + \E\left(\dfrac{\partial \log\eta_3}{\partial\bm\gamma}\bmid\B_0\T\bm\Gamma_0\T\Y,\X\right) - \E\left(\dfrac{\partial \log\eta_3}{\partial\bm\gamma}\bmid\B_0\T\bm\Gamma_0\T\Y\right)\\
		=& \text{vec}\T\bigg\{\big\{\Y - \E(\Y\mid\bm\Gamma_0\T\Y,\X) + \E(\Y\mid\B_0\T\bm\Gamma_0\T\Y,\X) - \E(\Y\mid\B_0\T\bm\Gamma_0\T\Y)\big\}\dfrac{\partial \log\eta_3}{\partial (\mathbf B_0\T\bm\Gamma_0\T\Y)\T}\B_0\T\bigg\}\dfrac{\partial \text{vec}(\bm\Gamma_0)}{\partial\bm\gamma\T}
	\end{align*}
	In addition, \begin{align*}
		\Y - \E(\Y\mid\bm\Gamma\T\Y,\X) &= \mathbf P_{\bm\Gamma}\Y + \mathbf Q_{\bm\Gamma} - \E(\mathbf P_{\bm\Gamma}\Y\mid\bm\Gamma\T\Y,\X) - \E(\mathbf Q_{\bm\Gamma}\Y\mid\bm\Gamma\T\Y,\X)\\
		&=\mathbf Q_{\bm\Gamma} - \E(\mathbf Q_{\bm\Gamma}\Y\mid\bm\Gamma\T\Y,\X)\\
		&=\mathbf Q_{\bm\Gamma} - \E(\mathbf Q_{\bm\Gamma}\Y\mid\X),
	\end{align*}
	where the last equation is because $\bm\Gamma\T\Y\indep\bm\Gamma_0\T\Y\mid\X$.
	\begin{align*}
		\Y - \E(\Y\mid\bm\Gamma_0\T\Y,\X) &= \mathbf P_{\bm\Gamma}\Y + \mathbf Q_{\bm\Gamma} - \E(\mathbf P_{\bm\Gamma}\Y\mid\bm\Gamma_0\T\Y,\X) - \E(\mathbf Q_{\bm\Gamma}\Y\mid\bm\Gamma_0\T\Y,\X)\\
		&=\mathbf P_{\bm\Gamma} - \E(\mathbf P_{\bm\Gamma}\Y\mid\bm\Gamma_0\T\Y,\X)\\
		&=\mathbf P_{\bm\Gamma} - \E(\mathbf P_{\bm\Gamma}\Y\mid\X).
	\end{align*}
	\begin{align*}
		&\Y - \E(\Y\mid\bm\Gamma_0\T\Y,\X) + \E(\Y\mid\B_0\T\bm\Gamma_0\T\Y,\X) - \E(\Y\mid\B_0\T\bm\Gamma_0\T\Y)\\
		=&\mathbf P_{\bm\Gamma} - \E(\mathbf P_{\bm\Gamma}\Y\mid\X) + \E(\mathbf P_{\bm\Gamma}\Y\mid\B_0\T\bm\Gamma_0\T\Y,\X) + \E(\mathbf P_{\bm\Gamma_0\B}\Y\mid\B_0\T\bm\Gamma_0\T\Y,\X)+\E(\mathbf P_{\bm\Gamma_0\B_0}\Y\mid\B_0\T\bm\Gamma_0\T\Y,\X)\\ & -\E(\mathbf P_{\bm\Gamma}\Y\mid\B_0\T\bm\Gamma_0\T\Y) -\E(\mathbf P_{\bm\Gamma_0\B}\Y\mid\B_0\T\bm\Gamma_0\T\Y) -\E(\mathbf P_{\bm\Gamma_0\B_0}\Y\mid\B_0\T\bm\Gamma_0\T\Y)\\
		=&\mathbf P_{\bm\Gamma} + \E(\mathbf P_{\bm\Gamma_0\B}\Y\mid\B_0\T\bm\Gamma_0\T\Y,\X) - \E(\mathbf P_{\bm\Gamma_0\B}\Y\mid\B_0\T\bm\Gamma_0\T\Y).
	\end{align*}
	The last equation is because 
	\begin{align*}
		&\E(\mathbf P_{\bm\Gamma}\Y\mid\B_0\T\bm\Gamma_0\T\Y,\X) = \E(\mathbf P_{\bm\Gamma}\Y\mid\X),\\
		& \E(\mathbf P_{\bm\Gamma_0\B_0}\Y\mid\B_0\T\bm\Gamma_0\T\Y,\X)-\E(\mathbf P_{\bm\Gamma_0\B_0}\Y\mid\B_0\T\bm\Gamma_0\T\Y) = \mathbf P_{\bm\Gamma_0\B_0}\Y - \mathbf P_{\bm\Gamma_0\B_0}\Y = \bm 0,\\
		& \E(\mathbf P_{\bm\Gamma}\Y\mid\B_0\T\bm\Gamma_0\T\Y) = \E(\mathbf P_{\bm\Gamma}\Y) = \mathbf P_{\bm\Gamma}\E(\Y) =  \bm 0.
	\end{align*}
	Therefore, 
	\begin{align*}
		S_{\rm eff,\bm\gamma} &= \text{vec}\T\left\{\mathbf Q_{\bm\Gamma}\bm\Delta_1\dfrac{\partial \log\eta_1}{\partial (\bm\Gamma\T\Y)\T}\right\}\dfrac{\partial \text{vec}(\bm\Gamma)}{\partial\bm\gamma\T} + \text{vec}\T\bigg\{(\mathbf P_{\bm\Gamma}\Y + \bm\Delta_2)\dfrac{\partial \log\eta_3}{\partial (\mathbf B_0\T\bm\Gamma_0\T\Y)\T}\B_0\T\bigg\}\dfrac{\partial \text{vec}(\bm\Gamma_0)}{\partial\bm\gamma\T}\\
			&\quad+ \text{vec}\T\bigg[\mathbf P_{\bm\Gamma}\bm\Delta_1\bigg\{\dfrac{\partial \log\eta_2}{\partial (\mathbf B\T\bm\Gamma_0\T\Y)\T}\B\T + \dfrac{\partial \log\eta_2}{\partial (\mathbf B_0\T\bm\Gamma_0\T\Y)\T}\B_0\T\bigg\}\bigg]\dfrac{\partial \text{vec}(\bm\Gamma_0)}{\partial\bm\gamma\T},
	\end{align*}
	where $\bm\Delta_1 = \Y - \E(\Y\mid\X)$, and $\bm\Delta_2 = \E(\mathbf P_{\bm\Gamma_0\B}\Y\mid \B_0\T\bm\Gamma_0\T\Y,\X) - \E(\mathbf P_{\bm\Gamma_0\B}\Y\mid \B_0\T\bm\Gamma_0\T\Y)$.
	Similarly, 
	\begin{equation*}
		\begin{aligned}
			S_{\text{eff},\mathbf b} &= S_{\mathbf b} - \E\left(\dfrac{\partial \log\eta_2}{\partial\mathbf b}\bmid\bm\Gamma_0\T\Y,\X\right) - \E\left(\dfrac{\partial \log\eta_3}{\partial\mathbf b}\bmid\bm\Gamma_0\T\Y,\X\right)\\
			&\quad + \E\left(\dfrac{\partial \log\eta_3}{\partial\mathbf b}\bmid\B_0\T\bm\Gamma_0\T\Y,\X\right) - \E\left(\dfrac{\partial \log\eta_3}{\partial\mathbf b}\bmid\B_0\T\bm\Gamma_0\T\Y\right)\\
			&= \text{vec}\T\bigg\{\bm\Gamma_0\T\PP_{\bm\Gamma}\bm\Delta_1\dfrac{\partial \log\eta_2}{\partial (\mathbf B\T\bm\Gamma_0\T\Y)\T}\bigg\}\dfrac{\partial \text{vec}(\B)}{\partial \mathbf b\T} + \text{vec}\T\bigg\{\bm\Gamma_0\T\PP_{\bm\Gamma}\bm\Delta_1\dfrac{\partial \log\eta_2}{\partial (\mathbf B_0\T\bm\Gamma_0\T\Y)\T}\bigg\}\dfrac{\partial \text{vec}(\B_0)}{\partial \mathbf b\T}\\ 
			&\quad + \text{vec}\T\bigg\{\bm\Gamma_0\T\bm\Delta_2\dfrac{\partial \log\eta_3}{\partial (\mathbf B_0\T\bm\Gamma_0\T\Y)\T}\bigg\}\dfrac{\partial \text{vec}(\B_0)}{\partial \mathbf b\T}.
		\end{aligned}
	\end{equation*}
	Since $\bm\Gamma_0\T\PP_{\bm\Gamma} = \bm\Gamma_0\T\bm\Gamma\bm\Gamma\T = 0$, $S_{\text{eff},\mathbf b}$ can be further simplified as 
	$$S_{\text{eff},\mathbf b} = \text{vec}\T\bigg\{\bm\Gamma_0\T\bm\Delta_2\dfrac{\partial \log\eta_3}{\partial (\mathbf B_0\T\bm\Gamma_0\T\Y)\T}\bigg\}\dfrac{\partial \text{vec}(\B_0)}{\partial \mathbf b\T}$$
	\end{proof}
	Before proving the following theorem, we first prove a lemma.

\begin{customthm}{5}
    \label{lemma: supp3}
	Under regularity conditions (A1)--(A5), we have 
	\begin{equation*}
		\begin{aligned}
			&\dfrac{1}{\sqrt n}\sum_{i=1}^n\bigg[{\bm\Delta}_{i1}\bigg\{\dfrac{\partial \log\hat\eta_1}{\partial ({\bm\Gamma}\T\Y_i)\T} - \dfrac{\partial \log\eta_1}{\partial ({\bm\Gamma}\T\Y_i)\T}\bigg\}\bigg] = O_p(h^2 + n^{1/2}h^4 + n^{-1/2}h^{-(u+p+1)}\log n),\\
			&\dfrac{1}{\sqrt n}\sum_{i=1}^n\bigg[{\bm\Delta}_{i1}\bigg\{\dfrac{\partial \log\hat\eta_2}{\partial (\B\T{\bm\Gamma}_0\T\Y_i)\T} - \dfrac{\partial \log\eta_2}{\partial (\B\T{\bm\Gamma}_0\T\Y_i)\T}\bigg\}\bigg] = O_p(h^2 + n^{1/2}h^4 + n^{-1/2}h^{-(r+p-u+1)}\log n),\\
			&\text{and }\quad  \dfrac{1}{\sqrt n}\sum_{i=1}^n\bigg[{\bm\Delta}_{i2}({\bm\theta_0})\bigg\{\dfrac{\partial \log\hat\eta_3}{\partial ({\mathbf B}_0\T{\bm\Gamma}_0\T\Y_i)\T} - \dfrac{\partial \log\eta_3}{\partial ({\mathbf B}_0\T{\bm\Gamma}_0\T\Y_i)\T}\bigg\}\bigg] = O_p(h^2 + n^{-1/2}h^4 + n^{-1/2}h^{-(r-u-d+1)}\log n).
		\end{aligned}
	\end{equation*}
	Under regularity conditions (B1)--(B4) and (A4)--(A5), we have 
	\begin{align*}
	&\dfrac{1}{\sqrt n}\sum_{i=1}^{n}\text{vec}\T\bigg[{\bm\Gamma}_0\T\left\{\hat{\bm\Delta}_{i2}({\bm\theta_0}) - {\bm\Delta}_{i2}({\bm\theta_0})\right\}\bigg\{\dfrac{\partial \log\eta_3^*}{\partial (\B_0\T{\bm\Gamma}_0\T\Y_i)\T} - \E\bigg(\dfrac{\partial \log\eta_3^*}{\partial (\B_0\T{\bm\Gamma}_0\T\Y_i)\T}\bigg) \bigg\}\bigg]\\
	=&O_p(h^2 + n^{-1/2}h^{p + r-u-d}\log n + n^{1/2}h^4),\end{align*}
	\begin{align*}
	    &\dfrac{1}{\sqrt n}\sum_{i=1}^{n}\mathbf P_{{\bm\Gamma}}\{\hat{\bm\Delta}_{i1}-{\bm\Delta}_{i1}\}\bigg\{\dfrac{\partial \log\eta_2^*}{\partial ({\mathbf B}\T{\bm\Gamma}_0\T\Y_i)\T}-\E\bigg(\dfrac{\partial \log\eta_2^*}{\partial ({\mathbf B}\T{\bm\Gamma}_0\T\Y_i)\T}\bmid \B_0\T{\bm\Gamma}_0\T\Y_i,\X_i\bigg) \bigg\}\\
	    =&O_p(h^2+n^{-1/2}h^p\log n + n^{1/2}h^4),
	\end{align*}
	where
\begin{align*}
	&\hat{\bm\Delta}_{i1} = \Y_i - \hat{\mathrm E}(\Y_i\mid\X_i)\\
	&\hat{\bm\Delta}_{i2}({\bm\theta_0}) = \mathbf P_{{\bm\Gamma}_0\B}\{\hat{\mathrm E}(\Y_i\mid {\B}_0{\bm\Gamma}_0\T\Y_i,\X_i) - \hat{\mathrm E}(\Y_i\mid {\B}_0{\bm\Gamma}_0\T\Y_i)\}.
\end{align*}
\end{customthm}
\begin{proof}[Proof of Lemma \ref{lemma: supp3}]
	Since the equalities and their proofs are similar, we only show the proof of the first one. Recall the kernel density estimation of ${\partial \eta_1}/{\partial ({\bm\Gamma}\T\Y)\T}$ has the form 
	$$\dfrac{\partial \hat\eta_1(\bm\Gamma\T\Y,\X)}{\partial ({\bm\Gamma}\T\Y)\T} = \dfrac{\sum_{i=1}^nK_h'(\bm\Gamma\T\Y - {\bm\Gamma}\T\Y_i)K_h(\X-\X_i)}{\sum_{i=1}^nK_h(\X-\X_i)}.$$
	Hence, 
	$$\dfrac{\partial \log\hat\eta_1(\bm\Gamma\T\Y,\X)}{\partial ({\bm\Gamma}\T\Y)\T} = \dfrac{\sum_{i=1}^nK_h'(\bm\Gamma\T\Y - {\bm\Gamma}\T\Y_i)K_h(\X-\X_i)}{\sum_{i=1}^nK_h(\bm\Gamma\T\Y - {\bm\Gamma}\T\Y_i)K_h(\X-\X_i)}.$$
	Let $\hat f(\bm\Gamma\T\Y_i, \X_i) = n^{-1}\sum_{j=1}^nK_h(\bm\Gamma\T\Y_i - \bm\Gamma\T\Y_j)K_h(\X_i - \X_j)$ and $\hat{\mathbf r}_1(\bm\Gamma\T\Y_i,\X_i) = n^{-1}\sum_{j=1}^nK_h'(\bm\Gamma\T\Y_i - \bm\Gamma\T\Y_j)K_h(\X_i - \X_j)$. Also, let ${\mathbf r}_1(\bm\Gamma\T\Y_i,\X_i) = {\partial \eta_1(\bm\Gamma\T\Y,\X)}/{\partial ({\bm\Gamma}\T\Y)\T}$. Hence, we have
	\begin{equation*}
		\begin{aligned}
			&\dfrac{1}{n}\sum_{i=1}^n{\bm\Delta}_{i1}\bigg\{\dfrac{\partial \log\hat\eta_1}{\partial ({\bm\Gamma}\T\Y_i)\T} - \dfrac{\partial \log\eta_1}{\partial ({\bm\Gamma}\T\Y_i)\T}\bigg\}\\
			=\quad&\dfrac{1}{n}\sum_{i=1}^n{\bm\Delta}_{i1}\left[\dfrac{\hat{\mathbf r}_1(\bm\Gamma\T\Y_i,\X_i)}{\hat f(\bm\Gamma\T\Y_i, \X_i)} - \dfrac{{\mathbf r}_1(\bm\Gamma\T\Y_i,\X_i)}{f(\bm\Gamma\T\Y_i, \X_i)}\right]\\
			=\quad&\dfrac{1}{n}\sum_{i=1}^n{\bm\Delta}_{i1}\left[\dfrac{\hat{\mathbf r}_1(\bm\Gamma\T\Y_i,\X_i) - {\mathbf r}_1(\bm\Gamma\T\Y_i,\X_i)}{f(\bm\Gamma\T\Y_i, \X_i)}\right] - \dfrac{1}{n}\sum_{i=1}^n{\bm\Delta}_{i1}\left[\dfrac{{\mathbf r}_1(\bm\Gamma\T\Y_i,\X_i)\{\hat f(\bm\Gamma\T\Y_i, \X_i) - f(\bm\Gamma\T\Y_i, \X_i)\}}{f^2(\bm\Gamma\T\Y_i, \X_i)}\right]+o_p(1)\\
			&\quad -\dfrac{1}{n}\sum_{i=1}^n\bm{\bm\Delta}_{i1}\left[\dfrac{\{\hat{\mathbf r}_1(\bm\Gamma\T\Y_i,\X_i) - {\mathbf r}_1(\bm\Gamma\T\Y_i,\X_i)\}\{\hat f(\bm\Gamma\T\Y_i, \X_i) - f(\bm\Gamma\T\Y_i, \X_i)\}}{f(\bm\Gamma\T\Y_i, \X_i)\hat f(\bm\Gamma\T\Y_i, \X_i)}\right]\\
			&\quad +\dfrac{1}{n}\sum_{i=1}^n{\bm\Delta}_{i1}\left[\dfrac{{\mathbf r}_1(\bm\Gamma\T\Y_i,\X_i)\{\hat f(\bm\Gamma\T\Y_i, \X_i) - f(\bm\Gamma\T\Y_i, \X_i)\}^2}{f^2(\bm\Gamma\T\Y_i, \X_i)\hat f(\bm\Gamma\T\Y_i, \X_i)}\right].
		\end{aligned}
	\end{equation*}
	The second equation is because $\|\hat{\bm\Gamma} - \bm\Gamma\|\leq Cn^{-1/2}$.
	By the uniform convergence of nonparametric regression \cite{li2007nonparametric}, we have 
	$$\sup_{\X,\Y}|\hat f(\bm\Gamma\T\Y, \X) - f(\bm\Gamma\T\Y, \X)| = O_p(\sqrt{n^{-1}h^{-(p+u)}\log n} + h^2)$$ and $$\sup_{\X,\Y}|\hat {\mathbf r}_1(\bm\Gamma\T\Y, \X) - \mathbf r_1(\bm\Gamma\T\Y, \X)| = O_p(\sqrt{n^{-1}h^{-(p+u+2)}\log n} + h^2).$$ Therefore, the third quantity can be bounded by
	\begin{equation*}
		\begin{aligned}
			&\bmid\dfrac{1}{n}\sum_{i=1}^n{\bm\Delta}_{i1}\left[\dfrac{\{\hat{\mathbf r}_1(\bm\Gamma\T\Y_i,\X_i) - {\mathbf r}_1(\bm\Gamma\T\Y_i,\X_i)\}\{\hat f(\bm\Gamma\T\Y_i, \X_i) - f(\bm\Gamma\T\Y_i, \X_i)\}}{f(\bm\Gamma\T\Y_i, \X_i)\hat f(\bm\Gamma\T\Y_i, \X_i)}\right]\bmid\\
			\leq\quad&\dfrac{1}{n}\sum_{i=1}^n|{\bm\Delta}_{i1}|\left[\dfrac{\sup_{\X,\Y}|\hat {\mathbf r}_1(\bm\Gamma\T\Y, \X) - \mathbf r_1(\bm\Gamma\T\Y, \X)|\sup_{\X,\Y}|\hat f(\bm\Gamma\T\Y, \X) - f(\bm\Gamma\T\Y, \X)|}{|f(\bm\Gamma\T\Y_i, \X_i)\hat f(\bm\Gamma\T\Y_i, \X_i)|}\right]\\
			=\quad&\dfrac{1}{n}\sum_{i=1}^n|{\bm\Delta}_{i1}|\left[\dfrac{O_p(n^{-1}h^{-(p+u+1)}\log n + h^4)}{|f(\bm\Gamma\T\Y_i, \X_i)\hat f(\bm\Gamma\T\Y_i, \X_i)|}\right]\\
			=\quad&O_p(n^{-1}h^{-(p+u+1)}\log n + h^4)\cdot\dfrac{1}{n}\sum_{i=1}^n|{\bm\Delta}_{i1}|\\
			=\quad&O_p(n^{-1}h^{-(p+u+1)}\log n + h^4).
		\end{aligned}
	\end{equation*}
	The second to the third equation is because
	$$|f(\bm\Gamma\T\Y_i, \X_i)\hat f(\bm\Gamma\T\Y_i, \X_i)|^{-1} = |f(\bm\Gamma\T\Y_i, \X_i)\{f(\bm\Gamma\T\Y_i, \X_i)+o_p(1)\}|^{-1} = O_p(1).$$
	Using the exact same technique, the fourth quantity in the above decomposition is of order $O_p(n^{-1}h^{-(p+u)}\log n + h^4)$.

	Notice that the first two quantities also share the same structure, in the sequel we only deal with the first quantity. 
	\begin{equation*}
		\begin{aligned}
			&\dfrac{1}{n}\sum_{i=1}^n{\bm\Delta}_{i1}\left[\dfrac{\hat{\mathbf r}_1(\bm\Gamma\T\Y_i,\X_i) - {\mathbf r}_1(\bm\Gamma\T\Y_i,\X_i)}{f(\bm\Gamma\T\Y_i, \X_i)}\right]\\
			=\quad&O_p(1)\cdot\dfrac{1}{n}\sum_{i=1}^n{\bm\Delta}_{i1}\big[\hat{\mathbf r}_1(\bm\Gamma\T\Y_i,\X_i) - \E\{\hat{\mathbf r}_1(\bm\Gamma\T\Y_i,\X_i)\mid\bm\Gamma\T\Y_i,\X_i\}\big] \\
			&\quad +O_p(1)\cdot\dfrac{1}{n}\sum_{i=1}^n{\bm\Delta}_{i1}\big[\E\{\hat{\mathbf r}_1(\bm\Gamma\T\Y_i,\X_i)\mid\bm\Gamma\T\Y_i,\X_i\} - {\mathbf r}_1(\bm\Gamma\T\Y_i,\X_i)\big].
		\end{aligned}
	\end{equation*}
	We can write 
	$$\dfrac{1}{n}\sum_{i=1}^n{\bm\Delta}_{i1}\hat{\mathbf r}_1(\bm\Gamma\T\Y_i,\X_i) = \dfrac{1}{n(n-1)}\sum_{i\neq j}K_h'(\bm\Gamma\T\Y_i - \bm\Gamma\T\Y_j)K_h(\X_i - \X_j)\{{\bm\Delta}_{i1} + \bm\Delta_{j1}\}+O_p(n^{-1})$$
	as a second order U-statistic with kernel function $K_h'(\bm\Gamma\T\Y_i - \bm\Gamma\T\Y_j)K_h(\X_i - \X_j)\{{\bm\Delta}_{i1} + \bm\Delta_{j1}\}.$
	Hence, by Lemma 5.2.1.A of \cite{serfling2009approximation}, the degenerated U-statistic has the rate
	$$\dfrac{1}{n}\sum_{i=1}^n{\bm\Delta}_{i1}\big[\hat{\mathbf r}_1(\bm\Gamma\T\Y_i,\X_i) - \E\{\hat{\mathbf r}_1(\bm\Gamma\T\Y_i,\X_i)\mid\bm\Gamma\T\Y_i,\X_i\}\big] = O_p(n^{-1}h^{-(p+u+1)}).$$
	
	Notice that $\E\{\hat{\mathbf r}_1(\bm\Gamma\T\Y_i,\X_i)\mid\bm\Gamma\T\Y_i,\X_i\} - {\mathbf r}_1(\bm\Gamma\T\Y_i,\X_i)$ is the bias term in nonparametric regression, hence 
	$$\sup_{\X,\Y}|\E\{\hat{\mathbf r}_1(\bm\Gamma\T\Y,\X)\mid\bm\Gamma\T\Y,\X\} - {\mathbf r}_1(\bm\Gamma\T\Y,\X)| = O_p(h^2).$$
	Therefore 
	\begin{equation*}
		\begin{aligned}
			\dfrac{1}{n}\sum_{i=1}^n{\bm\Delta}_{i1}\big[\E\{\hat{\mathbf r}_1(\bm\Gamma\T\Y_i,\X_i)\mid\bm\Gamma\T\Y_i,\X_i\} - {\mathbf r}_1(\bm\Gamma\T\Y_i,\X_i)\big] = O_p(n^{-1/2}h^2).
		\end{aligned}
	\end{equation*}
	Combine the above results, we have the desired order $O_p(n^{-1/2}h^2+h^{4}+n^{-1}h^{-(p+u+1)}\log n)$.
\end{proof}
\begin{proof}[Proof of Theorem \ref{thm: global}]
	Similar to the proof of Theorem 3, we use Lemma 6 and 7 to prove the consistency and asymptotic normality. We only check condition (iii) in Lemma 7. All other conditions can be proved using the same way as the proof of Theorem 2. Consider the functions $\hat Q_n(\bm\theta)$ and $Q_0(\bm\theta)$ as
	
	$$Q_0(\bm\theta) = \dfrac{1}{2}\left\{\frac{1}{n}\sum_{i = 1}^n {S}_{\rm eff}(\Y_i, \X_i, \eta_{1,2,3}, {\bm \Delta}_1, {\bm \Delta}_2; {\bm\theta})\right\}^2$$
	$$\hat Q_n(\bm\theta) = \dfrac{1}{2}\left\{\frac{1}{n}\sum_{i = 1}^n \hat{S}_{\rm eff}(\Y_i, \X_i, \hat\eta_{1,2,3}, \hat{\bm \Delta}_1, \hat{\bm \Delta}_2; {\bm\theta})\right\}^2.$$
	We want to show that $\sqrt n{\partial \hat Q_n(\bm\theta_0)}/{\partial \bm\theta}$ is asymptotically normal.
	
	\begin{align*}
		&\sqrt n\frac{\partial \hat Q_n(\bm\theta_0)}{\partial \bm\theta} = \frac{1}{\sqrt n}\sum_{i = 1}^n \hat{S}_{\rm eff}(\Y_i, \X_i, \hat\eta_{1,2,3}, \hat{\bm \Delta}_1, \hat{\bm \Delta}_2; {\bm\theta_0})\cdot \frac{1}{n}\sum_{i=1}^n\dfrac{\partial}{\partial\bm\theta}\hat S_{\rm eff}(\Y_i,\X_i, \hat\eta_{1,2,3}, \hat{\bm \Delta}_1, \hat{\bm \Delta}_2;\bm\theta_0),
	\end{align*}
	and 
	$$\frac{1}{n}\sum_{i=1}^n\dfrac{\partial}{\partial\bm\theta}\hat S_{\rm eff}(\Y_i,\X_i, \hat\eta_{1,2,3}, \hat{\bm \Delta}_1, \hat{\bm \Delta}_2;\bm\theta_0)\overset{p}{\rightarrow}\E\left\{\dfrac{\partial}{\partial\bm\theta}{S}_{\rm eff}(\Y_i, \X_i, \eta_{1,2,3}, {\bm \Delta}_1, {\bm \Delta}_2; {\bm\theta})\right\}.$$
	By regularity condition (B1), the order of differentiation and integration can be exchanged. Hence, we have $$\E\left\{\dfrac{\partial}{\partial\bm\theta}{S}_{\rm eff}(\Y_i, \X_i, \eta_{1,2,3}, {\bm \Delta}_1, {\bm \Delta}_2; {\bm\theta})\right\} = -\E\{{S}_{\rm eff}(\Y_i, \X_i, \eta_{1,2,3}, {\bm \Delta}_1, {\bm \Delta}_2; {\bm\theta})^{\otimes 2}\}.$$
	By Slutsky's Theorem, we only need to show $n^{-1/2}\sum_{i = 1}^n \hat{S}_{\rm eff}(\Y_i, \X_i, \hat\eta_{1,2,3}, \hat{\bm \Delta}_1, \hat{\bm \Delta}_2; {\bm\theta_0})$ converges to a normal distribution. Also, because 
	$$\frac{1}{\sqrt n}\sum_{i = 1}^n {S}_{\rm eff}(\Y_i, \X_i, \eta_{1,2,3}, {\bm \Delta}_1, {\bm \Delta}_2; {\bm\theta_0}) \overset{d}{\rightarrow}\mathcal N\bigg[\bm 0, \E\big\{S_{\rm eff}(\Y_i,\X_i, \eta_1, \eta_2, \eta_3;\bm\theta_0)^{\otimes2}\big\}\bigg],$$
    we only need to show 
	{\small$$\frac{1}{\sqrt n}\sum_{i = 1}^n \hat{S}_{\rm eff}(\Y_i, \X_i, \hat\eta_{1,2,3}, \hat{\bm \Delta}_1, \hat{\bm \Delta}_2; {\bm\theta_0})-\frac{1}{\sqrt n}\sum_{i = 1}^n {S}_{\rm eff}(\Y_i, \X_i, \eta_{1,2,3}, {\bm \Delta}_1, {\bm \Delta}_2; {\bm\theta_0})\overset{p}{\rightarrow}0.$$}
	Since $S_{\text{eff}}^* = (S_{\text{eff}, \bm\gamma}^*, S_{\text{eff}, \mathbf b}^*)$, and the proof for each component are similar, we only prove the convergence in $S_{\text{eff}, \mathbf b}^*$.
	
	Let $\bm\Gamma$ and $\mathbf B$ denote the orthogonal basis derived from $\bm\theta_0$. Then,
	
	\begin{align*}
	    &\frac{1}{\sqrt n}\sum_{i = 1}^n \hat{S}_{\rm{eff},\mathbf b}(\Y_i, \X_i, \hat\eta_{1,2,3}, \hat{\bm \Delta}_1, \hat{\bm \Delta}_2; {\bm\theta_0}) - \frac{1}{\sqrt n}\sum_{i = 1}^n {S}_{\rm{eff},\mathbf b}(\Y_i, \X_i, \eta_{1,2,3}, {\bm \Delta}_1, {\bm \Delta}_2; {\bm\theta_0})\\
	    =& n^{-{1}/{2}}\sum_{i=1}^n\bigg[\big\{\hat{\bm\Delta}_{i2}({\bm\theta_0})-{\bm\Delta}_{i2}({\bm\theta_0})\big\}\dfrac{\partial \log\eta_3}{\partial ({\mathbf B}_0\T{\bm\Gamma}_0\T\Y_i)\T}\bigg]\\
	&\quad + n^{-{1}/{2}}\sum_{i=1}^n\bigg[\big\{\hat{\bm\Delta}_{i2}({\bm\theta_0})-{\bm\Delta}_{i2}({\bm\theta_0})\big\}\bigg\{\dfrac{\partial \log\hat\eta_3}{\partial ({\mathbf B}_0\T{\bm\Gamma}_0\T\Y_i)\T} - \dfrac{\partial \log\eta_3}{\partial ({\mathbf B}_0\T{\bm\Gamma}_0\T\Y_i)\T}\bigg\}\bigg]\\
	&\quad + n^{-{1}/{2}}\sum_{i=1}^n\bigg[{\bm\Delta}_{i2}({\bm\theta_0})\bigg\{\dfrac{\partial \log\hat\eta_3}{\partial ({\mathbf B}_0\T{\bm\Gamma}_0\T\Y_i)\T} - \dfrac{\partial \log\eta_3}{\partial ({\mathbf B}_0\T{\bm\Gamma}_0\T\Y_i)\T}\bigg\}\bigg].
	\end{align*}

By Lemma 5, the first term is $O_p(h^2 + n^{1/2}h^4 + n^{-1/2}h^{-(r+p-u-d)}\log n)$ and the third term is $O_p(h^2 + n^{1/2}h^4 + n^{-1/2}h^{-(r-u-d+1)}\log n)$. By the uniform convergence theorem in nonparametric regression (Theorem 2.6 in \cite{li2007nonparametric} and Theorem 6 in \cite{hansen2008uniform}), we have
$$\sup_{\X_i,\Y_i}|\hat{\bm\Delta}_{i2}({\bm\theta_0})- {\bm\Delta}_{i2}({\bm\theta_0})| = O_p\bigg\{\left(\frac{\log n}{nh^{(r+p-u-d)}}\right)^{1/2}+h^2\bigg\},$$
and 
$$\sup_{\X_i,\Y_i}\bmid\dfrac{\partial \log\hat\eta_3}{\partial ({\mathbf B}_0\T{\bm\Gamma}_0\T\Y_i)\T} - \dfrac{\partial \log\eta_3}{\partial ({\mathbf B}_0\T{\bm\Gamma}_0\T\Y_i)\T}\bmid = O_p\bigg\{\left(\frac{\log n}{nh^{r-u-d+2}}\right)^{1/2}+h^2\bigg\}.$$
Hence, the second term can be bounded by $O_p\{n^{-1/2}h^{r-u-d+1+p/2}\log n + n^{1/2}h^4\} = o_p(1)$. Therefore, 
$$\frac{1}{\sqrt n}\sum_{i = 1}^n \hat{S}_{\rm{eff},\mathbf b}(\Y_i, \X_i, \hat\eta_{1,2,3}, \hat{\bm \Delta}_1, \hat{\bm \Delta}_2; {\bm\theta_0})-\frac{1}{\sqrt n}\sum_{i = 1}^n {S}_{\rm{eff},\mathbf b}(\Y_i, \X_i, \eta_{1,2,3}, {\bm \Delta}_1, {\bm \Delta}_2; {\bm\theta_0}) \overset{p}{\rightarrow}0 .$$
Hence, by Slutsky's Theorem,
$$\sqrt n\frac{\partial \hat Q_n(\bm\theta_0)}{\partial \bm\theta}\overset{d}{\rightarrow}\mathcal N(0, \mathcal V_{\bm\theta}^{-3})$$
where $\mathcal V_{\bm\theta} = \E\{{S}_{\rm eff}(\Y_i, \X_i, \eta_{1,2,3}, {\bm \Delta}_1, {\bm \Delta}_2; {\bm\theta})^{\otimes 2}\}^{-1}$.

Similarly, we can prove that 
$$\nabla_{\bm\theta\bm\theta}\hat{Q}_n(\bm\theta_0) 
	   \overset{p}{\rightarrow}\mathcal V_{\bm\theta}^{-2} = \mathbf H.$$
Therefore, by Lemma 7, we have 
$$\sqrt n (\hat{\bm\theta} - \bm\theta_0)\overset{d}{\rightarrow}\mathcal N(\bm 0,\mathcal V_{\bm\theta}),$$
which achieves the semiparametric efficiency bound.
\end{proof}
\begin{customthm}{6}
	(Theorem 2.1 in \cite{newey1994large}) Suppose there is a function $Q_0(\bm\theta)$ such that (i) $Q_0(\bm\theta)$ is uniquely minimized at $\bm\theta_0$, (ii) $\bm\Theta$ is compact, (iii) $Q_0(\bm\theta)$ is continuous, (iv) $\hat Q_n(\bm\theta)$ converges uniformly in probability to $Q_0(\bm\theta)$, then $\hat Q_n(\bm\theta)$, $\hat{\bm\theta}\overset{p}{\rightarrow}\bm\theta_0$.
	\end{customthm}
	\begin{customthm}{7}
	(Theorem 3.1 in \cite{newey1994large}) Suppose that $\hat{\bm\theta}$ is a minimizer of $\hat Q_n(\bm\theta)$, $\hat{\bm\theta}\overset{p}{\rightarrow}\bm\theta_0$ and (i) $\bm\theta_0\in\text{interior}(\bm\Theta)$, (ii) $\hat Q_n(\bm\theta)$ is twice continuously differentiable in a neighborhood $\mathcal N_\epsilon$ of $\bm\theta_0$, (iii) $\sqrt n \nabla_{\bm\theta}\hat Q_n(\bm\theta_0)\overset{d}{\rightarrow}\mathcal N(\bm0,\bm\Sigma)$, (iv) there is $H(\bm\theta)$ that is continuous at $\bm\theta_0$ and $\sup_{\bm\theta\in\mathcal N_\epsilon}\|\nabla_{\bm\theta\bm\theta}\hat Q_n(\bm\theta) - H(\bm\theta)\|\overset{p}{\rightarrow}\bm 0$, (v) $\mathbf H = H(\bm\theta_0)$ is nonsingular. Then, $\sqrt n (\hat{\bm\theta} - \bm\theta_0)\overset{d}{\rightarrow}\mathcal N(\bm0, \mathbf H^{-1}\bm\Sigma\mathbf H^{-1})$.
	\end{customthm}
	\section{Additional Materials}
	\subsection{MLE of the regression parameter under the inner envelope model}
	In this part, we present the MLE of the regression parameter $\bm\beta$ under the inner envelope model. Here, we assume the inner envelope spaces $\hat{\bm\Gamma}$, $\hat{\bm\Gamma}_0$, $\hat{\mathbf B}$ and $\hat{\mathbf B}_0$ are already calculated. 
	
	Let $\mathbf S_{\text{fit}}$ and $\mathbf S_{\text{res}}$ denote the sample covariance matrices of the fitted and residual vectors from the OLS fit of $\Y$ on $\X$. Let $\tilde\lambda_i(\mathbf G_0)$ denote the ordered, descending eigenvalues of $(\mathbf G_0\T\mathbf S_{\text{res}}\mathbf G_0)^{-1/2}(\mathbf G_0\T\mathbf S_{\text{fit}}\mathbf G_0)(\mathbf G_0\T\mathbf S_{\text{res}}\mathbf G_0)^{-1/2}$, where $\mathbf G_0\in\mathbb R^{r\times (r-u)}$ is a semi-orthogonal matrix. Also, we denote the matrices of ordered eigenvectors and eigenvalues as $\tilde{\mathbf V}(\mathbf G_0)$ and $\tilde{\bm\Lambda}(\mathbf G_0)=\text{diag}\{\tilde\lambda_1(\mathbf G_0),\ldots, \tilde\lambda_{r-u}(\mathbf G_0)\}$, and let $\tilde{\mathbf \K}(\mathbf G_0)=\text{diag}\{0,\ldots,0,\tilde\lambda_{p-u+1}(\mathbf G_0),\ldots, \tilde\lambda_{r-u}(\mathbf G_0)\}$. Let $\mathbf F$ denote the $n\times p$ matrix with $i$th row $\X_i\T$, let $\mathbf U$ be the $n\times r$ matrix with $i$th row $(\Y_i-\bar\Y)\T$, and let $\hat{\bm\beta}_{\text{ols}}$ denote the MLE of $\bm\beta$ under the standard model \eqref{eq: lm}. Then, 
	\begin{align*}
		&\hat{\bm\zeta}_1\T = \hat{\bm\Gamma}\T\hat{\bm\beta}_{\text{ols}},\\
		&\hat{\bm\Omega}_1 = (\mathbf U\hat{\bm\Gamma} - \mathbf F \hat{\bm\zeta}_1)\T(\mathbf U\hat{\bm\Gamma} - \mathbf F \hat{\bm\zeta}_1)/n,\\
		&\hat{\bm\Omega}_0 = \hat{\bm\Gamma}_0\mathbf S_{\text{res}}\hat{\bm\Gamma}_0 + (\hat{\bm\Gamma}_0\mathbf S_{\text{res}}\hat{\bm\Gamma}_0)^{1/2}\tilde{\mathbf V}(\hat{\bm\Gamma}_0)\tilde{\mathbf K}(\hat{\bm\Gamma}_0)\tilde{\mathbf V}(\hat{\bm\Gamma}_0)
		(\hat{\bm\Gamma}_0\mathbf S_{\text{res}}\hat{\bm\Gamma}_0)^{1/2}\\
		&\text{span}(\hat\B) = \hat{\bm\Omega}_0\mathcal S_{p-d}(\hat{\bm\Omega}_0, \hat{\bm\Gamma}_0\T)\\
		&\hat{\bm\zeta}_2\T = (\hat{\B}\T\hat{\bm\Omega}^{-1}_0\hat{\B})^{-1}\hat{\B}\T\hat{\bm\Omega}^{-1}_0\hat{\bm\Gamma}_0\T\hat{\bm\beta}_{\text{ols}},\\
		&\hat{\bm\beta} = \hat{\bm\Gamma}\hat{\bm\zeta}_1\T + \hat{\bm\Gamma}_0\hat{\bm\Omega}_0\hat{\bm\zeta}_2\T,\\
		&\hat{\bm\Sigma} = \hat{\bm\Gamma}\hat{\bm\Omega}_1\hat{\bm\Gamma}\T + \hat{\bm\Gamma}_0\hat{\bm\Omega}_0\hat{\bm\Gamma}_0\T,
	\end{align*}
	where $\mathcal S_k(\mathbf A, \mathbf B)$ is the span of $\mathbf A^{-1/2}$ times the first $k$ eigenvectors of $\mathbf A^{-1/2}\mathbf B\mathbf A^{-1/2}$.
	Detailed derivations are carried out by Su and Cook \cite{su2012inner}.
	\subsection{Nonparametric density estimation and nonparametric regression} The derivative of the log densities are estimated by 
\begin{align}
	&\frac{\partial \log\hat\eta_1}{\partial (\bm\Gamma\T\y)\T} =  \dfrac{\sum_{i=1}^nK_h'(\bm\Gamma^{\mathsf T}\Y_i-\bm\Gamma\T\y)K_h(\X_i-\x)}{\sum_{i=1}^nK_h(\bm\Gamma^{\mathsf T}\Y_i-\bm\Gamma\T\y)K_h(\X_i-\x)},\\
	&\frac{\partial \log\hat\eta_2}{\partial (\B\T\bm\Gamma_0\T\y)\T} =  \dfrac{\sum_{i=1}^nK_h'\{\B\T\bm\Gamma_0^{\mathsf T}(\Y_i-\y)\}K_h\{\B_0\T\bm\Gamma_0^{\mathsf T}(\Y_i-\y)\}K_h(\X_i-\x)}{\sum_{i=1}^nK_h\{\B\T\bm\Gamma_0^{\mathsf T}(\Y_i-\y)\}K_h\{\B_0\T\bm\Gamma_0^{\mathsf T}(\Y_i-\y)\}K_h(\X_i-\x)},\\
	&\frac{\partial \log\hat\eta_3}{\partial (\B_0\T\bm\Gamma_0\T\y)\T} =  \dfrac{\sum_{i=1}^nK_h'\{\B_0\T\bm\Gamma_0^{\mathsf T}(\Y_i-\y)\}}{\sum_{i=1}^nK_h\{\B_0\T\bm\Gamma_0^{\mathsf T}(\Y_i-\y)\}}.
	\end{align}
The nonparametric regression $\bm\Delta_1$ and $\bm\Delta_2$ are estimated by 
\begin{align}
	\hat{\bm\Delta}_1 =& \y - \frac{\sum_{i=1}^N \Y_iK_{h}(\X_i-\x)}{\sum_{i=1}^NK_{h}(\X_i-\x)},\nonumber\\
	\hat{\bm\Delta}_2 =&\dfrac{\sum_{i=1}^{N}\mathbf P_{\bm\Gamma_0\B}\Y_iK_h\{\B_0\T\bm\Gamma_0\T(\Y_i -\y) \}K_h(\X_i-\x)}{\sum_{i=1}^{N}K_h\{\B_0\T\bm\Gamma_0\T(\Y_i -\y) \}K_h(\X_i-\x)}\\ &\quad - \dfrac{\sum_{i=1}^{N}\mathbf P_{\bm\Gamma_0\B}\Y_iK_h\{\B_0\T\bm\Gamma_0\T(\Y_i -\y) \}}{\sum_{i=1}^{N}K_h\{\B_0\T\bm\Gamma_0\T(\Y_i -\y) \}}\nonumber.
\end{align}
The nonparametric regression for the conditional expectations of $\eta_i^*$ are estimated by
\begin{align}
    &\hat{\E}\left(\dfrac{\partial \log\eta_1^*}{\partial (\bm\Gamma^{\mathsf T}\Y_i)\T}\bmid\x\right) = \frac{\sum_{i=1}^N {\partial \log\eta_1^*}/{\partial (\bm\Gamma^{\mathsf T}\Y_i)\T} K_{h}(\X_i-\x)}{\sum_{i=1}^NK_{h}(\X_i-\x)},\\
    &\hat{\E}\left(\dfrac{\partial \log\eta_2^*}{\partial (\B\T\bm\Gamma_0^{\mathsf T}\Y_i)\T}\bmid\B_0\T\bm\Gamma_0^{\mathsf T}\y,\x\right) = \frac{\sum_{i=1}^N {\partial \log\eta_2^*}/{\partial (\B\T\bm\Gamma_0^{\mathsf T}\Y_i)\T} K_{h}\{\B_0\T\bm\Gamma_0^{\mathsf T}(\y_i-\y)\}K_{h}(\X_i-\x)}{\sum_{i=1}^NK_{h}\{\B_0\T\bm\Gamma_0^{\mathsf T}(\y_i-\y)\}K_{h}(\X_i-\x)},\\
    &\hat{\E}\left(\dfrac{\partial \log\eta_2^*}{\partial (\B_0\T\bm\Gamma_0^{\mathsf T}\Y_i)\T}\bmid\B_0\T\bm\Gamma_0^{\mathsf T}\y,\x\right) = \frac{\sum_{i=1}^N {\partial \log\eta_2^*}/{\partial (\B_0\T\bm\Gamma_0^{\mathsf T}\Y_i)\T} K_{h}\{\B_0\T\bm\Gamma_0^{\mathsf T}(\y_i-\y)\}K_{h}(\X_i-\x)}{\sum_{i=1}^NK_{h}\{\B_0\T\bm\Gamma_0^{\mathsf T}(\y_i-\y)\}K_{h}(\X_i-\x)},\\
    &\hat{\E}\left(\dfrac{\partial \log\eta_3^*}{\partial (\B_0\T\bm\Gamma_0^{\mathsf T}\Y_i)\T}\right) = \dfrac{1}{n}\sum_{i=1}^N\dfrac{\partial \log\eta_3^*}{\partial (\B_0\T\bm\Gamma_0^{\mathsf T}\Y_i)\T}.
\end{align}
\subsection{Simulation: linear model with additive, normal errors}
For the out-of-sample predictive RMSE results in Section 6.1, the oracle, InnEnv, local, global, GMM, OLS, Env and PLS methods have prediction RMSE equal to 1.82, 1.82, 1.83, 1.83, 1.83, 1.91, 1.97, 1.96 respectively.

\subsection{Real data: synthetic dataset from iris data}
The estimated inner envelope spaces obtained from the globally efficient algorithm is 
\begin{equation*}
	\hat{\bm\Gamma} = \begin{pmatrix}
		-0.01  \\
-0.02  \\
0.75  \\
-0.46 \\
0.35  \\
0.33
	\end{pmatrix},\hspace{5mm} \hat{\bm\Gamma}_0\hat\B = \begin{pmatrix}
		0  \\
0  \\
-0.50  \\
0.08 \\
0.66  \\
0.55
	\end{pmatrix},\hspace{5mm} \hat{\bm\Gamma}_0\hat\B_0 = \begin{pmatrix}
		-0.98&  0.08&  -0.14 & 0  \\
		0 & -0.86 & -0.50& 0.01  \\
		0&  0& -0.04& -0.43  \\
		0&  0&  0& -0.88 \\
		0.11& 0.32& -0.56& -0.12  \\
		-0.13&  -0.38&  0.64& -0.12
	\end{pmatrix}.
\end{equation*}

Because $\mathcal S_0$ and $\mathcal{S}_3$ are of different dimension, we consider the spaces $\mathbf P_{\mathcal S_3}\mathcal S_0$ and $\mathcal S_3^{\perp}$. The space $\mathbf P_{\mathcal S_3}\mathcal S_0$ should be close to $\mathcal S_0$, and $\mathcal S_3^{\perp}$ should be perpendicular to $\mathcal S_0$. It turns out that ${\rm dist}(\mathbf P_{\mathcal S_3}\mathcal S_0, \mathcal S_0) = 0.029$ and ${\rm dist}(\mathcal S_3^\perp, \mathcal S_0)=1.999$. Since the upper bound for ${\rm dist}(\mathcal S_3^\perp, \mathcal S_0)$ is 2, it indicates that $\mathcal S_0$ is contained in $\mathcal S_3$. That is, our method successfully identified the artificial noise space $\mathcal S_0$. 

\section{Tables and Figures} \label{sec: numerical_results}

\begin{table}[!htb]
	\centering
	\begin{tabular}{||c | c c c c c||} 
	\hline
	n & & InnEnv & Local & Global & GMM\\ [0.5ex] 
	\hline\hline
	100 & $dist(\mathcal S_1,\hat {\mathcal S}_1)$ &0.111& 0.134 & 0.166 & 0.224\\ 
	\hline
	& $dist(\mathcal S_3,\hat {\mathcal S}_3)$ & 0.089 & 0.121 & 0.152 & 0.201\\
	\hline
	300 & $dist(\mathcal S_1,\hat {\mathcal S}_1)$ &0.085& 0.101 & 0.125 & 0.161\\ 
	\hline
	& $dist(\mathcal S_3,\hat {\mathcal S}_3)$ & 0.062 & 0.078 & 0.102 & 0.119\\
	\hline
	500 & $dist(\mathcal S_1,\hat {\mathcal S}_1)$ &0.070& 0.089 & 0.114 & 0.141\\ 
	\hline
	& $dist(\mathcal S_3,\hat {\mathcal S}_3)$ & 0.050 & 0.064 & 0.087 & 0.101\\
	\hline
	750 & $dist(\mathcal S_1,\hat {\mathcal S}_1)$ &0.055& 0.071 & 0.095 & 0.112\\ 
	\hline
	& $dist(\mathcal S_3,\hat {\mathcal S}_3)$ & 0.041 & 0.057 & 0.072 & 0.084\\
	\hline
	1000 & $dist(\mathcal S_1,\hat {\mathcal S}_1)$ &0.045& 0.060 & 0.084 & 0.099\\ 
	\hline
	& $dist(\mathcal S_3,\hat {\mathcal S}_3)$ & 0.035 & 0.044 & 0.065 & 0.073\\
	\hline
   \end{tabular}
   \caption{Mean distance of the estimated space and the true space (linear case)}
   \label{tb: dist1}
   \end{table}

   \begin{table}[!htb]
	\centering
	\begin{tabular}{||c |c c c c c c c c||} 
	\hline
	Size&  Oracle & InnEnv &  Local & Global & GMM & OLS & Env &  PLS\\ [0.5ex] 
	\hline\hline
	100 & 0.136 & 0.162 & 0.190 & 0.212 & 0.244 & 0.669 & 0.705 & 0.751\\ 
	\hline
	300 & 0.077 & 0.088 & 0.097 & 0.115 & 0.135 & 0.378 & 0.403 & 0.460\\ 
	\hline
	500 & 0.060 & 0.065 & 0.071 & 0.084 & 0.102 & 0.287 & 0.386 & 0.317\\ 
	\hline
	750 & 0.049 & 0.054 & 0.063 & 0.074 & 0.085 & 0.240 & 0.346 & 0.278 \\
	\hline
	1000  &0.042 & 0.049 & 0.056 & 0.068 & 0.077 & 0.213 & 0.310 & 0.236\\ 
	\hline
   \end{tabular}
   \caption{Mean and standard errors of $\|\hat{\bm\beta} - \bm\beta\|_F^2$ in Scenario 1 with different sample sizes.}
   \label{tb: simu2}
   \end{table}
   
   \begin{table}[!htb]
	\centering
	\begin{tabular}{||c|c c c c c||} 
	\hline
	Size & &InnEnv & Local & Global & GMM\\ [0.5ex] 
	\hline\hline
	100 &$dist(\mathcal S_1,\hat {\mathcal S}_1)$ &1.175& 0.542 & 0.360 & 0.745\\ 
	\hline
	 &$dist(\mathcal S_3,\hat {\mathcal S}_3)$ & 1.028 & 0.382 & 0.274 & 0.507\\
	\hline
	300 &$dist(\mathcal S_1,\hat {\mathcal S}_1)$ &1.192& 0.429 & 0.277 & 0.589\\ 
	\hline
	 &$dist(\mathcal S_3,\hat {\mathcal S}_3)$ & 1.052 & 0.259 & 0.186 & 0.344\\
	\hline
	500 &$dist(\mathcal S_1,\hat {\mathcal S}_1)$ &1.213& 0.360 & 0.236 & 0.502\\ 
	\hline
	 &$dist(\mathcal S_3,\hat {\mathcal S}_3)$ & 1.068 & 0.209 & 0.142 & 0.277\\
	\hline
	750 &$dist(\mathcal S_1,\hat {\mathcal S}_1)$ &1.201& 0.297 & 0.192 & 0.412\\ 
	\hline
	 &$dist(\mathcal S_3,\hat {\mathcal S}_3)$ & 1.089 & 0.177 & 0.117 & 0.234\\
	\hline
	1000 &$dist(\mathcal S_1,\hat {\mathcal S}_1)$ &1.189& 0.265 & 0.168 & 0.368\\ 
	\hline
	 &$dist(\mathcal S_3,\hat {\mathcal S}_3)$ & 1.090 & 0.154 & 0.101 & 0.203\\
	\hline
   \end{tabular}
   \caption{Mean distance of the estimated space and the true space (nonlinear case)}
   \label{tb: dist2}
\end{table}

\begin{table}[!htb]
	\centering
	\begin{tabular}{||c|c c c c c||} 
	\hline
	Size & &InnEnv & Local & Global & GMM\\ [0.5ex] 
	\hline\hline
	100 &$dist(\mathcal S_1,\hat {\mathcal S}_1)$ &0.992& 0.385 & 0.320 & 0.425\\ 
	\hline
	 &$dist(\mathcal S_3,\hat {\mathcal S}_3)$ & 0.998 & 0.498 & 0.401 & 0.552\\
	\hline
	300 &$dist(\mathcal S_1,\hat {\mathcal S}_1)$ &0.984& 0.225 & 0.192 & 0.268\\ 
	\hline
	 &$dist(\mathcal S_3,\hat {\mathcal S}_3)$ & 0.920 & 0.288 & 0.238 & 0.340\\
	\hline
	500 &$dist(\mathcal S_1,\hat {\mathcal S}_1)$ &0.917& 0.175 & 0.147 & 0.208\\ 
	\hline
	 &$dist(\mathcal S_3,\hat {\mathcal S}_3)$ & 0.903 & 0.223 & 0.182 & 0.263\\
	\hline
	750 &$dist(\mathcal S_1,\hat {\mathcal S}_1)$ &0.980& 0.140 & 0.117 & 0.169\\ 
	\hline
	 &$dist(\mathcal S_3,\hat {\mathcal S}_3)$ & 0.913 & 0.182 & 0.150 & 0.215\\
	\hline
	1000 &$dist(\mathcal S_1,\hat {\mathcal S}_1)$ &0.932& 0.132 & 0.105 & 0.147\\ 
	\hline
	 &$dist(\mathcal S_3,\hat {\mathcal S}_3)$ & 0.919 & 0.157 & 0.128 & 0.186\\
	\hline
   \end{tabular}
   \caption{Mean distance of the estimated space and the true space (exponential case)}
   \label{tb: dist3}
\end{table}

\begin{table}[!htb]
	\centering
	\caption{The point estimates, bootstrap standard errors and $p$-values for the regression parameter for the \textit{iris} dataset}
		\begin{tabular}{c|ccc|ccc}
			\hline
			&\multicolumn{3}{c}{Our Method}&\multicolumn{3}{c}{Standard } \\
			\hline
			Corresponding to $X_1$& $\bm{\hat\beta}$ & $\hat {\mathrm{SE}}$   & $p$-value & $\bm{\hat\beta}$ & $\hat {\mathrm{SE}}$  &  $p$-value \\
			\hline
			Noise$_1$ & 0.08 & 0.10 & 0.40 & 0.07 & 0.12  &  0.54 \\ 
			Noise$_2$ & 0.05 & 0.13 & 0.68 & 0.07 & 0.14 &  0.60 \\  
      Sepal length & -1.01 & 0.06 & <0.01 & -1.01 & 0.06  &  <0.01 \\ 
			Sepal width & 0.85 & 0.12 & <0.01 & 0.85 & 0.12  &  <0.01 \\ 
			Pedal length & -1.30 & 0.02 & <0.01 & -1.30 & 0.02 & <0.01 \\ 
			Pedal width & -1.25 & 0.02 & <0.01 & -1.25 & 0.02 & <0.01 \\ 
			\hline	
			Corresponding to $X_2$& $\bm{\hat\beta}$ & $\hat {\mathrm{SE}}$   & $p$-value & $\bm{\hat\beta}$ & $\hat {\mathrm{SE}}$  &  $p$-value \\
			\hline
			Noise$_1$ & 0.03 & 0.03 & 0.27 & 0.07 & 0.13  &  0.59 \\ 
			Noise$_2$ & -0.01 & 0.03 & 0.73 & -0.02 & 0.16 &  0.91 \\  
      Sepal length & 0.10 & 0.07 & 0.12 & 0.14 & 0.10  &  0.17 \\ 
			Sepal width & -0.65 & 0.10 & <0.01 & -0.65 & 0.10  &  <0.01 \\ 
			Pedal length & 0.28 & 0.04 & <0.01 & 0.29 & 0.004 & <0.01 \\ 
			Pedal width & 0.22 & 0.04 & <0.01 & 0.17 & 0.04 & <0.01 \\ 
			\hline	
		\end{tabular}
	\label{tb: iris}
\end{table}
\begin{table}[!htb]
	\centering
	\begin{tabular}{||c c c c||} 
	\hline
	$p$-values & setosa:versicolor & setosa:virginica & versicolor:virginica \\ [0.5ex] 
	\hline\hline
	$\mathcal S_{31}$ &0.396& 0.717 & 0.272 \\ 
	\hline
	$\mathcal S_{32}$ & 0.869 & 0.717 & 0.549 \\
	\hline
	$\mathcal S_{33}$ & 0.869 & 0.396 & 0.869 \\
	\hline
	$\mathcal S_{34}$ & 0.869 & 0.998 & 0.717 \\
	\hline

   \end{tabular}
   \caption{$p$-values for testing whether the distributions of each species is different on $\mathcal S_3$.}
   \label{tb: KS_test}
  \end{table}

\begin{figure}[!h]
	\begin{subfigure}{.32\textwidth}
	  \centering
	  \includegraphics[width=\linewidth]{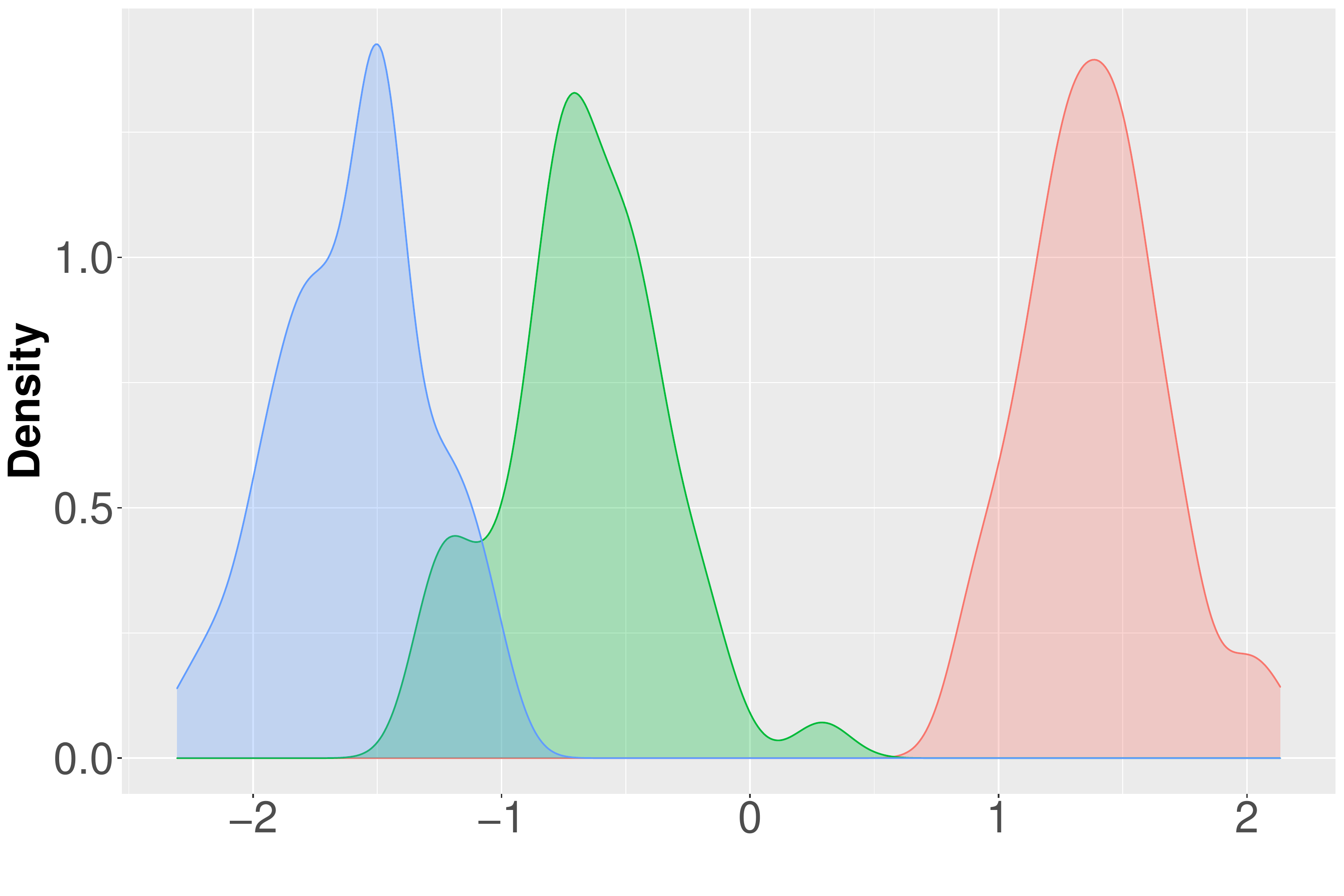}  
	  \caption{${\mathcal S}_1$}
	\end{subfigure}
	\begin{subfigure}{.32\textwidth}
	  \centering
	  \includegraphics[width=\linewidth]{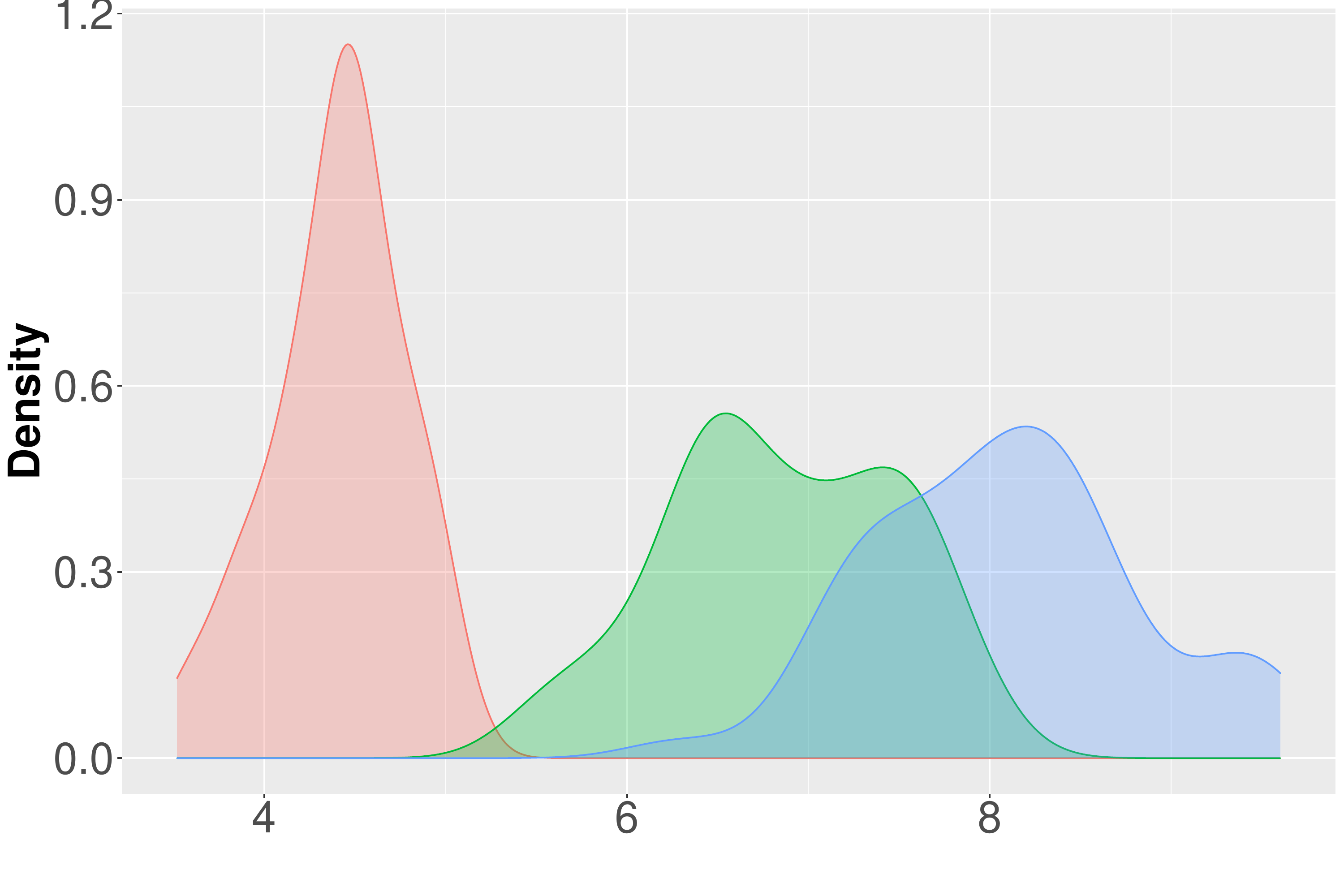}  
	  \caption{${\mathcal S}_2$}
	\end{subfigure}
	\begin{subfigure}{.32\textwidth}
		\centering
		\includegraphics[width=\linewidth]{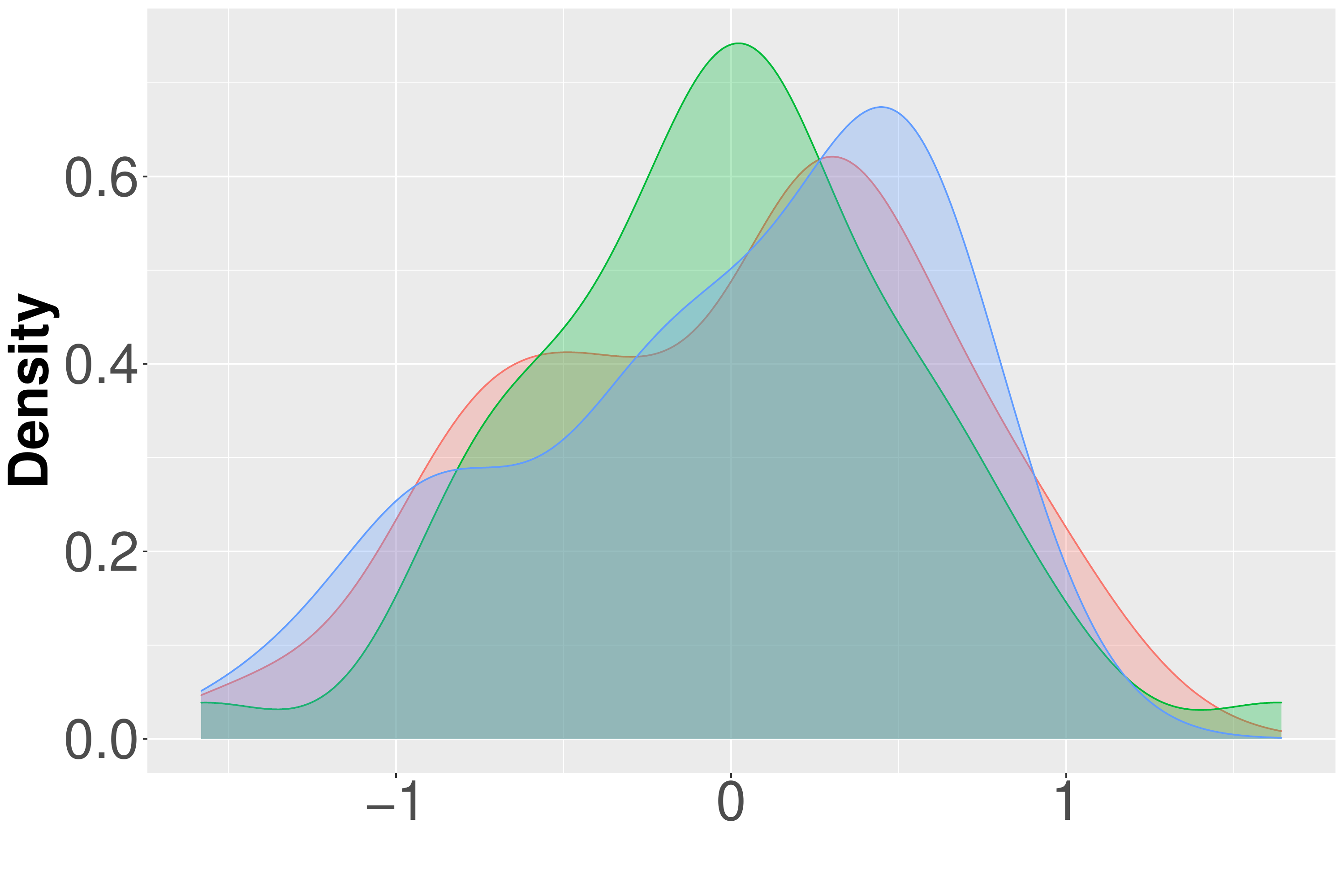}  
		\caption{$\mathcal S_{31}$}
	  \end{subfigure}
	  \begin{subfigure}{.32\textwidth}
		\centering
		\includegraphics[width=\linewidth]{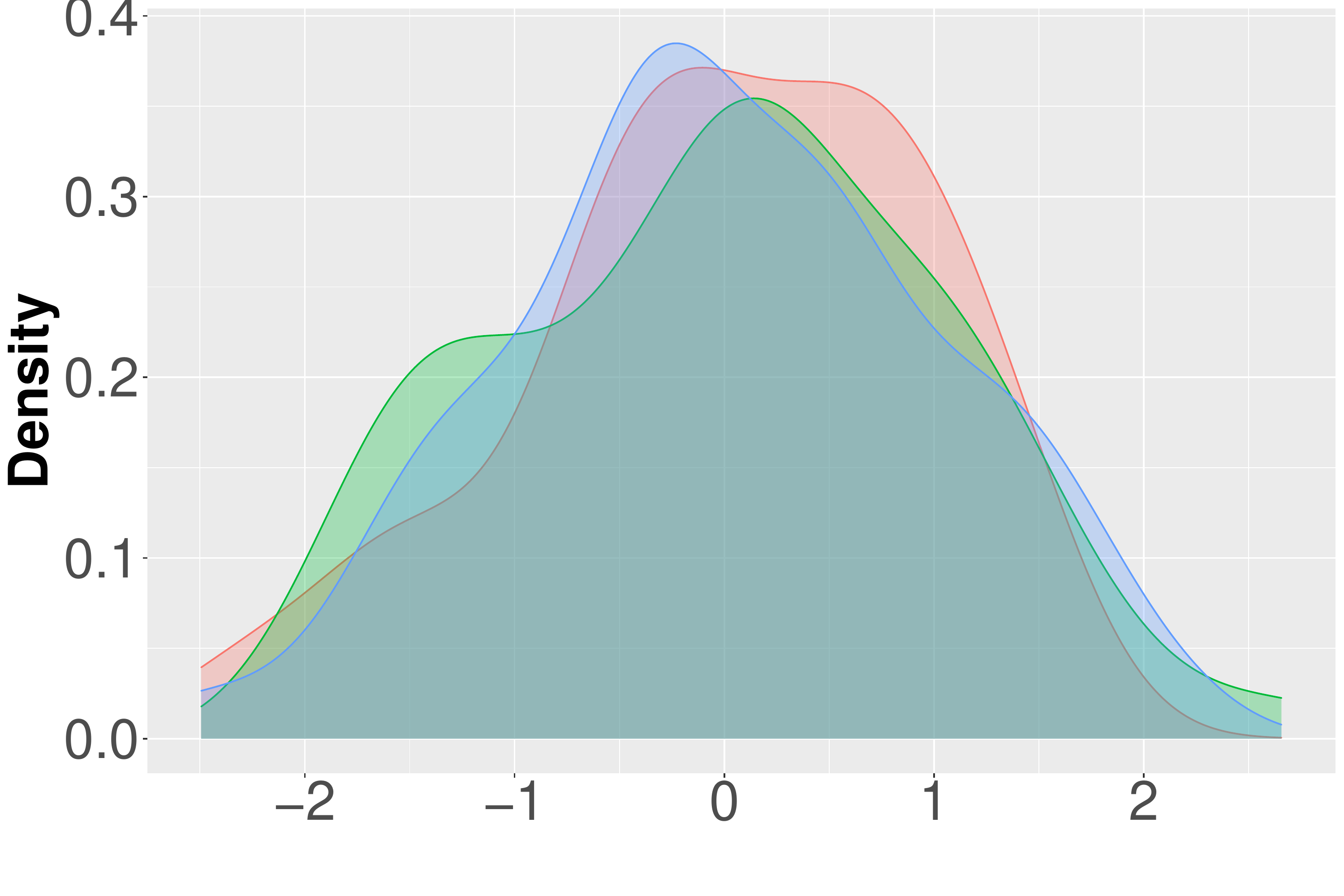}  
		\caption{$\mathcal S_{32}$}
	  \end{subfigure}
	  \begin{subfigure}{.32\textwidth}
		\centering
		\includegraphics[width=\linewidth]{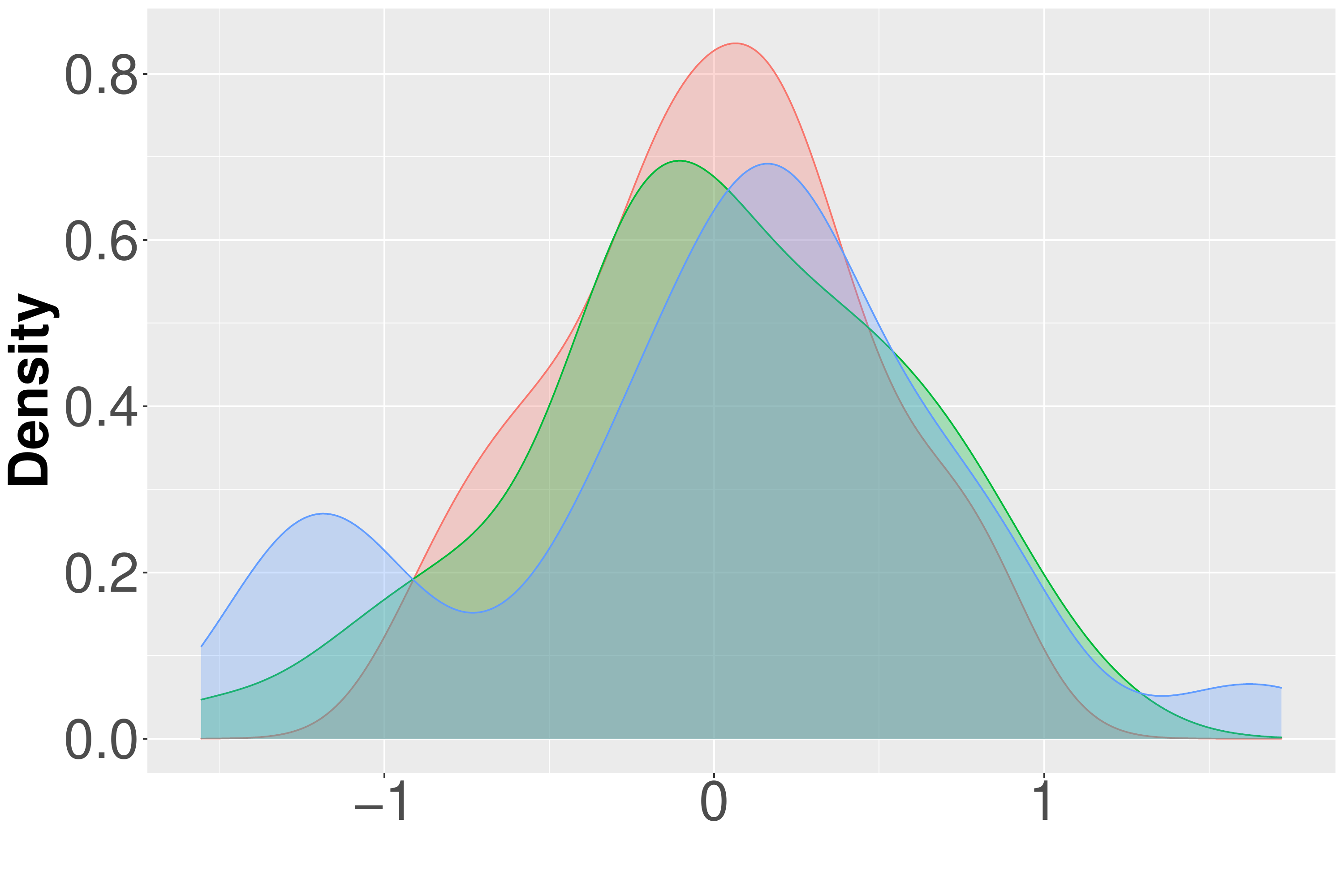}  
		\caption{$\mathcal S_{33}$}
	  \end{subfigure}
	  \begin{subfigure}{.32\textwidth}
		\centering
		\includegraphics[width=\linewidth]{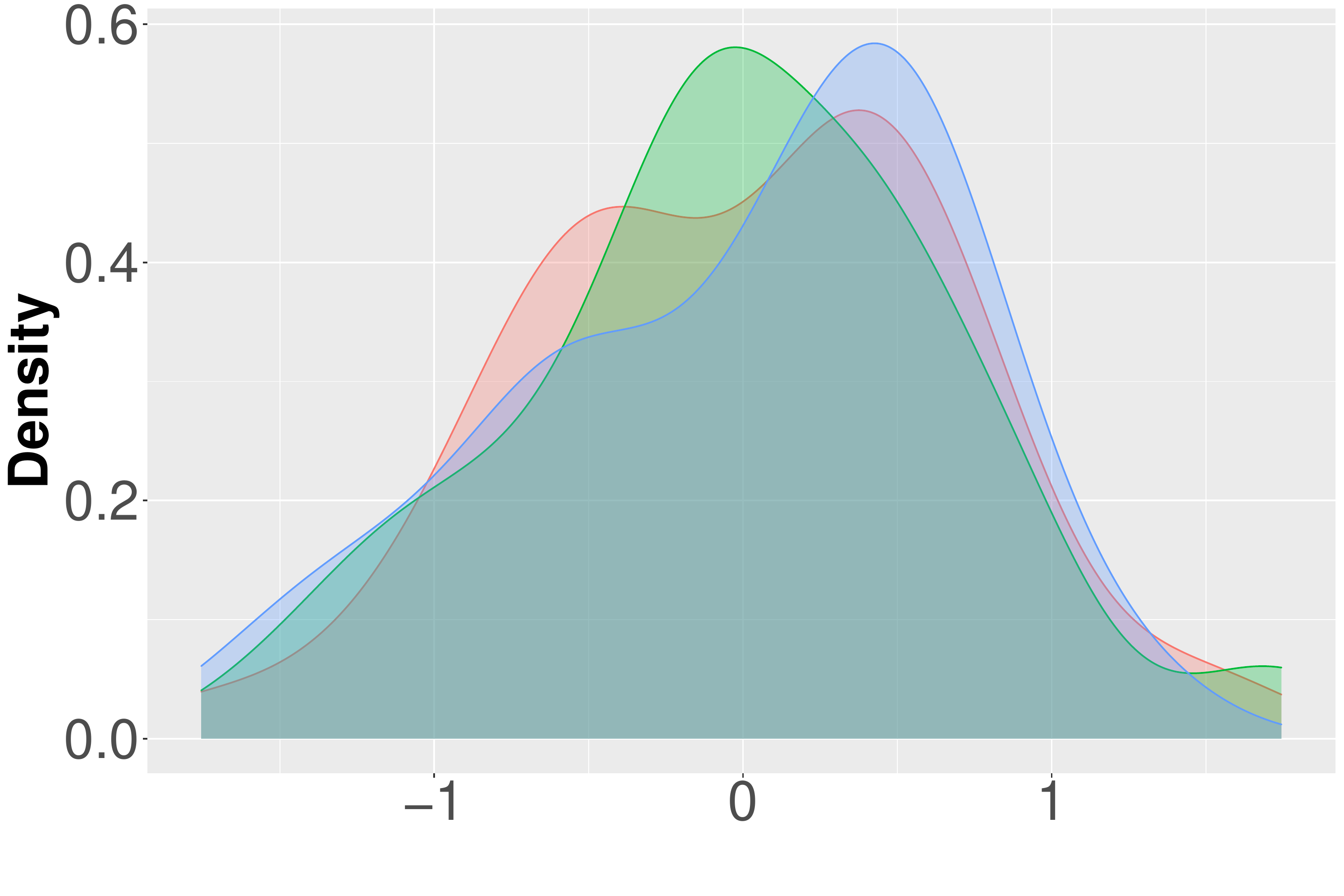}  
		\caption{$\mathcal S_{34}$}
	  \end{subfigure}
	\caption{Density plots for $(\hat{\bm\Gamma}\T\Y, \hat{\bm\Gamma}_0\T\hat\B\T\Y,\hat{\bm\Gamma}_0\T\hat\B_0\T\Y)$ with different groups, where red, blue, green stands for setosa, versicolor and virginica.}
	\label{fig: density}
\end{figure}

\begin{figure}[!h]
	\centering
	\begin{subfigure}{.4\textwidth}
	  \centering
	  \includegraphics[width=\linewidth]{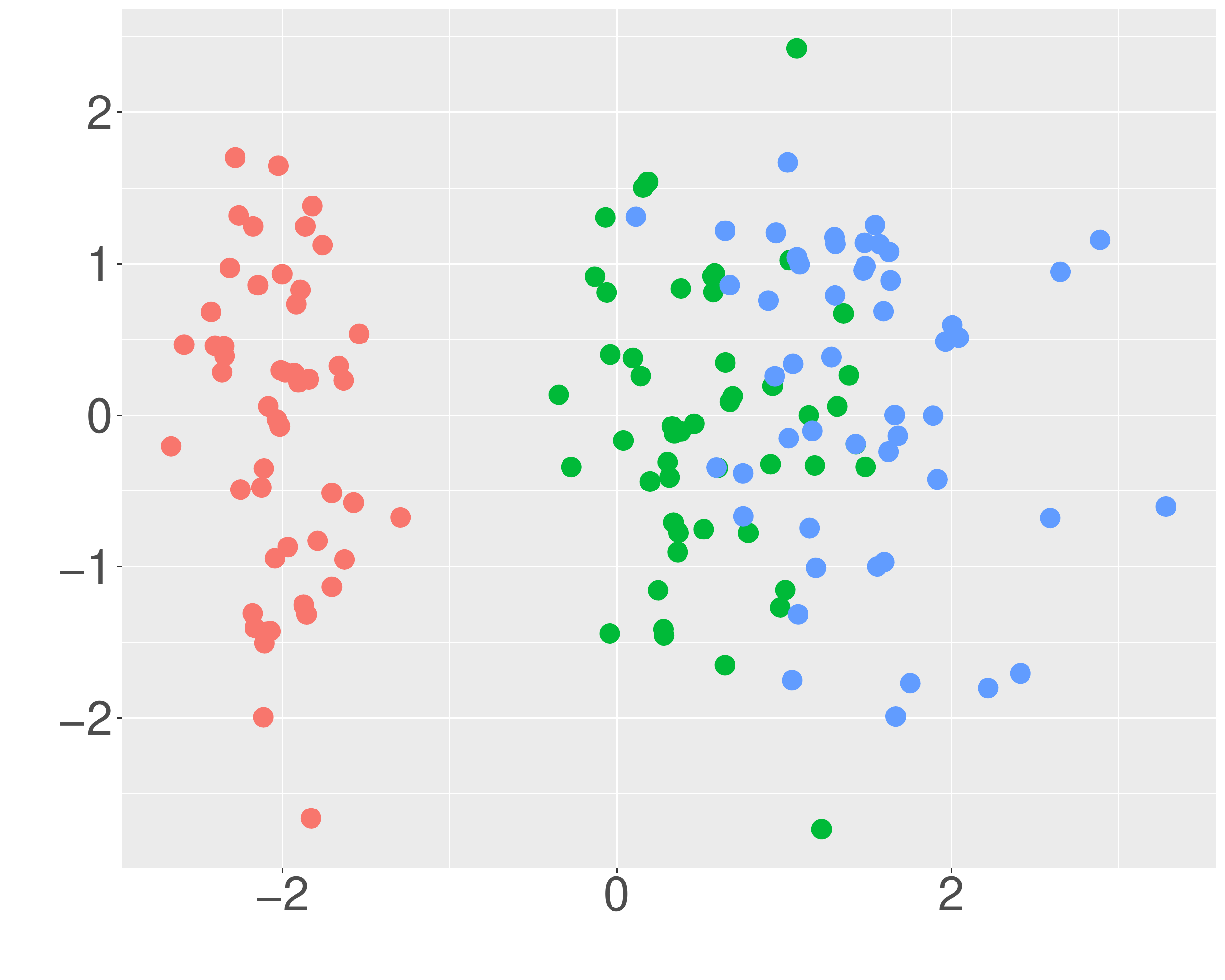}  
	\end{subfigure}
	\begin{subfigure}{.4\textwidth}
	  \centering
	  \includegraphics[width=\linewidth]{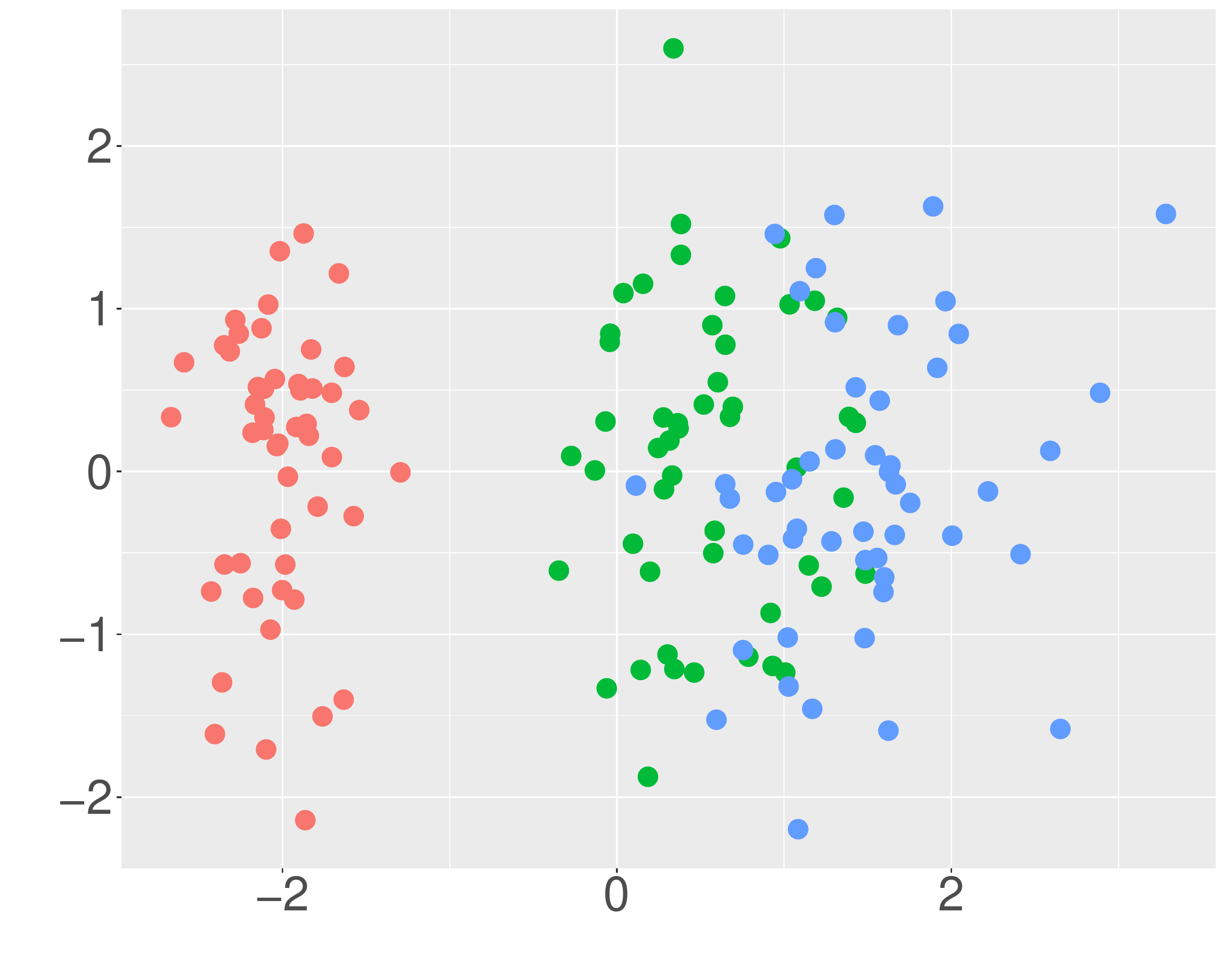}  
	\end{subfigure}

	\begin{subfigure}{.4\textwidth}
		\centering
		\includegraphics[width=\linewidth]{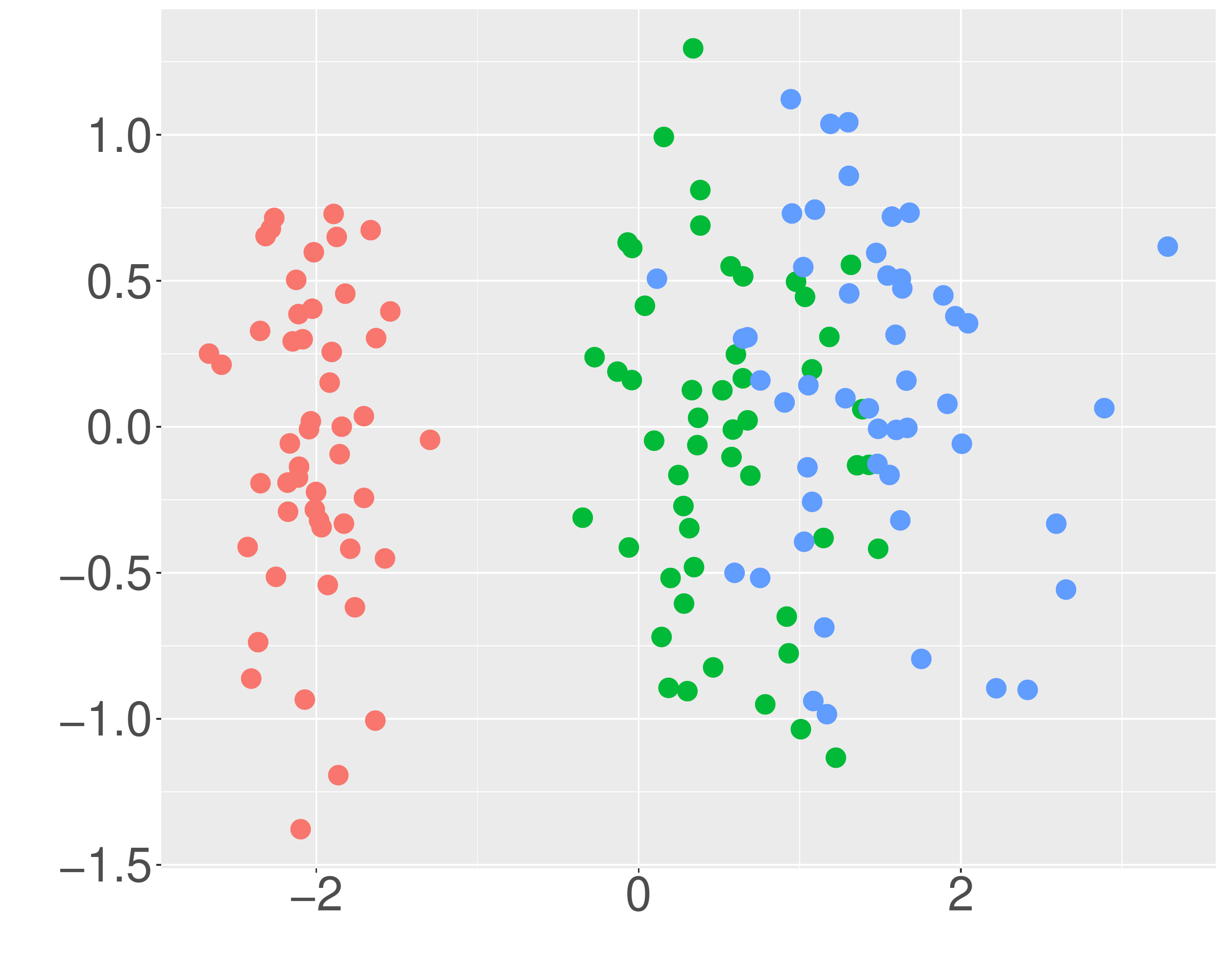}  
	  \end{subfigure} 
	  \begin{subfigure}{.4\textwidth}
		\centering
		\includegraphics[width=\linewidth]{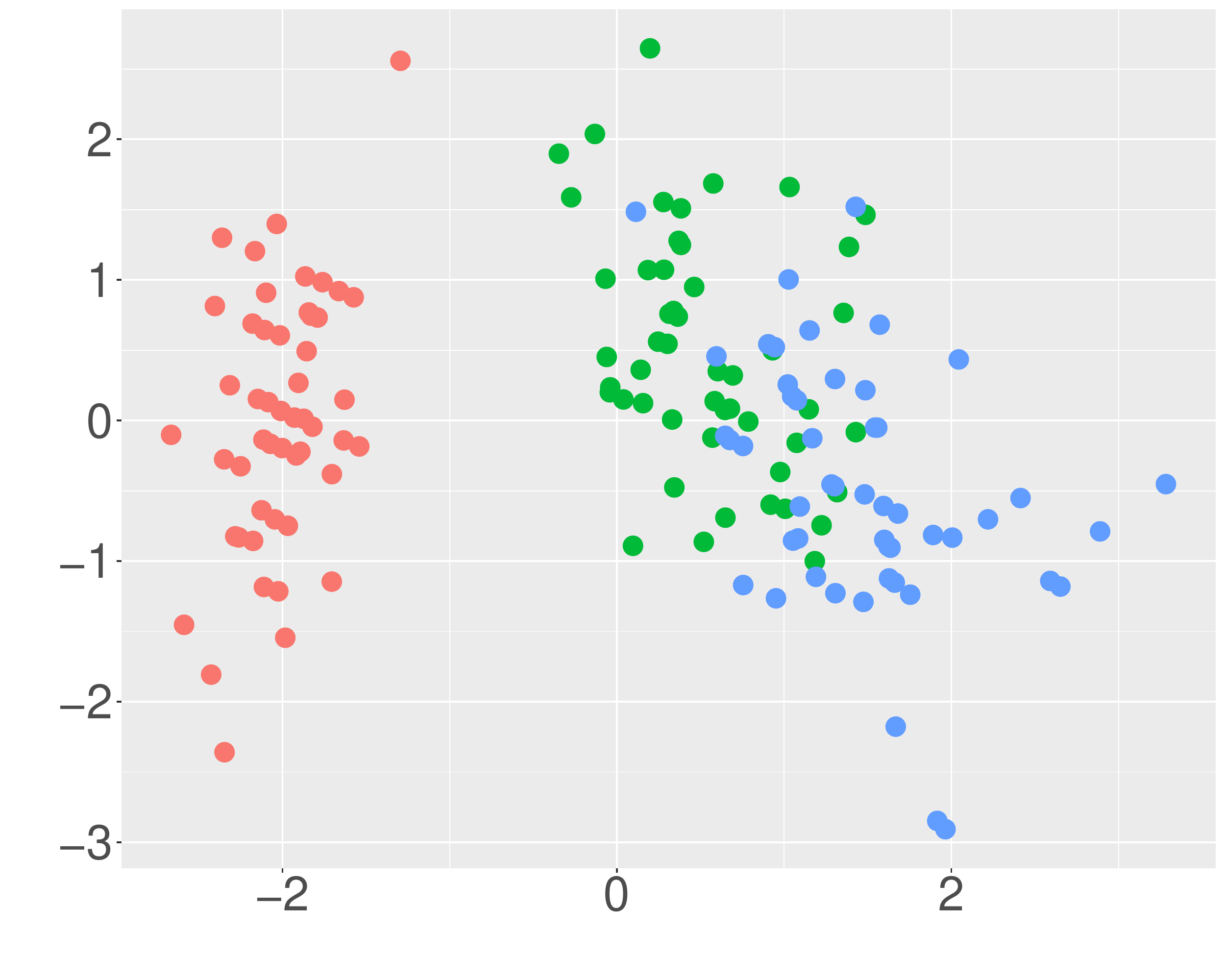}  
	  \end{subfigure}
	\caption{Scatter plots for $\mathbf P_{\mathcal S_1}\Y$ against different  $\mathbf P_{\mathcal S_3}\Y$ with different groups, where red, blue, green stands for setosa, versicolor and virginica.}
	\label{fig: scatter}
\end{figure}
\begin{figure}[!h]
	\centering
	\includegraphics[width=.4\linewidth]{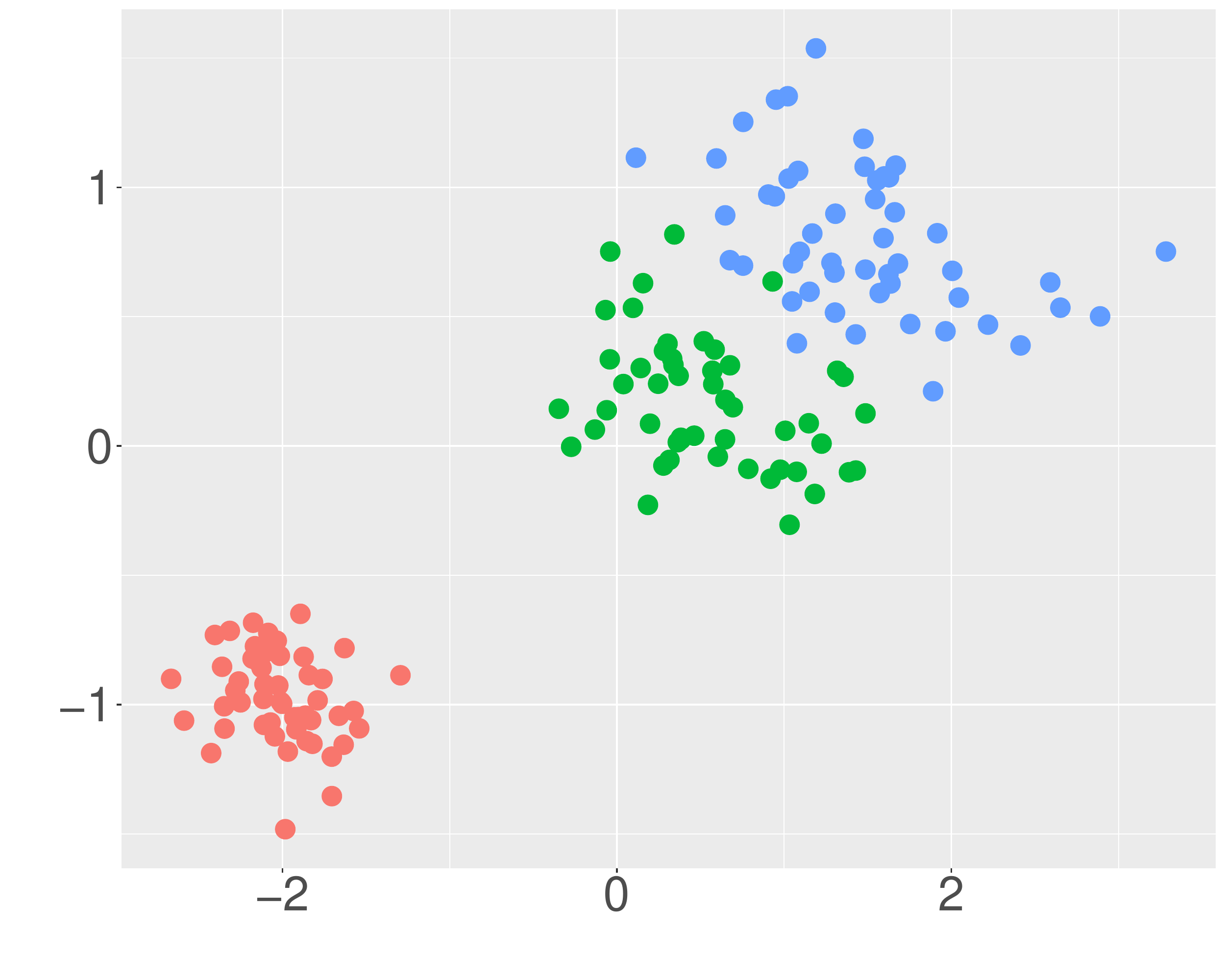}  
	\caption{Scatter plots for $\mathbf P_{\mathcal S_1}\Y$ against different  $\mathbf P_{\mathcal S_2}\Y$ with different groups, where red, blue, green stands for setosa, versicolor and virginica.}
	\label{fig: scatter2}
\end{figure}
\end{appendix}
\bibliographystyle{imsart-number} 
\bibliography{reference.bib}      
\end{document}